\documentclass[11pt]{article}

\usepackage[margin=1in]{geometry}
\usepackage{graphicx}
\usepackage{adjustbox}
\usepackage{caption}
\usepackage{subcaption}
\usepackage{bm}
\usepackage[table,xcdraw,dvipsnames]{xcolor}
\usepackage[maxbibnames=99,backend=bibtex,sorting=nyt,style=alphabetic,citestyle=alphabetic]{biblatex}
\usepackage{xurl}
\usepackage[colorlinks=true,linkcolor=BrickRed,citecolor=OliveGreen,urlcolor=SkyBlue]{hyperref}
\hypersetup{breaklinks=true}
\usepackage{amsmath, amsthm, amssymb, amstext}
\usepackage{mathtools}
\usepackage{algorithmic,algorithm}
\usepackage{dsfont}
\usepackage[most]{tcolorbox}
\usepackage{changepage}
\newtcolorbox{myframe}[1][]{
  enhanced,
  arc=0pt,
  outer arc=0pt,
  colback=white,
  boxrule=0.8pt,
  #1
}
\allowdisplaybreaks

\usepackage[sfmath]{kpfonts}
\usepackage{mathpazo}

\newcommand{\eps}{\epsilon}

\newcommand{\norm}[1]{\left\lVert #1 \right\rVert}
\DeclareMathOperator{\Tr}{Tr}
\newcommand{\ket}[1]{| #1 \rangle}
\newcommand{\bra}[1]{\langle #1 |}
\newcommand{\outerprod}[2]{| #1 \rangle \! \langle #2 |}

\DeclarePairedDelimiterX{\infdivx}[2]{(}{)}{%
  #1\;\delimsize\|\;#2%
}
\DeclareMathOperator{\Var}{{\textnormal{Var}}}

\newcommand{\expect}[1]{\underset{#1}{\expct}}
\newcommand\Tstrut{\rule{0pt}{2.6ex}}         % = `top' strut

\DeclareMathOperator{\expct}{{\mathbb E}}
\DeclareMathOperator{\rank}{rank}

\newtheorem{theorem}{Theorem}[section]
\newtheorem{lemma}[theorem]{Lemma}
\newtheorem{fact}[theorem]{Fact}
\newtheorem{corollary}[theorem]{Corollary}
\newtheorem{proposition}[theorem]{Proposition}

\theoremstyle{definition}
\newtheorem{definition}[theorem]{Definition}

\usepackage{upgreek}
\newcommand{\suppress}[1]{}

\newcommand{\complex}{{\mathbb C}}

\newcommand{\abs}[1]{\left| #1 \right|}

\newcommand{\complexi}{{\mathrm{i}}}

\newcommand{\unitary}{{\mathsf U}}

\newcommand{\qstate}{{\mathsf D}}
\DeclareMathOperator{\trace}{Tr}
\newcommand{\id}{{\mathbb 1}}

\newcommand{\ketbra}[2]{| #1 \rangle\!\langle #2 |}

\newcommand{\density}[1]{\ketbra{#1}{#1}}

\newcommand{\adjoint}{\dagger}
\newcommand{\eqdef}{\coloneqq}

\newcommand{\tensor}{\otimes}

\newcommand{\cA}{{\mathcal A}}
\newcommand{\cM}{{\mathcal M}}

\newcommand{\bU}{{\bm U}}

\newcommand{\tGamma}{{\widetilde{\Gamma}}}
\newcommand{\bpsi}{{\bm \psi}}
\newcommand{\normal}{{\mathrm N}}
\newcommand{\support}{\textrm{supp}}

\title{\vspace{-0.8cm}Lower bounds for learning quantum states with single-copy measurements}
\author{
   Angus Lowe~\thanks{Most of this work was completed while this author was with Department of Combinatorics \& Optimization and Institute for Quantum Computing, University of Waterloo, 200 University Ave.\ W., Waterloo, ON, N2L~3G1, Canada. Email: \texttt{alowe7@mit.edu}~.} 
   \and 
   Ashwin Nayak~\thanks{Department of Combinatorics \& Optimization and Institute for Quantum Computing, University of Waterloo, 200 University Ave.\ W., Waterloo, ON, N2L~3G1, Canada. Email: \texttt{ashwin.nayak@uwaterloo.ca}~.}}

\date{\today}

\begin{document}
\maketitle

\begin{abstract}
    We study the problems of quantum tomography and shadow tomography using measurements performed on individual, identical copies of an unknown~$d$-dimensional state. We first revisit known lower bounds~\cite{Haah_2017} on quantum tomography with accuracy $\epsilon$ in trace distance, when the measurement choices are independent of previously observed outcomes, i.e., they are nonadaptive. We give a succinct proof of these results through the $\chi^2$-divergence between suitable distributions. Unlike prior work, we work directly with measurements operators of arbitrary rank rather than with their decomposition into rank-one operators. This leads to stronger lower bounds when the learner uses measurements with a constant number of outcomes (e.g., two-outcome measurements). In particular, this rigorously establishes the optimality of the folklore ``Pauli tomography" algorithm in terms of its sample complexity. We also derive novel bounds of $\Omega(r^2 d/\epsilon^2)$ and $\Omega(r^2 d^2/\epsilon^2)$ for learning rank~$r$ states using arbitrary and constant-outcome measurements, respectively, in the nonadaptive case.

In addition to the sample complexity, a resource of practical significance for learning quantum states is the number of unique measurement settings required (i.e., the number of different measurements used by an algorithm, each possibly with an arbitrary number of outcomes). Motivated by this consideration, we employ concentration of measure of $\chi^2$-divergence of suitable distributions to extend our lower bounds to the case where the learner performs possibly adaptive measurements from a fixed set of~$\exp(O(d))$ possible measurements. This implies in particular that adaptivity does not give us any advantage using single-copy measurements that are efficiently implementable. We also obtain a similar bound in the case where the goal is to predict the expectation values of a given sequence of observables, a task known as shadow tomography. Finally, in the case of adaptive, single-copy measurements implementable with polynomial-size circuits, we prove that a straightforward strategy based on computing sample means of the given observables is optimal.
\end{abstract}

\section{Introduction}
\label{sec:overview}

\subsection{State tomography and its variants}
\label{sec-context}

In learning theory, an important resource is the number of samples of data used by the learner to correctly infer or predict their properties. The difficulty of a learning task, at first approximation, is therefore captured by its \textit{sample complexity}, defined to be the minimum number of samples required to solve the problem at hand with high probability. In this paper we consider the sample complexity of learning properties of an arbitrary unknown quantum state. Here, a sample amounts to preparing the state in some register, so that the number of samples is the number of identical copies of the state on which the learner can perform a measurement. For the most part, we focus on quantum state tomography, which is the fundamental task of estimating an unknown $d$-dimensional state $\rho$ to within some accuracy $\eps$ in the standard trace distance between states. Quantum tomography is of significant practical interest, for example, for the experimental verification of quantum devices. We are especially interested in how the sample complexity of tomography scales with the dimension $d$ of the state. In theory, the dimension is the primary obstacle to efficient learning, since this quantity grows exponentially with the number of qubits comprising the system.

In the most general scenario for state tomography, $n$ identical copies of a state~$\rho$ are prepared in registers that are jointly measured. It is then said that the measurements are \textit{entangled}. In a series of breakthroughs, O'Donnell and Wright~\cite{OW16-tomography,odonnell2017efficient2} as well as Haah, Harrow, Ji, Wu, and Yu~\cite{Haah_2017} proved that $O(d^2/\eps^2)$ samples suffice to perform tomography using entangled measurements. This matches an information-theoretic lower bound due to Ref.~\cite{Haah_2017} and improves upon previous upper bounds by a factor of $d$. Fewer samples are needed when a bound on the rank of the state is known (see, for instance, Ref.~\cite{OW16-tomography}).

From a practical standpoint, however, the joint measurements used in algorithms for optimal tomography may not be feasible. Firstly, in the case where one has access to just a single register that can be prepared in the state $\rho$, joint measurements of multiple copies of the state are impossible. (For instance, one might wish to perform tomography on the output state of a quantum computer by repeating a computation. Another example is that of photonic states that are difficult to store over extended periods of time.) Even given access to a suitably large system that can be prepared in the state $\rho^{\otimes n}$, it is not clear how efficiently the entangled measurements can be implemented. Finally, in some experimental realizations, only a limited set of measurements may be available. For these reasons, there is strong motivation to consider restricted measurement models, for instance, those in which each copy of~$\rho$ is measured separately, possibly using one of a fixed set of measurement settings. Measurements in which each copy of~$\rho$ is measured separately have been coined \textit{single-copy\/} measurements by some~\cite{Aaronson_2019,ALL22-distr-ip} (and \textit{unentangled\/} measurements by others~\cite{CD10-QMA2,wright2016, BCL20-mixedness-testing}). 
% or \textit{incoherent}~\cite{chen2022tightcertification, chen2022tighttomography} ).

Within the single-copy model of measurement, one has access only to a single $d$-dimensional register which can be repeatedly prepared in the state $\rho$ upon request, at which point a measurement is performed on the state and the resulting state is discarded. This means that the number of samples is equal to the number of measurements performed. Upper bounds on the sample complexity of single-copy tomography are well-established. Two prominent examples are the folklore ``Pauli tomography algorithm" (outlined in Section 8.4.2 in Nielsen and Chuang~\cite{nielsen_chuang_2010}) and algorithms based on low-rank matrix recovery due to Kueng, Rauhut, and Terstiege~\cite{KUENG2017}. In both examples, the upper bound on the sample complexity is worse than in the entangled case. (For other, simple such algorithms, see Refs.~\cite{wright2016,guta2018fast,yu2020sample}.)

What can be said about the sample complexity of quantum tomography using single-copy measurements? Haah \emph{et al.\/}~\cite{Haah_2017} address this question by providing an $\Omega(d^3/\eps^2)$ lower bound which matches the upper bound following from Ref.~\cite{KUENG2017,guta2018fast}, under the assumption that the choice of each measurement is independent of any previous outcomes (referred to as \textit{nonadaptive measurements}). However, this does not exhaust all realizable single-copy measurement strategies. Indeed, numerous proposals for state tomography (e.g., \cite{adaptivebayesian2012,mahler2013adaptive}) utilize \textit{adaptive\/} measurements, where the choice of measurement can depend on previous outcomes.

Adaptive measurements represent an intermediate restriction between nonadaptive and entangled measurements, and until very recently little was known about the sample complexity of learning quantum states or their properties in this setting. (For an early example of a problem for which adaptive measurements do not help, and for a separation between joint and single-copy measurements, see Refs.~\cite{HRS05-coset-states,HMRRS10-coset-states}.) This is despite the fact that bounding the power of adaptivity is a significant problem: proving separations between entangled and single-copy measurements requires showing that adaptive measurements result in strictly worse sample complexity. It was posed as an open problem in the Ph.D.\ thesis of Wright~\cite{wright2016} to provide examples where this is the case, and since then there has been significant progress on this topic. In 2020, Bubeck, Chen, and Li~\cite{BCL20-mixedness-testing} gave an unconditional separation between entangled and single-copy measurements for the problem of quantum state certification. Following this, Huang, Kueng, and Preskill~\cite{HKP21-it-bounds} proved an \textit{exponential} separation for the problem of determining the expectations of Pauli operators to constant accuracy. Then, in 2021 Chen, Cotler, Huang, and Li~\cite{CCHL21-single-copy-measurements, chen2021hierarchy} proved many additional exponential separations for different learning tasks, including shadow tomography. In this work, we continue along this line of research to investigate the sample complexity of adaptive quantum tomography in a realistic setting. We then apply the techniques developed and derive a new lower bound for single-copy shadow tomography in the same setting.

\subsection{Summary of results}
\label{sec-summary}

We first provide a simplified proof of the lower bound for tomography in the nonadaptive case due to Haah \emph{et al.\/}~\cite[Theorem~4]{Haah_2017}. In the process, we improve it by a factor of $d$ to $\Omega(d^4/\eps^2)$ when the measurements have a constant number of outcomes. This implies that the straightforward Pauli tomography algorithm (described in Appendix~\ref{sec:upper_bounds_tomography}) is information-theoretically optimal in this setting. Using the same techniques, we derive a lower bound of~$\Omega(r^2 d / \eps^2)$ when the states are known to have bounded rank~$r$. This bound is a multiplicative factor of~$\log(1/\eps)$ larger than the best previous lower bound~\cite[Theorem~4]{Haah_2017}, and is optimal~\cite{KUENG2017,guta2018fast} (see Section~\ref{sec:single_copy_tomog_upper_bound} for the details of how the upper bound follows from Ref.~~\cite{KUENG2017}.). Moreover, it applies to the case of learning pure states~($r = 1$), which is not covered by the proof of Theorem~4 in Ref.~\cite{Haah_2017}. The rank-dependent bound can be further strengthened to~$\Omega(r^2 d^2 / \eps^2)$ for measurements with a constant number of outcomes. 

Since state tomography requires~$\Omega(d^2/ \eps^2)$ samples, any quantum algorithm for this problem necessarily has run-time at least quadratic in~$d$. This is exponential in~$\log d$, the number of qubits representing the unknown state. However, algorithms that measure one copy of the state at a time, interleaved with classical processing of the measurement outcomes, allow for the possibility that the \emph{individual\/} measurements be more time-efficient. Such algorithms are more attractive from a practical point of view, given the current challenges in implementing quantum computation. It is thus no surprise that most of the algorithms based on single-copy measurements mentioned in Section~\ref{sec-context} involve measurements that can be implemented efficiently, in particular with quantum circuits of size \emph{polynomial\/} in the number of qubits. 

We present new arguments showing there is a broad class of algorithms, including the ones described above, for which adaptivity \textit{makes no difference} to the worst-case sample complexity of learning a quantum state. Specifically, we prove a lower bound of~$\Omega(d^3/\eps^2)$ for the sample complexity of any single-copy, adaptive tomography algorithm which uses measurements chosen from a fixed set of up to $\exp(O(d))$ measurements. This encompasses measurement strategies which are efficiently implementable, i.e., the measurements may be performed using (uniformly generated) circuits of size polynomial in $\log d$ over some finite universal gate-set. We also show using the Solovay-Kitaev Theorem that, up to a factor of roughly $\log\log d + \log(1/\eps)$, the same bound applies to all measurement strategies which are efficiently implementable using circuits on possibly infinite universal gate-sets. The bounds entail that either {(i)} adaptivity does not give any advantage over non-adaptive measurements for single-copy tomography, or {(ii)} any adaptive algorithm using~$o(d^3/\eps^2)$ samples necessarily uses measurements with super-polynomial-size circuits. We summarize lower bounds for single-copy tomography in comparison to previous work in Table~\ref{tab:tomog_bounds_summary}, in the full-rank case. In the final column, by ``efficient" we mean efficiently implementable, as defined above.

\begin{table}[t]
\begin{center}
\begin{adjustbox}{width=\columnwidth,center}
\begin{tabular}{ c | c | c | c | c | c}
 & \multicolumn{2}{c|}{Nonadaptive} &  Adaptive & \multicolumn{2}{c}{Adaptive \& efficient} \\
\hline
Allowed meas. & $O(1)$-outcome  & Arbitrary & Binary Pauli & $O(1)$-outcome & Arbitrary\\
\hline
Upper bound\Tstrut & $O(d^4/\eps^2)$ & $O(d^3/\eps^2)$~\cite{KUENG2017,guta2018fast} 
& $O(d^4/\eps^2)$ & $O(d^4/\eps^2)$ & $O(d^3/\eps^2)$\\
Lower bound & $\Omega(d^4/\eps^2)$ [{Cor.~\ref{cor:indep_lbs_const_outcome}}] & $\Omega(d^3/\eps^2)$~\cite{Haah_2017} & $\Omega(d^4)$~\cite{Flammia_2012} & $\widetilde{\Omega}(d^4/\eps^2)$ {[Thm.~\ref{thm:tomography_efficient_meas_const_outcome}]} & $\widetilde{\Omega}(d^3/\eps^2)$ {[Thm.~\ref{thm:tomography_efficient_meas}]}
\end{tabular}
\end{adjustbox}
\end{center}
 \caption{\label{tab:tomog_bounds_summary}Best known upper and lower bounds for the sample complexity of quantum state tomography using single-copy measurements under various measurement restrictions, without a known bound on the rank of the state. $\widetilde{\Omega}$ hides $\log(d)$ and $\text{polylog}(1/\eps)$ factors while lack of citation indicates folklore or implied by other bounds.}
\end{table}

We also obtain lower bounds of the above kind for computing classical shadows~\cite{Huang2020} and for shadow tomography~\cite{Aaronson2020shadow}. In these tasks, one is interested in estimating the expectations of some collection of observables, and they have practical applications ranging from entanglement verification to near-term proposals of variational quantum algorithms~\cite{Huang2020,struchalin_2021}. We show that any procedure for $\eps$-accurate shadow tomography of $M$ observables using efficiently implementable single-copy measurements requires $\Omega(d\log(M)/\eps^2)$ samples of the unknown $d$-dimensional quantum state. Recently, building upon techniques developed in {Ref.~\cite{Aharonov2022}}, Ref.~\cite{CCHL21-single-copy-measurements} almost fully resolved the sample complexity of shadow tomography in the more general case where the learner can implement arbitrary single-copy measurements. They showed a lower bound of $\widetilde{\Omega}(\min\{M,d\}/\eps^2)$. This, while being more general than our result, is potentially exponentially looser in the setting of efficient measurements. In particular, even for $M$ a small constant, our lower bound is linear in the dimension of the state, whereas the more general lower bound has no dependence on the dimension at all.

Besides classical shadows, it would be interesting to investigate whether our lower bound techniques may be relevant in other settings, e.g., in identity testing with single-copy measurements, previously considered {in~\cite{yu2021sampleefficientidentity}}. We leave this as an interesting direction for future work.

Finally, we present a simple procedure for classical shadows using single-copy measurements that are efficiently implementable. The algorithm is optimal in this setting as well as in the case where the measurements are nonadaptive but otherwise arbitrary. The procedure is simpler than the one given in Ref.~\cite{Huang2020}.

\paragraph{Related work}
Most of the results in this article were included in the first author's Master's {thesis~\cite{Lowe21-tomography}} and were presented at {QIP~2022~\cite{LN22-tomography}}. Chen, Huang, Li, and {Liu~\cite{chen2022tightcertification}} subsequently proved that known non-adaptive algorithms for state certification are optimal even when adaptive measurements are used. Shortly thereafter, the same set of authors along with {Sellke~\cite{chen2023doesadaptivityhelpquantum}} reported {an~$\Omega(d^3 / \eps^2)$} lower bound on the sample complexity of tomography of states of possibly full rank, using adaptive single-copy measurements. These bounds imply that adaptivity does not give any advantage over non-adaptive measurements in terms of sample complexity for state tomography (as a function of the dimension). On the other hand, the same work, concurrently with~\cite{Flammia_2024}, shows that when considering a metric of success besides the trace distance, it is possible to gain an advantage in copy complexity through adaptive measurements.

Interestingly, in some settings besides full-state tomography it is possible to show that neither adaptivity nor entanglement provide any advantage over single-copy, non-adaptive measurements~\cite{yu2020sample,Anshu_2022}.

\subsection{Overview of techniques}

We first describe a basic framework for proving lower bounds on the task of quantum tomography common to much of the work on the topic. Here, we use the observation that state discrimination of well-separated states reduces to tomography with sufficient accuracy. The lower bounds then follow from the construction of difficult instances of the state discrimination problem, for which the amount of information that the measurement statistics can reveal about the chosen state is severely limited. ``Discretizing" the learning problem in this manner for the purposes of providing worst-case lower bounds is a standard technique in the field of density estimation, which is the classical analogue of quantum tomography. (See for example Chapter 2 of Ref.~\cite{Tsybakov2009}.) To the best of the authors' knowledge, the method was first employed in the context of tomography by Flammia, Gross, Liu, and Eisert~\cite{Flammia_2012}.
% to derive a $\Omega(d^4/\log(d))$ lower bound when one is restricted to using adaptive binary Pauli measurements.

One way to make this argument rigorous is by using Fano's inequality and Holevo's theorem, which suggests an interpretation in terms of a communication protocol between two parties, Alice and Bob. To this end, imagine they have agreed upon an encoding of $2^N$ quantum states into bit-strings $x$ of length $N$. In a single round of communication, Alice sends a quantum state $\rho_x^{\otimes n}$ encoding the message $x\in \{0,1\}^{N}$ to Bob who then attempts to decode the message through tomography. Assuming Bob can perform accurate tomography using $n$ copies of the unknown state, Alice will have successfully transmitted $N$ bits of information to Bob. On the other hand, the Holevo information of the ensemble of quantum states gives an upper bound on the size of a message that could be sent reliably. In particular, it can be shown that when $n$ is small the Holevo information is also small. This provides the necessary contradiction to arrive at a lower bound: a procedure for tomography that succeeds when $n$ is small could be used by Bob to reliably decode too long a message from Alice. Therefore, there is no such procedure.

In summary, this argument may be used to show that the mutual information between the random choice of state $\bm{x}$ and the measurement outcome $\bm{y}$ satisfies $\Omega(d^2)\leq I(\bm{x}:\bm{y})\leq n\eps^2$, where the first inequality comes from Fano's inequality, and the second from Holevo's theorem. However, using Holevo's theorem in this manner does not take into account restrictions on the measurements we are allowed to perform on the $n$ copies of the state. One might therefore expect that it be possible to derive a tighter bound on the mutual information by exploiting the fact that the measurements are not entangled. It turns out that this is indeed the case, as demonstrated by the $\Omega(d^3/\eps^2)$ lower bound for nonadaptive measurements due to Ref.~\cite{Haah_2017}.

Our approach differs from previous work in making direct use of a connection between the mutual information of two random variables and the $\chi^2$-divergence of related distributions, as well as techniques for Haar integration based on symmetry. Additionally, we do not require that the measurements be rank-one POVMs as in Ref.~\cite{Haah_2017}; this allows us to conclude the $\Omega(d^4/\eps^2)$ lower bound in the constant-outcome, nonadaptive case, as well as the more precise bounds stated in Section~\ref{sec-summary} for states of bounded rank. We further build on these simplifications to derive lower bounds robust to a wide class of adaptive measurements. We accomplish this by adversarially constructing instances of the state discrimination problem that are as difficult as possible for the specific set of measurements under consideration. This involves making use of well-known concentration of measure results for the unitary group. This idea is reminiscent of the lower bounds for tomography restricted to binary Pauli measurements due to Flammia \emph{et al.\/}~\cite{Flammia_2012}. A key technical step is the analysis of $\chi^2$-divergence rather than the probability of individual measurement outcomes. This enables tight lower bounds agnostic to the measurements we consider.

\section{Preliminaries}\label{sec:preliminaries}
\subsection{Mathematical background}
This section contains relevant notation and properties that may be referred to as needed.

\paragraph{Sets.}
We let $\mathbb{Z}_+$ denote the set of nonnegative integers, $\mathsf{U}(d)$ the set of unitary operators acting on $\mathbb{C}^d$, $\mathsf{H}(d)$ the set of Hermitian operators acting on $\mathbb{C}^d$, $\mathsf{Psd}(d)$ the subset of $\mathsf{H}(d)$ consisting of positive semidefinite operators, and $\mathsf{D}(d)$ the subset of operators in $\mathsf{Psd}(d)$ with unit trace (i.e., the set of $d$-dimensional quantum states). We also denote by $\mathsf{L}(d)$ the set of square operators acting on $\mathbb{C}^d$.
\paragraph{Operators.}
For any square operator $A\in\mathsf{L}(d)$, we denote its adjoint by $A^\dag$. We let $\norm{A}_1 = \Tr(\sqrt{A^\dag A})$ denote the ``trace norm" of the operator $A$ and $\norm{A}_{\mathrm{F}} = \sqrt{\Tr(A^\dag A)}$ it Frobenius norm. The \textit{trace distance} between two quantum states is $\norm{\rho-\sigma}_1$.  We use $\norm{A}$ to denote the spectral norm of the operator $A$; this is the operator norm induced by the Euclidean norm on $\mathbb{C}^d$. We have the useful relations $\norm{A}_{\mathrm{F}}\leq \norm{A}_1\leq \sqrt{d}\norm{A}_{\mathrm{F}}$ and $\norm{AB}_{\mathrm{F}}\leq \norm{A}\norm{B}_{\mathrm{F}}$. For any two operators $P,Q\in\mathsf{Psd}(d)$, we use the notation $P\preceq Q$ if and only if $Q-P\in \mathsf{Psd}(d)$. Let $A,B\in \mathsf{H}(d)$ and consider the operator $A\otimes B$. We denote by $\Tr_2(\cdot)$ the partial trace over the second system, i.e., $\Tr_2(A\otimes B)=A\Tr(B)$. The rank of a linear operator $X$, denoted $\textnormal{rank}(X)$, is the dimension of its image, which we denote by $\textnormal{im}(X)$.

\paragraph{Permutation operator and $t$-designs.}
The swap operator $W$ acting on $(\mathbb{C}^{d})^{\otimes 2}$ is the linear operator defined by the action $W \ket{\psi}\otimes \ket{\phi} =\ket{\phi}\otimes \ket{\psi}$ for any two vectors $\ket{\psi},\ket{\phi}\in \mathbb{C}^d$. We may extend this procedure to arbitrary permutations, defining the linear operator $W_\pi$ for each $\pi\in S_n$ and acting on $(\mathbb{C}^{d})^{\otimes n}$ as
\begin{align*}
W_{\pi} \; \ket{x_1}\otimes \dots \otimes \ket{x_n} = \ket{x_{\pi^{-1}(1)}}\otimes\dots\otimes \ket{x_{\pi^{-1}(n)}}
\end{align*}
for every choice of vectors $\ket{x_1},\dots,\ket{x_n}\in\mathbb{C}^d$. Here, $S_n$ denotes the symmetric group on $\{1,\dots,n\}$.

We make use of unitary and state $t$-designs throughout this paper.
\begin{definition}[Unitary $t$-design]
For positive integers $t,d>0$ we say that a random unitary operator $\bm{V}\in\mathsf{U}(d)$ is a \emph{unitary $t$-design} if the the following holds for every operator $X \in \mathsf{L}(d)^{\otimes t}$:
\begin{align*}
\expct \bm{V}^{\otimes t} X (\bm{V}^\dag)^{\otimes t} = \int_{\mathsf{U}(d)} U^{\otimes t} X (U^\dag)^{\otimes t}\mathrm{d}\mu(U)
\end{align*}
where $\mu$ is the Haar measure on the space of $d$-dimensional unitary operators.
\end{definition}
\begin{definition}[State $t$-design]
For positive integers $t,d > 0$, a \emph{state $t$-design} is a random quantum state $\ket{\bm{u}}\in \mathsf{S}(d)$ which satisfies
\begin{align}
\expct (\outerprod{\bm{u}}{\bm{u}})^{\otimes t} = \int_{\mathsf{S}(d)} (\outerprod{v}{v})^{\otimes t}\ \mathrm{d}\mu(v)
\end{align}
where $\mathsf{S}(d)$ is the set of unit vectors in $\mathbb{C}^d$.
\end{definition}

\paragraph{Random variables.}
We denote random variables using bold font, including matrix-valued random variables. We use lowercase (e.g., $p,q$) with appropriate subscripts to denote the distributions of random variables. For example, suppose $\bm{x}$ is a random variable taking values in $\mathcal{X}$ according to some distribution $p_{\bm{x}} : \mathcal{A}\to [0,1]$, where $\mathcal{A}$ is the set of Borel-measurable subsets of $\mathcal{X}$. Let $\mathcal{S}$ be some finite-dimensional vector space, and let $f:\mathcal{X}\to \mathcal{S}$. Then we write interchangeably $\expct_{\bm{x}}f(\bm{x})$ and $\expct_{\bm{x}\sim p_{\bm{x}}}f(\bm{x})$ to refer to the expectation of $f$ with respect to the distribution $p_{\bm{x}}$ (i.e., $\int_{\mathcal{X}}f(x)\mathrm{d}p_{\bm{x}}(x)$) using the latter notation when there may be some ambiguity about what the distribution is. When it is clear enough from context, we drop the subscripts altogether and write $\expct f(\bm{x})$. In the case where $\bm{x}$ is a discrete random variable taking values in some finite set (or alphabet) $\mathcal{X}$, we write its probability mass function (PMF) as $p_{\bm{x}}$, and corresponding expectations $\expct_{\bm{x}\sim p_{\bm{x}}}f(\bm{x})=\sum_{x\in\mathcal{X}}p_{\bm{x}}(x)f(x)$. We also refer to $p_{\bm{x}}$ as the distribution of $\bm{x}$ in this case.
Next suppose we have random variables $(\bm{x},\bm{y})$ jointly distributed on $\mathcal{X}\times\mathcal{Y}$. If $\bm{y}$ is discrete, we write $p_{\bm{y}|x}(y)$ to mean the probability that $\bm{y}=y$ given $\bm{x}=x$, when it is well-defined. We will often have occasion to use functionals $F$ mapping distributions to the reals. Then if $\bm{x}$ has marginal distribution given by $p_{\bm{x}}$, we write $\expct_{\bm{x}^\prime\sim p_{\bm{x}}}F(p_{\bm{y}|\bm{x}^\prime})$ to denote the expectation $\int_{\mathcal{X}}F(p_{\bm{y}|x})\mathrm{d}p_{\bm{x}}(x)$. Finally, we sometimes use in the subscripts of expectations the notation $\bm{x}|y$ to mean the random variable $\bm{x}$ conditioned on $\bm{y}=y$, when it is well-defined. For example, suppose we have a function $g:\mathcal{X}\times\mathcal{Y}\to\mathbb{R}$. It holds by definition that $\expct_{\bm{x}}\expct_{\bm{y}|\bm{x}}g(\bm{x},\bm{y})=\expct_{\bm{x},\bm{y}}g(\bm{x},\bm{y})=\expct_{\bm{y}}\expct_{\bm{x}|\bm{y}}g(\bm{x},\bm{y})$.

\paragraph{Information theory.}
Consider discrete random variables taking values on the same space. One may then use the KL-divergence between their distributions to compare them. The KL-divergence between two discrete distributions (PMFs) $p$, $q : \mathcal{X}\to [0,1]$ defined on the same sample space $\mathcal{X}$ is
\begin{align*}
\mathrm{D}_{\textnormal{KL}}\infdivx{p}{q} 
= \begin{cases}
\sum_{x\in\mathcal{X}}p(x)\log\left(\frac{p(x)}{q(x)}\right),& \text{supp}(p)\subseteq \text{supp}(q)\\
+\infty,& \text{otherwise}
\end{cases}
\end{align*} 
where we take $0 \log(0) = 0$. (Throughout this work, $\log$ denotes the logarithm with base~$2$.)

We next define some entropic quantities. Let $\bm{x}$ be a discrete random variable taking values in $\mathcal{X}$ with distribution $p_{\bm{x}}:\mathcal{X}\to[0,1]$. The Shannon entropy measures our uncertainty about $\bm{x}$ and is defined as
\begin{align*}
H(\bm{x}) = -\sum_{x\in\mathcal{X}} p_{\bm{x}}(x)\log(p_{\bm{x}}(x)).
\end{align*}
We also write $H(p_{\bm{x}})$ to refer to the same quantity. A useful property of the entropy is \textit{concavity}, whereby for any two discrete distributions $p$, $q$ defined on the same sample space and $\lambda\in[0,1]$ it holds that 
\begin{align*}
H(\lambda p+(1-\lambda)q)\geq \lambda H(p) + (1-\lambda)H(q).
\end{align*}
Next, let $\bm{\bm{y}}$ be a different discrete random variable taking values in $\mathcal{Y}$, so that $\bm{x}$ and $\bm{y}$ have joint distribution given by $p_{\bm{x},\bm{y}}:\mathcal{X}\times \mathcal{Y}\to [0,1]$. The joint entropy of these random variables is
\begin{align*}
H(\bm{x},\bm{y}) = -\sum_{x\in\mathcal{X}}\sum_{y\in\mathcal{Y}}p_{\bm{x},\bm{y}}(x,y)\log(p_{\bm{x},\bm{y}}(x,y))
\end{align*}
and the conditional entropy of $\bm{x}$ given $\bm{y}$ is
\begin{align*}
H(\bm{x}|\bm{y}) = H(\bm{x},\bm{y}) - H(\bm{y}).
\end{align*}
These definitions are valid only in the case where $\bm{x}$ and $\bm{y}$ are discrete. Mutual information, on the other hand, is well-defined for arbitrary random variables $\bm{x}$, $\bm{y}$ though for our purposes it will suffice to define this quantity in the following way, which is valid when $\bm{y}$ is discrete.
\begin{definition}[Mutual information]\label{fact:mut_inf_to_kl_relationship}
Consider two random variables $\bm{x}$ and $\bm{y}$ such that $\bm{y}$ is discrete. Let $p_{\bm{y}|x}$ be the conditional distribution of $\bm{y}$ given $\bm{x}=x$, $p_{\bm{x}}$ the marginal distribution of $\bm{x}$, and $p_{\bm{y}}$ the marginal distribution of $\bm{y}$. The \textit{mutual information} between $\bm{x}$ and $\bm{y}$ is
\begin{align*}
I(\bm{x}:\bm{y}) \coloneqq \expect{\bm{x}\sim p_{\bm{x}}} D_{\textnormal{KL}}\infdivx{p_{\bm{y}|\bm{x}}}{p_{\bm{y}}}.
\end{align*}
\end{definition}
As the name suggests, the mutual information between two random variables quantifies the shared information between them. Since this definition is somewhat non-standard, it is worth taking the time to see how it reduces to the more standard definitions in familiar settings. Firstly, it may be shown that the above is equal to
\begin{align*}
I(\bm{x}:\bm{y}) = H(\bm{y}) - \expect{\bm{x}^\prime\sim p_{\bm{x}}}H(\bm{y}|\bm{x}=\bm{x}^\prime)
\end{align*}
where $\bm{y}|\bm{x}=x$ is the random variable $\bm{y}$ conditioned on the event $\bm{x}=x$. Then, if $\bm{x}$ is also discrete, it holds that $H(\bm{y}|\bm{x})=\expct_{\bm{x}^\prime\sim p_{\bm{x}}}H(\bm{y}|\bm{x}=\bm{x}^\prime)$ in which case we arrive at the commonly used expression for the mutual information,
\begin{align*}
I(\bm{x}:\bm{y}) = H(\bm{y}) - H(\bm{y}|\bm{x}) = H(\bm{x})-H(\bm{x}|\bm{y}) = H(\bm{x})+H(\bm{y})-H(\bm{x},\bm{y}).
\end{align*}
Next, suppose $\bm{z}$ is another random variable jointly distributed with $\bm{x}$ and $\bm{y}$. When $\bm{z}$ has a fixed value $z$, we use the notation
\begin{align*}
I(\bm{x}:\bm{y}|\bm{z}=z)\coloneqq I(\ (\bm{x}|\bm{z}=z)\ :\ (\bm{y}|\bm{z}=z)\ )
\end{align*}
where $(\bm{x}|\bm{z}=z)$ is $\bm{x}$ conditioned on $\bm{z}=z$, and likewise for $(\bm{y}|\bm{z}=z)$. The conditional mutual information between $\bm{x}$ and $\bm{y}$ given $\bm{z}$ is then defined as
\begin{align*}
I(\bm{x}:\bm{y}|\bm{z}) \coloneqq \expect{\bm{z}^\prime\sim p_{\bm{z}}}\ I(\bm{x}:\bm{y}|\bm{z}=\bm{z}^\prime).
\end{align*}
We now present three exceedingly useful facts about mutual information. We will use these to derive stronger lower bounds on tomography than the ones obtained by applying Holevo's theorem in the case where there is some restriction on the measurements.
%\begin{fact}\label{fact:conditioning_sometimes_reduces_mut_inf}
%Let $\bm{x}$, $\bm{y}$, and $\bm{z}$ be random variables, and suppose that $\bm{y}$ and $\bm{z}$ are independent given $\bm{x}$. Then
%\begin{align*}
%I(\bm{x}:\bm{y}|\bm{z}) \leq I(\bm{x}:\bm{y}).
%\end{align*}
%\end{fact}
\begin{fact}[Chain rule for mutual information]\label{fact:chain_rule_for_mut_inf}
It holds that
\begin{align*}
I(\bm{x}:\bm{y}_1,\dots,\bm{y}_n)= \sum_{i=1}^n I(\bm{x}:\bm{y}_{i}|\bm{y}_{i-1},\dots,\bm{y}_1).
\end{align*}
\end{fact}

\begin{corollary}[Subadditivity of mutual information]\label{cor:subadditivity_of_mut_inf}
If $\bm{y}_1,\dots,\bm{y}_n$ are independent given $\bm{x}$, it holds that
\begin{align*}
I(\bm{x}:\bm{y}_1,\dots,\bm{y}_n)\leq \sum_{i=1}^n I(\bm{x}:\bm{y}_i).
\end{align*}
\end{corollary}
The random variables $\bm{x}$, $\bm{y}$, $\bm{z}$ form a \textit{Markov chain\/} $\bm{x}\to\bm{y}\to\bm{z}$ if given $\bm{y}$, the random variables~$\bm{x}$ and $\bm{z}$ are independent (Ref.~\cite{Cover2005}, Section 2.8). Under this assumption, the following lemma holds, which is indispensable toward proving information-theoretic lower bounds on estimation tasks.
\begin{lemma}[Fano's inequality~\cite{fano1966}]\label{fact:fanos_ineq}
Let $\bm{x}$, $\bm{y}$, $\hat{\bm{x}}$ be discrete random variables forming a Markov chain $\bm{x}\to \bm{y}\to\hat{\bm{x}}$, where $\bm{x}$ takes values in $\mathcal{X}$. It holds that
\begin{align*}
H(p_e)+p_e\log(|\mathcal{X}|) \geq H(\bm{x}|\bm{y}).
\end{align*}
where $p_e\coloneqq \Pr[\bm{x}\neq\hat{\bm{x}}]$, and $H(\cdot)$ is the binary entropy function.
\end{lemma}

\begin{corollary}\label{cor:fanos_ineq}
Let $\bm{x},\bm{y},\hat{\bm{x}}$ be discrete random variables forming a Markov chain $\bm{x}\to\bm{y}\to\hat{\bm{x}}$. Suppose Alice has a message $\bm{x}\sim\textnormal{Unif}([N])$ and Bob is able to decode the message with constant probability of success using $\hat{\bm{x}}$. It must hold that
\begin{align*}
I(\bm{x}:\bm{y}) \in \Omega(\log(N)).
\end{align*}
\end{corollary}
\begin{proof}
Using the definition of mutual information we have $I(\bm{x}:\bm{y})=H(\bm{x})-H(\bm{x}|\bm{y})$. Let $p_e$ be as in Lemma~\ref{fact:fanos_ineq}. By Lemma~\ref{fact:fanos_ineq} we have $I(\bm{x}:\bm{y})\geq H(\bm{x})-p_e\log(N)-H(p_e)$. Using the fact that $H(\bm{x})=\log(N)$ and $H(p_e)\leq 1$ we obtain $I(\bm{x}:\bm{y})\geq (1-p_e)\log(N)-1$. 
\end{proof}
Besides the KL-divergence, another way to compare distributions defined on the same space is the following.
\begin{definition}[$\chi^2$-divergence]\label{def:chi_squared}
The $\chi^2$-divergence between two discrete distributions $p,q:\mathcal{X}\to[0,1]$ defined on the same sample space $\mathcal{X}$ is
\begin{align*}
\mathrm{D}_{\chi^2}\infdivx{p}{q} = \sum_{x\in\mathcal{X}}q(x)\left(\frac{p(x)}{q(x)}-1\right)^2 =  \sum_{x\in\mathcal{X}}q(x)\left(\frac{p(x)}{q(x)}\right)^2-1
\end{align*}
%if $p$ is absolutely continuous with respect to $q$, and is equal to $+\infty$ otherwise.
if $\support(p) \subseteq \support(q)$, and is equal to $+\infty$ otherwise.
\end{definition}
These divergences are related in the following way.
\begin{lemma}[KL vs.\ $\chi^2$ inequality]\label{lem:chi_squared_vs_kl_inequality}
Let $p,q:\mathcal{X}\to[0,1]$ be discrete distributions defined on the same sample space $\mathcal{X}$. We have
\begin{align*}
\mathrm{D}_{\textnormal{KL}}\infdivx{p}{q}\leq \frac{1}{\ln(2)}\cdot \mathrm{D}_{\chi^2}\infdivx{p}{q}.
\end{align*}
\end{lemma}
\begin{proof}
By Eq.~(5) in Ref.~\cite{sason_2016}, we have the inequality $\mathrm{D}_{\textnormal{KL}}\infdivx{p}{q}\leq \log(1+\mathrm{D}_{\chi^2}\infdivx{p}{q})$ from which the lemma follows by the inequality $\log(1+x)\leq x/\ln(2)\ \forall x\geq 0$. For an exposition of the many other relationships between divergences, we refer the interested reader to Ref.~\cite{sason_2016}. 
\end{proof}

\subsection{Single-copy measurements}\label{sec:measurement_models}
In general, an $m$-outcome \textit{measurement} of a $d$-dimensional quantum state is a linear map $\mathcal{M}:\mathsf{D}(d)\to\mathsf{L}(m)$ acting on quantum states $\rho\in\mathsf{D}(d)$ by
\begin{align*}
\mathcal{M}: \rho \mapsto \sum_{z\in\mathcal{Z}} \Tr(M_z\rho)\ket{z}\bra{z}
\end{align*}
for some ``positive operator-valued measure'' (POVM) $( M_z:z\in\mathcal{Z}, \; M_z \in \mathsf{Psd}(d) )$ satisfying $\sum_{z\in\mathcal{Z}} M_z =\mathds{1}$, and where $\mathcal{Z}$ is a set of $m$ possible outcomes of the measurement. For a measurement~$\mathcal{M}$, $\rank(\mathcal{M})$ denotes the number of possible outcomes~$|\mathcal{Z}|$. Without loss of generality we can assume $\mathcal{Z}=[m]$. In this work we focus on measurements with a finite number of outcomes, letting $\Xi(d,m)$ denote the set of all $m$-outcome measurements on $d$-dimensional states, and $\Xi(d)\coloneqq\bigcup_{m\in \mathbb{Z}_+}\Xi(d,m)$ denote the set of all finite-outcome measurements on $d$-dimensional states. The distribution of the random outcome $\bm{z}$ from measuring the state $\rho$ is described by the PMF $p_{\bm{z}} = \text{diag}(\mathcal{M}(\rho))$, so that $p_{\bm{z}}(z)=\Tr(M_z\rho)$ for all outcomes $z$.

Suppose there is a single $d$-dimensional register which can be prepared in the state $\rho$ upon request, at which point it is measured once, and this process is repeated $n$ times. We refer to the class of measurements corresponding to this scenario as \textit{single-copy measurements}, where the number of samples used is equal to the number of measurements performed. Within this class, there are two models of particular interest.

\paragraph{Nonadaptive measurements.}
Consider $n$ copies of the state $\rho\in\mathsf{D}(d)$ prepared in the above manner, so that they must be measured individually. In the \textit{nonadaptive} measurement model, we use a sequence of measurements $\mathcal{M}_i\in\Xi(d)$ for $i=1,\dots,n$ which are determined before any measurements are performed. Equivalently, we measure the state $\rho^{\otimes n}$ using a tensor product of measurements on $d$-dimensional states, $\mathcal{M}_1\otimes\mathcal{M}_2\otimes\dots\otimes\mathcal{M}_n$. Note that allowing the choice of the $i^\text{th}$ measurement to be an independent random variable is equivalent to the above description, since the randomness in the choice of measurement can be incorporated into the measurement itself. I.e., the resulting linear maps on $d$-dimensional states still correspond to some fixed measurements.

\paragraph{Adaptive measurements.}
In the \textit{adaptive} measurement model, the choice of each $d$-dimensional measurement in the sequence can depend on the outcomes obtained by the previous measurements. This means that the $i^\text{th}$ measurement in the sequence can be written $\mathcal{M}^{y_{<i}}$, where $y_{<i}=y_{i-1}\dots y_1$ are the outcomes of the previous $i-1$ measurements. For each possible value of $y_{<i}$ there is a POVM $( M^{y_{<i}}_{y_i})_{y_i}$ corresponding to $i^\text{th}$ measurement and potentially depending on~$y_{<i}$, such that the measurement has the action
\begin{align*}
\mathcal{M}^{y_{<i}}: \rho \mapsto \sum_{y_i} \Tr(M^{y_{<i}}_{y_i}\rho)\ket{y_i}\bra{y_i}
\end{align*}
on quantum states $\rho\in \mathsf{D}(d)$.

\section{Packing construction}

To demonstrate lower bounds for quantum tomography, it suffices to show that there exists a large, but well-separated collection of quantum states (an $\eps$-packing) which are difficult to discriminate with too few copies of the state. This is due to the fact that the task of state discrimination reduces to tomography with sufficient accuracy when the states are far enough apart, since the latter task allows one to correctly identify the state in the ensemble under these conditions. We therefore aim to construct a hard instance of the state discrimination problem, and then argue that if the number of samples $n$ is too small the success probability of our protocol goes to zero as the parameters $d$ and $1/\eps$ increase.
\begin{definition}[$\eps$-packing]
A finite set of quantum states $\mathcal{S}\subset\mathsf{D}(d)$ is an $\eps$-\textit{packing} for some $\eps>0$ if it holds that $\norm{\rho-\sigma}_1> \eps$ for every $\rho,\sigma\in \mathcal{S}$ such that $\rho\neq \sigma$. 
\end{definition}
Let $\{\ket{i}:i\in [d]\}$ denote the standard basis for $\mathbb{C}^d$, and let $Q_{k}$ be the orthogonal projection operator onto the subspace spanned by $\{\ket{i}:i\in [k]\}$. The $\eps$-packing we construct comprises states of the following form:
\begin{align}\label{eqn:parametrization}
    \rho_{\eps,U} \eqdef \frac{2\eps}{d} U Q_{d/2}  U^\dag + \frac{1-\eps}{d}\mathds{1}
\end{align}
where $\eps \in (0,1)$ and we assume $d$ is even for simplicity. The assumption of~$d$ being even does not take away from the argument, and we may proceed analogously with a floor or ceiling when it is odd. States of the above form have also been considered in the previous lower bounds for tomography and related tasks (see, e.g., Refs.~\cite{Haah_2017,BCL20-mixedness-testing}). Intuitively, these states are useful because they represent a hard case where the completely mixed state is slightly perturbed, which leads to ``noisy" measurement statistics. This is in analogy with the packing of distributions which one would construct to prove lower bounds for distribution estimation, the classical analogue of tomography. We make use of the definition in Eq.~\eqref{eqn:parametrization} frequently in the remainder of this paper.

We apply the probabilistic method to construct an $\eps$-packing of states of this form. We draw a sequence of i.i.d.\ unitary operators $U_1,U_2,\dots$ from the Haar distribution on $\mathsf{U}(d)$ and consider the states $\rho_{\eps,U_i}$. We then  apply standard concentration of measure results to argue that the probability of selecting an undesirable state (that our state ``collides" with a previously chosen one) is exponentially small. This in turn implies that a large fraction of the states are ``safe" choices, so that we may choose one and repeat the argument many times. 

We use the following ``concentration of projector overlaps" result, which is implied by the proof of Lemma III.5 in Ref.~\cite{Hayden2006}, and has also been employed in the lower bounds for tomography which appear in~\cite{Haah_2017} as well as lower bounds for similar tasks (see for example Refs.~\cite{Aaronson2020shadow,HKP21-it-bounds}).
\begin{lemma}\label{lem:concentration_projector_overlaps}
Let $\bm{U}\in \mathsf{U}(d)$ be a Haar-random unitary operator and let $\Pi_1,\Pi_2\in\mathsf{Psd}(d)$ be orthogonal projection operators with rank $r_1$,$r_2$ respectively. For all $t\in (0,1)$ it holds that
\begin{align*}
\Pr_{\bm{U}\sim \textnormal{Haar}}\left[\Tr(\Pi_1\bm{U} \, \Pi_2\bm{U}^\dag)\leq (1-t)\frac{r_1r_2}{d}\right]&\leq \exp\left(-r_1r_2t^2/2\right) \enspace, \qquad \text{and} \\
\Pr_{\bm{U}\sim \textnormal{Haar}}\left[\Tr(\Pi_1\bm{U} \, \Pi_2\bm{U}^\dag)\geq (1+t)\frac{r_1r_2}{d}\right]&\leq \exp\left(-r_1r_2t^2/4\right)
\enspace.
\end{align*}
\end{lemma}
\begin{proof}
By Lemma III.5 in Ref.~\cite{Hayden2006} we have
\begin{align*}
\Pr_{\bm{U}\sim \textnormal{Haar}}\left[\Tr(\Pi_1\bm{U} \, \Pi_2\bm{U}^\dag)\leq (1-t)\frac{r_1r_2}{d}\right]&\leq \exp\left(r_1r_2(t+\ln(1-t))\right)
\end{align*}
for all $t\in (0,1)$, and the first bound follows immediately from the inequalities $\ln(1-t)\leq -t -t^2/2$ which holds for all $t\in(0,1)$. Similarly, the second bound in Lemma III.5 of Ref.~\cite{Hayden2006} is
\begin{align}
    \Pr_{\bm{U}\sim \textnormal{Haar}}\left[\Tr(\Pi_1\bm{U} \, \Pi_2\bm{U}^\dag)\geq (1+t)\frac{r_1r_2}{d}\right]&\leq \exp\left(-r_1r_2(t-\ln(1-t))\right)
\end{align}
for all $t\in(0,1)$, and noting that the inequality $\ln(1+t)\leq t -t^2/4$ holds for all $t\in (0,1)$ completes the proof.
\end{proof}
The second tail bound above is a bit looser than that shown in Ref.~\cite{Haah_2017}, but suffices for our purposes. We now construct a sufficiently large packing of quantum states of the form in Eq.~\eqref{eqn:parametrization} which are difficult to discriminate, using a probabilistic existence argument. This is a special case of the approach adopted in Ref.~\cite{Haah_2017}. 
\begin{lemma}\label{lem:packing_construction}
Fix an $\epsilon\in(0,1)$ and a positive integer $d$, and let $N \leq\lfloor\xi\mathrm{e}^{d^2/32}\rfloor$ be a positive integer for some $\xi\in (0,1]$. Consider a finite set of quantum states $\{\rho_1,\rho_2,\dots,\rho_N\}\subset \mathsf{D}(d)$ where
\begin{align*}
\rho_i = \frac{2\eps}{d} U_iQ_{d/2}  U_i^\dag + (1-\eps)\frac{\mathds{1}}{d}
\end{align*}
for each $i\in [N]$ and $U_1,U_2,\dots,U_N\in\mathsf{U}(d)$ are arbitrary unitary operators. For Haar-random $\bm{U}\in\mathsf{U}(d)$, the probability that ${\norm{\rho_{\eps,\bm{U}}-\rho_i}_1 \leq \eps/2}$ for any $i \in [N]$ is at most $\xi$.
\end{lemma}
\begin{proof}
Define a rank-$d/2$ orthogonal projection operator $P\in\mathsf{Psd}(d)$ as $P \eqdef \mathds{1}-Q_{d/2}$. A straightforward consequence of Lemma~\ref{lem:concentration_projector_overlaps} is that
\begin{align}\label{eq:application_of_projector_overlaps}
\Pr_{\bm{U}\sim\text{Haar}}\left[\Tr(P\bm{U}Q_{d/2}\bm{U}^\dag)\leq d/8\right]\leq \mathrm{e}^{-d^2/32}.
\end{align}
This follows by taking $t=1/2$ in the lemma. Using the definition of $\rho_{\eps,U}$ we have
\begin{align*}
\rho_{\eps,U}-\rho_{\eps,\mathds{1}} = \frac{2\eps}{d}\left(UQ_{d/2}U^\dag - Q_{d/2}\right)
\end{align*}
for any $U\in\mathsf{U}(d)$. We also have
\begin{align*}
\Tr(PUQU^\dag) &= \frac{1}{2}\left[\Tr(PUQ_{d/2}U^\dag)+\Tr((\mathds{1}-Q_{d/2})UQ_{d/2}U^\dag)\right] & \text{(by the definition of $P$)}\\
&=\frac{1}{2}\Tr\left((UQ_{d/2}U^\dag-Q_{d/2})(P-Q_{d/2})\right) & \text{($P,Q_{d/2}$ are orthogonal)}\\
&\leq \frac{1}{2}\norm{UQ_{d/2}U^\dag - Q_{d/2}}_1 = \frac{d}{4\eps}\norm{\rho_{\eps,U}-\rho_{\eps,\mathds{1}}}_1
\end{align*}
where the final line follows from the property that $\norm{X}_1 = \max\{|\Tr(XU)|:U\in\mathsf{U}(d)\}$ for any square operator $X\in\mathsf{L}(d)$, and $P-Q_{d/2}\in\mathsf{U}(d)$. Therefore, if $\norm{\rho_{\eps,U}-\rho_{\eps,\mathds{1}}}_1\leq \eps/2$ for a unitary operator $U\in\mathsf{U}(d)$, we also have that $\Tr(PUQ_{d/2}U^\dag)\leq d/8$, from which we may conclude
\begin{align*}
\Pr_{\bm{U}\sim\text{Haar}}\left[\norm{\rho_{\eps,\bm{U}}-\rho_{\eps,\mathds{1}}}_1\leq \eps/2\right] \leq \mathrm{e}^{-d^2/32}
\end{align*}
by Eq.~\eqref{eq:application_of_projector_overlaps}. Next, consider the unitary operator $U_i$ and corresponding state $\rho_i$ in the lemma, for some $i\in [N]$. Using the invariance of the trace distance under unitary transformations, we have
\begin{align*}
\norm{\rho_{\eps,U_i\bm{U}}-\rho_i}_1 = \norm{U_iUQ_{d/2}U^\dag U_i^\dag - U_iQ_{d/2}U_i^\dag}_1 = \norm{U_i(UQ_{d/2}U^\dag - Q_{d/2})U_i^\dag}_1 = \norm{\rho_{\eps,\bm{U}}-\rho_{\eps,\mathds{1}}}_1
\end{align*}
which leads to the conclusion that
\begin{align}\label{eqn:packing_upper_bound}
\Pr_{\bm{U}\sim\text{Haar}}\left[\norm{\rho_{\eps, \bm{U}}-\rho_i}_1\leq \eps/2\right] \leq \mathrm{e}^{-d^2/32}
\end{align}
by invariance of the Haar measure. Since this inequality holds for any index $i\in [N]$ the proof is complete upon applying the union bound over the events $\norm{\rho_{\eps,\bm{U}}-\rho_i}_1 \leq \eps/2$, $i\in [N]$: we have that this probability is at most $N\mathrm{e}^{-d^2/32}\leq \xi$.
\end{proof}
Using Lemma~\ref{lem:packing_construction} we may construct a (non-explicit) set of~$N$ states with~$N\in\exp(\Omega(d^2))$, which form an $\eps/2$-packing in trace distance, using a probabilistic existence argument.
\begin{corollary}\label{cor:packing_cor}
Fix an $\eps\in(0,1)$ and a positive integer $d>1$. There exists an $\eps/2$-packing $\mathcal{S}\subset \mathsf{D}(d)$ of ${N\in \exp\left(\Omega(d^2)\right)}$ quantum states of the form in Eq.~\eqref{eqn:parametrization}.
\end{corollary}
\begin{proof}
First, suppose we have a set of states $\mathcal{S}_k=\{\rho_1,\dots,\rho_k\}\subset \mathsf{D}(d)$ which are of the same form as in Eq.~\eqref{eqn:parametrization}, where $k\leq \lceil\mathrm{e}^{d^2/32}\rceil - 1$. Suppose further that this set is an $\eps/2$-packing. From Lemma~\ref{lem:packing_construction} we know that the probability of choosing a unitary operator $\bm{U}\in\mathsf{U}(d)$ Haar randomly such that $\mathcal{S}_k \cup \{\rho_{\eps,\bm{U}}\}$ is \textit{not} an $\eps/2$-packing is strictly less than one. Therefore, there  exists at least one state which we can add to the packing. The result follows by induction on $k$.
\end{proof}
This packing of states is used in the following section to prove lower bounds for nonadaptive tomography.  Then, in Section~\ref{sec:adaptive_lower_bounds} we alter this construction to derive lower bounds on adaptive tomography.

\section{Lower bounds for tomography with nonadaptive measurements}\label{sec:nonadaptive_lower_bounds}

\subsection{Information in measurement outcomes}
We begin with some useful results quantifying our intuition that measurements performed on states in the packing described above are uninformative. Recall that in the nonadaptive case, measurement choices do not depend on the previously observed outcomes. The following lemma enables us to bound mutual information in terms of the $\chi^2$-divergence, which is more amenable to analysis in this context.
\begin{lemma}\label{lem:mut_inf_ub}
Let $\bm{x}$ be an arbitrary random variable and $\bm{y}\in \mathcal{Y}$ be a discrete random variable for some sample space $\mathcal{Y}$. Denote by $p_{\bm{y}|x}:\mathcal{Y}\to [0,1]$ the distribution of $\bm{y}$ conditioned on the event~$\bm{x} = x$. For an arbitrary discrete distribution $q:\mathcal{Y}\to [0,1]$, it holds that
\begin{align}
    I(\bm{x}:\bm{y}) \leq \frac{1}{\ln(2)}\expect{\bm{x}\sim p_{\bm{x}}}\mathrm{D}_{\chi^2}\infdivx{p_{\bm{y}|\bm{x}}}{q}.\label{eq:eqn_605}
\end{align}
\end{lemma}
\begin{proof}
By Lemma~\ref{lem:chi_squared_vs_kl_inequality} we have the inequality $\mathrm{D}_{\textnormal{KL}}\infdivx{a}{b}\leq \mathrm{D}_{\chi^2}\infdivx{a}{b}/\ln(2)$ for any two discrete distributions $a$ and $b$ defined on the same sample space. This implies the relation in Eq.~\eqref{eq:eqn_605} upon showing that 
\begin{align}\label{eq:kl_ineq}
I(\bm{x}:\bm{y}) =  \expect{\bm{x}\sim p_{\bm{x}}}\mathrm{D}_{\textnormal{KL}}\infdivx{p_{\bm{y}|\bm{x}}}{p_{\bm{y}}}\leq \expect{\bm{x}\sim p_{\bm{x}}}\mathrm{D}_{\textnormal{KL}}\infdivx{p_{\bm{y}|\bm{x}}}{q}.
\end{align}
This inequality is a special case of Lemma~6 in Ref.~\cite{buscemi2010}, but for completeness we include a proof below. Using the definition of KL-divergence, we have
\begin{align*}
\expect{\bm{x}\sim p_{\bm{x}}}\mathrm{D}_{\textnormal{KL}}\infdivx{p_{\bm{y}|\bm{x}}}{q} &= \expect{\bm{x}\sim p_{\bm{x}}}\sum_{y\in\mathcal{Y}} p_{\bm{y}|\bm{x}}(y)\log\left(\frac{p_{\bm{y}|\bm{x}}(y)}{q(y)}\right)\\
&= \sum_{y\in\mathcal{Y}}p_{\bm{y}}(y)\log\left(\frac{1}{q(y)}\right)-H(\bm{y}|\bm{x})\\
&= \mathrm{D}_{\textnormal{KL}}\infdivx{p_{\bm{y}}}{q} + H(\bm{y}) - H(\bm{y}|\bm{x})\\
&= \mathrm{D}_{\textnormal{KL}}\infdivx{p_{\bm{y}}}{q}+I(\bm{x}:\bm{y})
\end{align*}
which proves the inequality in Eq.~\eqref{eq:kl_ineq} since $\mathrm{D}_{\textnormal{KL}}\infdivx{p_{\bm{y}}}{q}\geq 0$.
\end{proof}
\begin{corollary}\label{cor:mut_inf_ub}
Define $\bm{x},\bm{y}$ as in Lemma~\ref{lem:mut_inf_ub}. It holds that
\begin{align}\label{eq:mut_inf_ub_cor}
I(\bm{x}:\bm{y}) \leq \frac{1}{\ln(2)} \expect{\bm{x}\sim p_{\bm{x}}}\mathrm{D}_{\chi^2}\infdivx{p_{\bm{y}|\bm{x}}}{p_{\bm{y}}} = \frac{1}{\ln(2)}\left(\sum_{y\in\mathcal{Y}}\expect{\bm{x}\sim p_{\bm{x}}}\frac{p_{\bm{y}|\bm{x}}(y)^2}{p_{\bm{y}}(y)}-1\right).
\end{align}
\end{corollary}
In the analysis of state tomography, $\bm{x}$ corresponds to a random state from a suitably chosen ensemble. Although these results could be applied directly to the information contained in each measurement about $\bm{x}$, it would be intractable to compute an expectation over $\bm{x}$ since we do not explicitly know the states in our ensemble, whose existence is argued by means of the probabilistic method. Fortunately, we can make use of an intermediate result to effectively replace that ensemble with one which admits such explicit calculations, as explained in the following proposition. (A similar property is also used in the proof of Lemma~10 in Ref.~\cite{Haah_2017}.)
\begin{proposition}\label{prop:suffices_to_consider_haar_mut_inf}
Fix an $\eps\in(0,1)$ and a positive integer $d>1$. Let $\bm{U}\in\mathsf{U}(d)$ be a Haar-random unitary operator and $\bm{z}$ be the outcome obtained upon measuring the random state $\rho_{\eps,\bm{U}}^{\otimes n}\in\mathsf{D}(d^n)$ with the measurement $\mathcal{M}\in \Xi(d^n)$, where $\rho_{\eps,U}\in \mathsf{D}(d)$ is defined as in Eq.~\eqref{eqn:parametrization} for any $U\in\mathsf{U}(d)$. There exists a set of $N\in\exp(\Omega(d^2))$ quantum states $\mathcal{S}  =\{ \rho_1,\dots,\rho_N\}
\subset \mathsf{D}(d)$ of the form in Lemma~\ref{lem:packing_construction} which is an $\eps/2$-packing and which satisfies
\begin{align*}
I(\bm{x}:\bm{y}) \leq I(\bm{U}:\bm{z})
\end{align*}
where $\bm{x}\sim\textnormal{Unif}([N])$ and $\bm{y}$ is the outcome obtained from measuring the random state $\rho_{\bm{x}}^{\otimes n}$ with $\mathcal{M}$.
\end{proposition}
\begin{proof}
Consider a fixed set of $N\in \exp(\Omega(d^2))$ quantum states $\mathcal{S}^\prime=\{\rho_1',\dots,\rho_N'\}\subset\mathsf{D}(d)$ of the form in Lemma~\ref{lem:packing_construction} which is an $\eps/2$-packing. We know such a set exists from Corollary~\ref{cor:packing_cor}. Let $\mathcal{U}=\{U_1,\dots,U_N\}$ be the set of unitary operators such that $\rho_{i}' =\rho_{\eps,U_i}$ for each $i\in [N]$. Note that making the replacement $\mathcal{U}\to W\mathcal{U}$ for an arbitrary unitary operator $W\in\mathsf{U}(d)$ results in another $\eps/2$-packing of $N$ states. Indeed, for any $\rho_i',\rho_j' \in \mathcal{S}'$ we have
\begin{align*}
\norm{\rho_i' -\rho_j' }_1 &= \frac{2\eps}{d}{\lVert U_iQ_{d/2} U_i^\dag-U_jQ_{d/2} U_j^\dag\rVert}_1\\
&= \frac{2\eps}{d}{\lVert WU_iQ_{d/2} U_i^\dag W^\dag-WU_jQ_{d/2} U_j^\dag W^\dag\rVert}_1 \\
&= \norm{\rho_{\eps,WU_i}-\rho_{\eps,WU_j}}_1
\end{align*}
by invariance of the trace distance under unitary transformation. Next, define $\bm{y}_W$ to be the outcome obtained by measuring $\rho_{\eps,WU_{\bm{x}}}^{\otimes n}$ with $\mathcal{M}$, and let $\bm{W}\in\mathsf{U}(d)$ be a Haar-random unitary operator chosen independently of $\bm{x}$. We claim that
\begin{align}\label{eq:claim_haar_mut_inf_ub}
\underset{\bm{W}\sim \textnormal{Haar}}{\mathbb{E}}\ I(\bm{x}:\bm{y}_{\bm{W}})\leq I(\bm{U}:\bm{z}).
\end{align}
Let $p_{\bm{y}|W,x}$ to be the distribution of $\bm{y}_W$ given $\bm{x}=x$. We have
\begin{align}
\expect{\bm{W}}\ I(\bm{x}:\bm{y}_{\bm{W}}) &= \expect{\bm{W}}\ H\left(\expect{\bm{x}\sim [N]}\ p_{\bm{y}|W,\bm{x}}\right) - \expect{\bm{W}}\ \expect{\bm{x}\sim [N]}\ H(p_{\bm{y}|W,\bm{x}})\nonumber \\
&\leq H\left(\expect{\bm{W}}\ \expect{\bm{x}\sim [N]}\ p_{\bm{y}|W,\bm{x}}\right) - \expect{\bm{W}}\ \expect{\bm{x}\sim [N]}H(p_{\bm{y}|W,\bm{x}})\nonumber \\
&= H\left(\underset{\bm{x}\sim [N]}{\mathbb{E}}\ \underset{\bm{W}}{\mathbb{E}}\ p_{\bm{y}|\bm{W},\bm{x}}\right)-\underset{\bm{x}\sim [N]}{\mathbb{E}}\ \underset{\bm{W}}{\mathbb{E}}\ H(p_{\bm{y}|\bm{W},\bm{x}}) \enspace, \label{eq:haar_ensemble_mut_inf_ineq}
\end{align}
where the first line follows from the definition of mutual information, the second line uses the concavity of entropy, and in the final line we make use of the independence of $\bm{x}$ and random unitary operator $\bm{W}$. Furthermore, by right-invariance of the Haar measure we have
\begin{align}
\expect{\bm{W}}\ p_{\bm{y}|\bm{W},x} &= \expect{\bm{W}}\ \textnormal{diag}\left(\mathcal{M}(\rho_{\eps,\bm{W}U_x}^{\otimes n})\right)\nonumber\\
&=\expect{\bm{W}}\ \textnormal{diag}\left(\mathcal{M}(\rho_{\eps,\bm{W}}^{\otimes n})\right)\nonumber\\
&= p_{\bm{z}}.\label{eq:haar_invariance_first_term}
\end{align}
Similarly, we have for any $x\in[N]$ that
\begin{align}\label{eq:haar_invariance_second_term}
\expect{\bm{W}}\ H(p_{\bm{y}|\bm{W},x}) = \expect{\bm{U}}\ H(p_{\bm{z}}).
\end{align}
By substituting Eqs.~\eqref{eq:haar_invariance_first_term} and~\eqref{eq:haar_invariance_second_term} into Eq.~\eqref{eq:haar_ensemble_mut_inf_ineq} we arrive at the inequality in Eq.~\eqref{eq:claim_haar_mut_inf_ub}. We may once again invoke a probabilistic existence argument: since the expectation of $I(\bm{x}:\bm{y}_{\bm{W}})$ over unitary operators $\bm{W}$ is at most $I(\bm{U}:\bm{z})$, there exists at least one unitary operator $V\in\mathsf{U}(d)$ for which the inequality $I(\bm{x}:\bm{y}_V)\leq I(\bm{U}:\bm{z})$ holds. The proposition follows by considering the set of quantum states $\mathcal{S}\eqdef\{\rho_{\eps,VU_1},\rho_{\eps,VU_2},\dots,\rho_{\eps,VU_N}\}$.
\end{proof}
Note that in this proposition the measurements performed on the product state can be arbitrary. 

\subsection{Lower bounds for nonadaptive measurements}
\label{sec:nonadaptive_lb_chi_squared}

In light of Proposition~\ref{prop:suffices_to_consider_haar_mut_inf}, in order to prove limitations of algorithms for tomography,
it suffices to bound quantities of the form $I(\bm{U}:\bm{z})$ for Haar-random $\bm{U}\in\mathsf{U}(d)$ and measurement outcome $\bm{z}$. To this end, it is helpful to establish the following relations based on Haar integration.
\begin{lemma}\label{lem:haar_expecs_arb_povms}
Fix an $\eps\in(0,1)$ and a positive integer $d>1$. Let $\bm{U}\in\mathsf{U}(d)$ be a Haar-random unitary operator, $M\in\mathsf{Psd}(d)$ be a positive semidefinite operator such that $M\preceq \mathds{1}$, $\rho_{\eps,U}\in\mathsf{D}(d)$ be defined as in Eq.~\eqref{eqn:parametrization} for each $U\in\mathsf{U}(d)$, and $w\eqdef\Tr(M)/d$. It holds that
\begin{align*}
\expect{\bm{U}}\ \Tr\left(M\rho_{\eps, \bm{U}}\right) = w,
\end{align*}
and
\begin{align*}
\expect{\bm{U}}\  \left( \Tr\left(M\rho_{\eps, \bm{U}}\right)\right)^2\leq w^2\left(1 + \frac{\eps^2}{d+1}\cdot\min\left\{1,\frac{1}{w(d-1)}\right\}\right).
\end{align*}
\end{lemma}
\begin{proof}
We defer the calculation of some Haar integrals to Appendix~\ref{sec:haar_integrals}. By the definition of $\rho_{\eps,U}$ in Eq.~\eqref{eqn:parametrization} the first expectation is
\begin{align*}
\expect{\bm{U}\sim\text{Haar}}\Tr\left(M\rho_{\eps,\bm{U}}\right) &= \frac{2\eps}{d}\expect{\bm{U}\sim\text{Haar}}\Tr\left(M\bm{U}Q_{d/2}\bm{U}^\dag\right)+(1-\eps)w.
\end{align*}
Recall that $Q_{d/2}\in\mathsf{Psd}(d)$ is a rank-$d/2$ orthogonal projection operator. By Proposition~\ref{prop:haar_expecs} in Appendix~\ref{sec:haar_integrals} and the linearity of trace we have
\begin{align*}
\expect{\bm{U}\sim\text{Haar}}\Tr\left(M\bm{U}Q_{d/2}\bm{U}^\dag\right) = \frac{\Tr(M)}{2}.
\end{align*}
This leads to the first identity in the lemma. For the second expectation in the lemma, note that by substituting the definition of $\rho_{\eps,U}$ and expanding we have
\begin{align}
\expect{\bm{U}\sim\text{Haar}}\ \left(\Tr(M\rho_{\eps,\bm{U}})\right)^2 &= \frac{4\eps^2}{d^2}\expect{\bm{U}\sim\text{Haar}}\ \left(\Tr(M\bm{U}Q_{d/2}\bm{U}^\dag)\right)^2 + w^2(1 -\eps^2)\nonumber\\
&=\frac{4\eps^2}{d^2}\Tr\left(M^{\otimes 2}\expect{\bm{U}\sim\text{Haar}}(\bm{U}Q_{d/2}\bm{U}^\dag)^{\otimes 2}\right) + w^2(1 -\eps^2).\label{eq:eqn_654}
\end{align}
The Haar integral on the right-hand side is evaluated explicitly in Proposition~\ref{prop:haar_expecs_2} by setting the rank parameters to $r_1=r_2=d/2$. This yields
\begin{align*}
\expect{\bm{U}\sim\text{Haar}}(\bm{U}Q_{d/2}\bm{U}^\dag)^{\otimes 2} = \frac{1}{4(d^2-1)}\left[(d^2-2)\mathds{1}+dW\right]
\end{align*}
where the identity and swap operator $W$ act on $(\mathbb{C}^{d})^{\otimes 2}$. Substituting into Eq.~\eqref{eq:eqn_654} and making use of the identity $\Tr(W (A\otimes B))=\Tr(AB)$ we find that the right-hand side is equal to
\begin{align}
\frac{\eps^2(d^2-2)(\Tr(M))^2}{d^2(d^2-1)} + \frac{\eps^2\Tr(M^2)}{d(d^2-1)} + w^2(1-\eps^2)&= \frac{\eps^2(d^2-2)w^2}{d^2-1} + \frac{\eps^2\Tr(M^2)}{d(d^2-1)} + w^2(1-\eps^2)\nonumber\\
&= w^2 + \frac{\eps^2(\Tr(M^2)-dw^2)}{d(d^2-1)}.\label{eq:simple_second_moment}
\end{align}
Assume for now that
\begin{align}\label{eq:tr_squared_ineq_meas_op}
\Tr(M^2)-dw^2 \leq \min\{w^2d(d-1),wd\}.
\end{align}
Then the right-hand side of Eq.~\eqref{eq:simple_second_moment} is at most
\begin{align*}
w^2 + \frac{\eps^2}{d(d^2-1)}\cdot\min\{w^2d(d-1),wd\} = w^2\left(1 + \frac{\eps^2}{d+1}\cdot \min\left\{1,\frac{1}{w(d-1)}\right\}\right)
\end{align*}
as required. To prove the inequality in Eq.~\eqref{eq:tr_squared_ineq_meas_op}, we make use of the relations $\Tr(M^2)\leq (\Tr(M))^2 = w^2d^2$ and $\Tr(M^2)\leq \Tr(M)=wd$ both of which follow from the property that $0\preceq M\preceq \mathds{1}$. The second bound of~$wd$ follows  from the nonnegativity of~$dw^2$.
\end{proof}

This leads us to the lower bounds stated in Theorem~\ref{thm:main_nonadaptive_theorem} below. Intuitively, the theorem establishes the following property: for the family of quantum states of the form in Eq.~\eqref{eqn:parametrization}, the ability to distinguish the distribution over outcomes of a measurement from some fixed distribution---as quantified by their $\chi^2$-divergence---is small on average, no matter the measurement performed. In proving this theorem, our analysis is simplified due to Lemma~\ref{lem:mut_inf_ub} as well as techniques for Haar integration based on permutation invariance. (We refer the interested reader to Section~7.2 of Ref.~\cite{watrous2018} for more on this topic.) We also do not assume that the measurement operators which comprise a given POVM are rank-one, as has been considered in other works~\cite{Haah_2017,HKP21-it-bounds,CCHL21-single-copy-measurements}. This allows us to conclude the novel $\Omega(d^4/\eps^2)$ lower bound in the constant-outcome case, in addition to laying the groundwork for the results in Sections~\ref{sec:adaptive_lower_bounds} and~\ref{sec:classical_shadows}.

\begin{theorem}\label{thm:main_nonadaptive_theorem}
Fix an $\eps \in (0,1)$ and a positive integer $d>1$. Let $\rho_{\eps,U}\in\mathsf{D}(d)$ be defined as in Eq.~\eqref{eqn:parametrization}, $\bm{U}\in \mathsf{U}(d)$ be a Haar-random unitary operator, and $\bm{z}$ be the outcome of a measurement $\mathcal{M}\in\Xi(d)$ performed on the random state $\rho_{\eps,\bm{U}}$ such that $p_{\bm{z}|U} = \textnormal{diag}(\mathcal{M}(\rho_{\eps,U}))$ for every $U\in\mathsf{U}(d)$. Then
\begin{align*}
\expect{\bm{U}\sim\textnormal{Haar}}\mathrm{D}_{\chi^2}\infdivx{p_{\bm{z}|\bm{U}}}{p_{\bm{z}}}\leq \frac{\eps^2}{d+1}\cdot \min\left\{1,\frac{\textnormal{rank}(\mathcal{M})}{d-1}\right\}.
\end{align*}
\end{theorem}
\begin{proof}
Let $\mathcal{Z}$ be an alphabet denoting the set of possible outcomes of the measurement $\mathcal{M}$, such that $z\in \mathcal{Z}$ if and only if $|z\rangle\langle z|\in \text{im}(\mathcal{M})$ for orthonormal $\{\ket{z}\}$. By Definition~\ref{def:chi_squared} we have
\begin{align}\label{eq:mut_inf_term_ub_arb_povms}
\expect{\bm{U}\sim\textnormal{Haar}}\mathrm{D}_{\chi^2}\infdivx{p_{\bm{z}|\bm{U}}}{p_{\bm{z}}} = \sum_{z\in\mathcal{Z}}\expect{\bm{U}\sim \text{Haar}}\frac{p_{\bm{z}|\bm{U}}(z)^2}{p_{\bm{z}}(z)}-1
\end{align}
where for fixed $U\in \mathsf{U}(d)$ the conditional probabilities may be written as $p_{\bm{z}|U}(z)=\Tr(M_{z}\rho_{\eps,U})$ for the POVM $(M_{z})_{z}$ corresponding to the measurement $\mathcal{M}$, and the marginal probabilities in the denominator can be written as $p_{\bm{z}}(z)=\mathbb{E}_{\bm{U}\sim\text{Haar}}\Tr(M_{z}\rho_{\eps,\bm{U}})$. Let $w(z)=\Tr(M_z)/d$ for all~$z\in\mathcal{Z}$. By Lemma~\ref{lem:haar_expecs_arb_povms} the right-hand side of Eq.~\eqref{eq:mut_inf_term_ub_arb_povms} is at most
\begin{align}
\sum_{z\in\mathcal{Z}}  w(z)\left(1 + \frac{\eps^2}{d+1}\cdot \min\left\{1,\frac{1}{w(z)(d-1)}\right\}\right) - 1 &= \frac{\eps^2}{d+1}\cdot \min\left\{1,\frac{|\mathcal{Z}|}{d-1}\right\} .
\end{align}
Since $|\mathcal{Z}|= \text{rank}(\mathcal{M})$, this concludes the proof.
\end{proof}
In the above theorem, the rank of $\mathcal{M}$ may be interpreted as the maximum number of outcomes that can be resolved using the measurements, under the assumption that the learner discards each copy of the state after measuring it.

We now have the tools we need to prove the two lower bounds for the nonadaptive case shown in Table~\ref{tab:tomog_bounds_summary}.  The first is a result originally due to Ref.~\cite{Haah_2017}.

\begin{corollary}[Special case of Theorem 4 in Ref.~\cite{Haah_2017}]\label{cor:independent_lbs}
Let $\eps\in(0,1)$. Any procedure for quantum tomography of $d$-dimensional quantum states that is $\eps/4$-accurate in trace distance using nonadaptive, single-copy measurements requires $n\in\Omega\left(d^3/\eps^2\right)$ samples of the unknown state.
\end{corollary}
\begin{proof}
Let $\mathcal{M} = \mathcal{M}_1\otimes\dots\otimes \mathcal{M}_n\in \Xi(d^n)$ be the single-copy, nonadaptive measurement which is performed on the $n$ copies of the unknown state to do tomography. By Proposition~\ref{prop:suffices_to_consider_haar_mut_inf}, there exists an $\eps/2$-packing $\mathcal{S}=\{\rho_1,\dots,\rho_N\}\subset\mathsf{D}(d)$ of $N\in\exp(\Omega(d^2))$ quantum states of the form in Lemma~\ref{lem:packing_construction} such that the following holds. Let $\bm{x}\sim\text{Unif}([N])$ and $\bm{y}=(\bm{y}_1,\dots,\bm{y}_n)$ be the outcome of the measurement $\mathcal{M}$ when performed on $n$ copies of the random state $\rho_{\bm{x}}$. Then $I(\bm{x}:\bm{y})\leq I(\bm{U}:\bm{z})$ where $\bm{U}$ and $\bm{z}=(\bm{z}_1,\dots,\bm{z}_n)$ are defined as in the proposition: $\bm{U}\in\mathsf{U}(d)$ is Haar-random, and $\bm{z}_k$ is the measurement outcome obtained by measuring $\rho_{\eps,\bm{U}}$ with $\mathcal{M}_k$, for each $k\in [n]$. Since the random variables $\bm{z}_k$ are independent given $\bm{U}$, 
% subadditivity of mutual information applies and we have
using the chain rule for mutual information, and monotonicity of entropy under conditioning, we have 
\begin{align*}
I(\bm{U}:\bm{z}) & = \sum_{k=1}^n H(\bm{z}_k|\bm{z}_{<k}) - H(\bm{z}_k|\bm{z}_{<k},\bm{U}) \\
    & =\sum_{k=1}^n H(\bm{z}_k|\bm{z}_{<k}) - H(\bm{z}_k|\bm{U}) \\
    & \leq \sum_{k=1}^n H(\bm{z}_k) - H(\bm{z}_k|\bm{U})\\
    &= \sum_{k=1}^n I(\bm{U}:\bm{z}_k).
\end{align*}
We apply Corollary~\ref{cor:mut_inf_ub} to bound mutual information from above in terms of $\chi^2$-divergence, and then Theorem~\ref{thm:main_nonadaptive_theorem} to each of the terms in this sum to get
\begin{align}
I(\bm{x}:\bm{y})&\leq \frac{n}{\ln(2)}\left(\max_{k\in [n]} \ \expect{\bm{U}\sim\textnormal{Haar}}\mathrm{D}_{\chi^2}\infdivx{p_{\bm{z}_k|\bm{U}}}{p_{\bm{z}_k}}\right)\nonumber\\
&\leq\frac{n\eps^2}{\ln(2)(d+1)}  \min\left\{1,\frac{\max_{k\in [n]}\textnormal{rank}(\mathcal{M}_k)}{d-1}\right\}\label{eq:mut_inf_ub_gen_proof}\\
&\leq \frac{n\eps^2}{\ln(2)(d+1)}.\label{eq:mut_inf_ub_d3_proof}
\end{align}
Under the assumption that the tomography algorithm gives us a state that is accurate to within $\eps/4$ in trace distance, the measurement $\mathcal{M}$ can be used to decode $\bm{x}$ with some constant probability of success. By Fano's inequality as well as the bound in Eq.~\eqref{eq:mut_inf_ub_d3_proof}, it holds that
\begin{align*}
\frac{n\eps^2}{\ln(2)(d+1)}\in\Omega(d^2)
\end{align*}
which is true if and only if $n\in\Omega(d^3/\eps^2)$.
\end{proof}
\begin{corollary}\label{cor:indep_lbs_const_outcome}
Any procedure for quantum tomography of $d$-dimensional quantum states that is $\eps$-accurate in trace distance using nonadaptive, single-copy measurements, each with at most~$\ell$ outcomes, requires $n\in\Omega\left(d^4/\eps^2 \ell \right)$ samples of the unknown state.
\end{corollary}
\begin{proof}
The proof is identical to that for Corollary~\ref{cor:independent_lbs} except we use Theorem~\ref{thm:main_nonadaptive_theorem} to bound the right-hand side of Eq.~\eqref{eq:mut_inf_ub_gen_proof} in terms of the maximum rank of the measurement operators, which in this case is at most~$\ell$ by assumption. 
% Suppose that $L\in O(1)$ is an upper bound on this quantity. 
We then have
\begin{align*}
\frac{n\eps^2 \ell}{\ln(2)(d^2-1)}\in\Omega(d^2)
\end{align*}
which is true if and only if $n\in\Omega(d^4/\eps^2 \ell)$.
\end{proof}
Corollary~\ref{cor:indep_lbs_const_outcome} implies that there is a strong sense in which the folklore ``Pauli tomography" algorithm---which has an upper bound of $O(d^4/\eps^2)$ measurements---is sample-optimal: amongst all possible strategies making use of constant-outcome (and in particular, two-outcome) measurements, there is no way to perform tomography that is more efficient. Note that here it is assumed that each copy of the state is discarded upon performing the measurement. In the more general case where one may perform further non-adaptive measurements on post-measurement states, the lower bound from Corollary~\ref{cor:independent_lbs} applies.

\subsection{Rank-dependent bounds}
\label{sec-na-rank-dependent}

In this section we derive lower bounds for state tomography using non-adaptive single-copy measurements, when the states are known to have bounded rank. 

We consider a different packing of states defined as follows (cf.\ Ref.~\cite[Section~VI.B]{Haah_2017}). Fix~$\nu \in (0,1)$, positive integers~$d \ge 3$ and $r \in [1,d/3]$. For~$i \in [r]$, define the pure state
\begin{align}
\label{eq-rank-1-state}
\ket{\psi_{\nu,i}} \eqdef \sqrt{1 - \nu} \, \ket{d+1-i} + \sqrt{\nu} \, \ket{i} \enspace.
\end{align}
For a unitary operator~$U \in \unitary(\complex^{d-r})$, which we extend to~$\complex^d$ by taking a direct sum with the identity, define the rank~$r$ state
\begin{align}
\label{eq-rank-r-state}
\sigma_{\nu,U} \eqdef U \left( \frac{1}{r} \sum_{i = 1}^r \density{\psi_{\nu,i}} \right) U^\adjoint \enspace.
\end{align}

There is a large packing of states of this form.
\begin{lemma}[part of Lemma~7 in Ref.~\cite{Haah_2017}]
\label{lem-packing-r}
For any~$\nu \in (0,1/4)$ there exists a $\sqrt{\nu}/4$-packing $\mathcal{S}\subset \qstate(d)$ of $N$ quantum states of the form in Eq.~\eqref{eq-rank-r-state}, with $N \in \exp \left( \Omega(rd) \right)$.
\end{lemma}

By the same reasoning as for Proposition~\ref{prop:suffices_to_consider_haar_mut_inf}, we have
\begin{lemma}
\label{lemma-haar-mi-r}
Let $\bm{U}$ be a Haar-random unitary operator over~$\complex^{d-r}$ and $\bm{z}$ be the outcome obtained upon measuring the random state $\sigma_{\nu,\bm{U}}^{\otimes n}$ with some measurement $\mathcal{M}$. There exists a set of $N$ quantum states $\mathcal{S} \eqdef \{ \sigma_1,\dots,\sigma_N\}
\subset \mathsf{D}(d)$ with~$N \in\exp(\Omega(rd))$ of the form in Eq.~\eqref{eq-rank-r-state} which is a $\sqrt{\nu}/4$-packing and satisfies~$ I(\bm{x}:\bm{y}) \leq I(\bm{U}:\bm{z})$, where $\bm{x}\sim\textnormal{Unif}([N])$ and $\bm{y}$ is the outcome obtained from measuring the random state $\sigma_{\bm{x}}^{\otimes n}$ with $\mathcal{M}$.
\end{lemma}

We bound some measurement statistics associated with a random state of the form in Eq.~\eqref{eq-rank-r-state} in preparation for the main results of this section. Let~$\Gamma_1 \eqdef \sum_{i = 1}^{d-r} \density{i}$ and~$\Gamma_0 \eqdef \id - \Gamma_1$.
\begin{lemma}
\label{lem:haar_expecs_arb_povms-r}
Let $\bm{U}$ be a Haar-random unitary operator on~$\complex^{d-r}$, $M\in\mathsf{Psd}(d)$ be a positive semidefinite operator such that $M\preceq \mathds{1}$, $\sigma_{\nu,U} \in\mathsf{D}(d)$ be defined as in Eq.~\eqref{eq-rank-r-state} for each unitary operator~$U$ on~$\complex^{d-r}$, and
\begin{align}
\label{eq-def-w}
w \eqdef \frac{(1 - \nu)}{r} \trace(M \, \Gamma_0 ) + \frac{\nu}{d-r} \trace(M \, \Gamma_1 ) \enspace.
\end{align}
Then
\begin{align*}
\expect{\bm{U}}\ \Tr\left(M \, \sigma_{\nu, \bm{U}} \right) = w \enspace,
\end{align*}
and
\begin{align*}
\expect{\bm{U}}\  \left( \Tr\left(M \, \sigma_{\nu, \bm{U}} \right)\right)^2  
    & \leq  w^2 + \frac{2 \nu^2}{(d-r)^4} \left( \trace (M \, \Gamma_1) \right)^2 + \frac{3 \nu^2}{r(d-r)^2} \trace \left( (M \, \Gamma_1)^2 \right) \\
    & \qquad \mbox{} + \frac{ 2 \nu (1 - \nu) }{r^2 (d-r)} \trace \left( M \, \Gamma_1 M \, \Gamma_0 \right) \enspace.
\end{align*}
\end{lemma}
\begin{proof}
Due to the~$\pm 1$ symmetry of the Haar measure, the terms with an odd number of occurrences of~$\bU$ or~$\bU^\adjoint$ in the expansion of~$\sigma_{\nu, \bm{U}} $ and~$\sigma_{\nu, \bm{U}}^{\tensor 2}$ evaluate to~$0$ in expectation. The expectation of~$\Tr(M \, \sigma_{\nu, \bm{U}} )$ then follows as before. For the bound on the second expectation, note that
\begin{align*}
\expect{\bm{U}}\  \left( \Tr \left(M \, \sigma_{\nu, \bm{U}} \right)\right)^2
    & = \expect{\bm{U}}\  \trace \left( (M \tensor M) \left(\sigma_{\nu, \bm{U}} \tensor \sigma_{\nu, \bm{U}} \right)\right) \enspace.
\end{align*}
Define~$\tGamma_1 \eqdef \sum_{i = 1}^{r} \density{i}$, $ \tGamma_{10} \eqdef \sum_{i = 1}^{r} \ketbra{i}{d+1-i}$, and~$ \tGamma_{01} \eqdef \sum_{i = 1}^{r} \ketbra{d+1-i}{i}$.
Combining the~$\pm 1$ symmetry of the Haar measure with Propositions~\ref{prop-haar-expct-3} and~\ref{prop:haar_expecs_2}, we get
\begin{align*}
\expect{\bU}\ \sigma_{\nu, \bm{U}}^{\tensor 2} 
    & = \expect{\bU} \Big[ \frac{(1 - \nu)^2}{ r^2} \, \Gamma_0^{\tensor 2} + \frac{\nu (1 - \nu)}{ r^2} \Big( \Gamma_0 \tensor \bU \, \tGamma_1 \bU^\adjoint \\
    & \qquad \mbox{} + \bU \, \tGamma_1 \bU^\adjoint \tensor \Gamma_0 + \big( \bU ( \tGamma_{10} + \tGamma_{01} ) \bU^\adjoint \big)^{\tensor 2} \Big) + \frac{ \nu^2}{ r^2} \big( \bU \, \tGamma_1 \bU^\adjoint \big)^{\tensor 2} \Big] \\
    & =  \frac{(1 - \nu)^2}{ r^2} \Gamma_0^{ \tensor 2} + \frac{ \nu (1 - \nu)}{r (d-r)} \Big( \Gamma_0 \tensor \Gamma_1 + \Gamma_1 \tensor \Gamma_0 \Big) \\
    & \qquad \mbox{} + \frac{\nu (1 - \nu)}{ r^2 (d-r)} \sum_{i = 1}^r \sum_{k = 1}^{d-r} \big( \ketbra{k}{d+1-i} \tensor \ketbra{d+1-i}{k} + \ketbra{d+1-i}{k} \tensor \ketbra{k}{d+1-i} \big)  \\
    & \qquad \mbox{} + \frac{\nu^2}{ r (d-r) ((d-r)^2 - 1)} \Big( (r (d - r) - 1) \id + (d - 2r) W \Big) ( \Gamma_1 \tensor \Gamma_1)  \enspace,
\end{align*}
where~$\id$ and~$W$ are the identity and swap operators on~$\complex^d \tensor \complex^d$, respectively.
We have
\[
\frac{1}{(d-r)^2 - 1} \le \frac{1}{(d-r)^2} \left( 1 + \frac{2}{(d-r)^2} \right) \enspace,
\]
since~$(d - r)^2 \ge 2$. Noting that~$\trace((A \tensor B) W) = \trace(AB)$ and~$d-2r \le d - r$, we get
\begin{align*}
\expect{\bm{U}}\  \left( \Tr \left(M \, \sigma_{\nu, \bm{U}} \right)\right)^2
    & \le \frac{(1 - \nu)^2}{ r^2} \big( \trace( M \, \Gamma_0) \big)^2 + \frac{ 2 \nu (1 - \nu)}{r (d-r)} \trace( M \, \Gamma_0) \trace( M \, \Gamma_1) \\
    & \qquad \mbox{} + \frac{ 2 \nu (1 - \nu)}{r^2 (d-r)} \trace( M \, \Gamma_1 M \, \Gamma_0) \\
    & \qquad \mbox{} + \frac{\nu^2}{ (d-r)^2}  \left( 1 + \frac{2}{(d-r)^2} \right) \big( \trace( M \, \Gamma_1) \big)^2 + \frac{3 \nu^2}{ r(d-r)^2}  \trace \big( ( M \, \Gamma_1)^2  \big) \enspace.
\end{align*}
The bound in the statement of the lemma now follows by the definition of~$w$ in Eq.~\eqref{eq-def-w}.
\end{proof}

We now prove a slightly stronger lower bound for the tomography of states with bounded rank as compared with the bound implied by Ref.~\cite{Haah_2017}; see the remark following Theorem~4 in this reference. The proof of the said theorem assumes that the rank of the input states is strictly greater than one, so the bound for pure states in the theorem we establish below appears to be new.

Recall that~$d \ge 3$ and~$r \in [1,d/3]$.
\begin{theorem}
\label{thm-na-lb-r}
Let $\eps \in (0,1/8)$. Any algorithm for quantum tomography of rank~$r$ quantum states in $d$-dimensions that uses non-adaptive, single-copy measurements and produces an approximation within~$\eps$ in trace distance with positive constant probability requires $ \Omega\left(r^2 d/\eps^2 \right)$ samples of the unknown state.
\end{theorem}
\begin{proof}
We proceed as in the proof of Corollary~\ref{cor:independent_lbs}. Consider any algorithm as in the statement of the theorem, and let $\mathcal{M} \eqdef \mathcal{M}_1\otimes\dots\otimes \mathcal{M}_n\in \Xi(d^n)$ be the single-copy, non-adaptive measurement performed by it on the~$n$ copies of the unknown pure state.

We take~$\nu \eqdef 64 \eps^2$. By Lemma~\ref{lemma-haar-mi-r}, there exists a $2 \eps$-packing $\mathcal{S} \eqdef \{\sigma_1, \dotsc, \sigma_N\} \subset \mathsf{D}(d)$ of quantum states of the form in Eq.~\eqref{eq-rank-r-state}, with $N\in\exp(\Omega(rd))$, which satisfy the following property. Let $\bm{x}\sim\text{Unif}([N])$ and $\bm{y} \eqdef (\bm{y}_1,\dots,\bm{y}_n)$ be the outcome of the measurement $\mathcal{M}$ when performed on $n$ copies of the random state $\sigma_{\bm{x}}$. Then $I(\bm{x}:\bm{y})\leq I(\bm{U}:\bm{z})$, where $\bm{U}$ and $\bm{z} \eqdef (\bm{z}_1,\dots,\bm{z}_n)$ are defined as in the lemma: $\bm{U}$ is a Haar-random unitary operator over~$\complex^{d-r}$, and $\bm{z}_k$ is the measurement outcome obtained by measuring $\sigma_{\nu,\bm{U}}$ with $\mathcal{M}_k$, for each $k\in [n]$.

As in Corollary~\ref{cor:independent_lbs}, using Lemma~\ref{lem:haar_expecs_arb_povms-r}, we get
\begin{align}
\nonumber
I(\bm{x}:\bm{y}) &\leq \frac{n}{\ln(2)} \left(\max_{k\in [n]} \ \expect{\bm{U}} \, \mathrm{D}_{\chi^2}\infdivx{p_{\bm{z}_k|\bm{U}}}{p_{\bm{z}_k}}\right) \\
\nonumber
& \leq \frac{n}{\ln(2)} \left[ \frac{2 \nu^2}{ (d-r)^4} \sum_z \frac{1}{w_z} \left( \trace (M_z \, \Gamma_1) \right)^2 +  \frac{3 \nu^2}{ r (d-r)^2} \sum_z \frac{1}{w_z} \trace \left( (M_z \, \Gamma_1)^2 \right) \right. \\
\label{eq-na-mi-ub-r}
    & \qquad \left. \mbox{} + \frac{ 2 \nu (1 - \nu) }{ r^2 (d-r)} \sum_z \frac{1}{w_z} \trace \left( M_z \, \Gamma_1 M_z \, \Gamma_0 \right) \right] \enspace,
\end{align}
where~$(M_z)$ is one of the~$n$ measurements~$\cM_k$ which maximizes the expected~$\chi^2$-divergence above, and~$w_{z}$ is given by Eq.~\eqref{eq-def-w} with~$M \eqdef M_z$. 

Since the algorithm approximates the unknown state to within~$\eps$ and the states~$\sigma_x$ form a~$2\eps$-packing, the algorithm correctly identifies~$\bm{x}$ with positive constant probability. By Fano's Inequality, we have
\begin{align}
\label{eq-na-mi-lb-r}
I(\bm{x}:\bm{y}) & \in \Omega(rd) \enspace.
\end{align}
To conclude the lower bound of~$\Omega(r^2 d/ \eps^2 )$ on~$n$, it suffices to show that the right side of Eq.~\eqref{eq-na-mi-ub-r} is of the order of~$n \nu / r$.

We have~$\trace \big(( M_z \, \Gamma_1)^2 \big) \le \big( \trace ( M_z \Gamma_1) \big)^2$, and
\begin{align*}
\trace \left( M_z \, \Gamma_1 M_z \, \Gamma_0 \right)
    & \le \left( \trace ( \Gamma_1 M_z^2 \, \Gamma_1) \right)^{1/2}  \left( \trace ( \Gamma_0 M_z^2 \, \Gamma_0 ) \right)^{1/2} \\
    & \le \left( \trace ( M_z \, \Gamma_1) \right)  \left( \trace ( M_z \, \Gamma_0 ) \right) \enspace.
\end{align*}
Combining this with~$2r/(d-r)^2 \le 1$ and the definition of~$w_z$ (in particular that~$w_z \ge (\nu/(d-r)) \trace(M_z \, \Gamma_1)$), we get the desired bound on the right hand side of Eq.~\eqref{eq-na-mi-ub-r}:
\begin{align*}
\frac{n}{\ln(2)} \cdot \frac{\nu}{r(d-r)} \sum_z \frac{ \trace( M_z \, \Gamma_1 )}{w_z} & \left[ \frac{4 \nu}{ (d-r)} \trace (M_z \, \Gamma_1) + \frac{2 (1 - \nu)}{r} \trace ( M_z \, \Gamma_0 ) \right] \\
    & \le \frac{n}{\ln(2)} \cdot \frac{ 4 \nu}{r(d-r)} \sum_z \trace( M_z \, \Gamma_1 ) \\
    & \le \frac{4n \nu}{r \ln(2)} \enspace.
\end{align*}
This completes the proof.
\end{proof}

The same approach allows us to derive stronger bounds when the measurements used by the tomography algorithm have a constant number of outcomes. Again,~$d \ge 3$ and~$r \in [1,d/3]$.
\begin{theorem}
\label{thm-na-lb-const-r}
Let $\eps \in (0,1/8)$. Any algorithm for quantum tomography of rank~$r$ quantum states in $d$-dimensions that uses non-adaptive, single-copy measurements \emph{with at most~$\ell$ outcomes\/} and produces an approximation within~$\eps$ in trace distance with positive constant probability requires $ \Omega\left(r^2 d^2 /\ell \eps^2 \right)$ samples of the unknown state.
\end{theorem}
\begin{proof}
The proof largely proceeds as for Theorem~\ref{thm-na-lb-r}. We use the same notation here, and only indicate where we deviate from that proof. 

To conclude the claimed lower bound on~$n$, it suffices to show that the right side of Eq.~\eqref{eq-na-mi-ub-r} is of the order of~$n \ell \nu /rd $. We have
\[
\trace \big(( M_z \, \Gamma_1)^2 \big) = \trace \big( (\Gamma_1 M_z \, \Gamma_1)^2 \big) \le \trace ( \Gamma_1 M_z \, \Gamma_1) \enspace, 
\]
and
\[
\trace ( M_z \, \Gamma_1 M_z \, \Gamma_0) \le \trace ( M_z^2 \, \Gamma_0) \le \trace ( M_z \, \Gamma_0) \enspace,
\]
since~$\Gamma_1 \preceq \id$, $ M_z^2 \preceq M_z$, and~$\Gamma_0 \succeq 0$. Further observe that~$w_z \geq \nu\Tr(M_z\Gamma_1)/(d-r)$, $2r / (d - r)^2 \le 1$, and~$\trace (M_z \, \Gamma_1) \le d-r$. We thus get the following bound on the right side of Eq.~\eqref{eq-na-mi-ub-r}:
\begin{align*}
\frac{n}{\ln(2)} \cdot \frac{\nu}{r (d-r)} & \left[ \sum_z \frac{1}{w_z} \left( \frac{3 \nu}{d-r} \trace (M_z \, \Gamma_1) +  \frac{2 (1 - \nu)}{r} \trace ( M_z \, \Gamma_0) \right) \right. \\
    & \qquad \qquad \left. \mbox{}  + \frac{2r}{ (d-r)^2} \sum_z \frac{\nu }{w_z (d-r)} \left( \trace (M_z \, \Gamma_1) \right)^2 \right] \\
    & \le \frac{n}{\ln(2)} \cdot \frac{\nu}{r(d-r)} \left[  3 \ell + \frac{2r}{ (d-r)^2} \sum_z \trace (M_z \, \Gamma_1) \right] \\
    & \le \frac{n}{\ln(2)} \cdot \frac{4 \nu \ell}{r (d-r)} \enspace.
\end{align*}
Here, we bounded the term in the curved parentheses in the first line by~$3 w_z$.
The result is the desired bound on the right hand side of Eq.~\eqref{eq-na-mi-ub-r}.
\end{proof}

\section{Lower bounds for tomography with adaptive measurements}
%\section{Lower bounds for tomography with adaptive measurements and finite measurement settings}
\label{sec:adaptive_lower_bounds}

In Section~\ref{sec:nonadaptive_lower_bounds} we saw that it is possible to derive lower bounds on tomography that are stronger than the bound of $\Omega(d^2)$ obtained by a direct application of Holevo's theorem, by considering restricted measurements. (See also Theorem~4 in Ref.~\cite{Haah_2017}, or Chapter~5 in Wright's Ph.D.\ thesis~\cite{wright2016}.) Specifically, we were able to show optimal lower bounds on tomography in the nonadaptive case for both constant-outcome and arbitrary measurements. In this section we consider a different kind of restriction on the measurements; namely, the measurements may be adaptively chosen, so long as they are chosen from a finite set of $m$ different measurements. 

We show that even when we have a choice of $\exp(d)$ different measurements, the $\Omega(d^3/\eps^2)$ lower bound from the previous section continues to hold. In particular, this lower bound applies to the case when the single-copy measurements are all efficiently implementable, i.e., implementable as uniformly generated polylog($d$)-size quantum circuits. In other words, adaptivity does not help while using measurements involving quantum computation with a number of gates growing only polynomially in the number of qubits measured. In Section~\ref{sec:solovay_kitaev} we explain how these results can also be applied to rule out an advantage for adaptivity using a possibly infinite number of measurement settings (i.e., when the measurements are chosen from an infinite set), whenever the measurements are implementable with $\text{polylog}(d)$-size quantum circuits.

The approach we take also leads to a lower bound in the adaptive, constant-outcome case which generalizes an earlier result due to Ref.~\cite{Flammia_2012}. There it is shown that $\Omega(d^4/\log(d))$ copies of the state are required when one is restricted to two-outcome, projective, possibly adaptive Pauli measurements.

\subsection{Distinguishability of a hard ensemble}
\label{sec-distinguishability}

To arrive at lower bounds robust to adaptivity, we once again appeal to difficult instances of the quantum state discrimination problem. This time, however, we construct a packing of quantum states with the additional requirement that all selected states lead to uninformative measurements using \textit{any\/} measurement from a fixed set of possibilities. Such a construction is enabled by a tail bound on the $\chi^2$-divergence quantities we have been considering, so that most states, in addition to being well-separated from previous choices, offer only uninformative measurement statistics. Similar tail bounds have been derived in prior work for the purpose of showing unconditional lower bounds for quantum state certification with adaptive measurements~\cite{BCL20-mixedness-testing}.

The concentration of measure property we invoke to arrive at our tail bounds follow from log-Sobolev inequalities, and is analogous to L\'evy's Lemma for functions on the unit sphere~\cite{M19-random-matrix-theory}. A detailed discussion is beyond the scope of this work, but roughly speaking these imply that sufficiently well-behaved functions of unitary operators concentrate strongly around their expectation. In particular, we have the following theorem.
\begin{theorem}[Special case of Theorem 5.17 in Ref.~\cite{M19-random-matrix-theory}]
\label{lem:unitary_levy_lemma}
Let $d>1$ be a positive integer, $f:\mathsf{U}(d)\to \mathbb{R}$ be $\kappa$-Lipschitz with respect to the metric induced by the Frobenius norm, and let $\mu \eqdef \expct_{\bm{U}\sim\textnormal{Haar}}f(\bm{U})$. Then, for any $t>0$, it holds that
\begin{align*}
\Pr_{\bm{U}\sim\textnormal{Haar}}\left[f(\bm{U})\geq \mu + t\right]\leq \exp\left(-\frac{(d-2) t^2}{24 \kappa^2}\right). 
\end{align*}
\end{theorem}

Before proceeding, we introduce a more convenient short-hand notation for the $\chi^2$-divergence quantities which arose in the analysis in the previous section.
\begin{definition}\label{def:chi_squared_function} For any $\eps\in(0,1)$ and positive integer $d>1$, let $\rho_{\eps,U}\in \mathsf{D}(d)$ be defined as in Eq.~\eqref{eqn:parametrization}. We define the function $F^{\chi^2}_{\eps,d}:\Xi(d) \times \mathsf{U}(d)\to \mathbb{R}$ by
\begin{align*}
F^{\chi^2}_{\eps,d}(\mathcal{M},U) \coloneqq \mathrm{D}_{\chi^2}\infdivx{p_{\bm{z}|U}}{w}
\end{align*}
for all $U\in \mathsf{U}(d)$ and $\mathcal{M}\in\Xi(d)$, where $p_{\bm{z}|U} \coloneqq \textnormal{diag}(\mathcal{M}(\rho_{\eps,U}))$ and $w \coloneqq \expct_{\bm{U}\sim\textnormal{Haar}}p_{\bm{z}|\bm{U}}$.
\end{definition}
For any $\mathcal{M}\in\Xi(d)$ with corresponding measurement operators $\{M_z:z\in\mathcal{Z}\}\subset \mathsf{Psd}(d)$ and $U\in\mathsf{U}(d)$,
\begin{align*}
F^{\chi^2}_{\eps,d}(\mathcal{M},U) = \sum_{z\in\mathcal{Z}}\frac{\Tr(M_z\rho_{\eps,U})^2}{w(z)} - 1 = \sum_{z\in\mathcal{Z}}\frac{\eps^2(2\Tr(M_zUQ_{d/2}U^\dag)/d-w(z))^2}{w(z)} \enspace,
% \leq \eps^2 \enspace,
\end{align*}
where $w(z)  \coloneqq \Tr(M_z)/d$. Here, we have used the definition of $\rho_{\eps,U}$ from {Eq.~\eqref{eqn:parametrization}} to write the expression in terms of the operator $Q_{d/2}$.
%Since $0\leq \Tr(M_zUQ_{d/2}U^\dag)/d\leq 2 w(z)$, we have $\lVert F^{\chi^2}_{\eps,d}\rVert_{\infty}\leq \eps^2$.

We turn to the tail bound which we use in this section to derive lower bounds in the case of adaptive measurements.
\begin{lemma}[$\chi^2$-squared tail bound]\label{lem:chi_squared_concentration}
Fix an $\eps\in(0,1)$ and a positive integer $d \ge 4$. For any finite-outcome measurement $\mathcal{M}\in\Xi(d)$ it holds that
\begin{align}\label{eq:chi_sqaured_concentration}
\Pr_{\bm{U}\sim \textnormal{Haar}}\left[F^{\chi^2}_{\eps,d}(\mathcal{M},\bm{U})>\alpha + t\right] \leq \exp\left(-\frac{C d^2 t}{\eps^2}\right)
\end{align}
where $\alpha\eqdef c\eps^2/d$ and $c,C$ are universal constants that we may take to be~$2$ and~$ 1/(3 \cdot 2^8)$, respectively. Furthermore, if $\mathcal{M}$ is restricted to having $\ell$ outcomes then the inequality holds with $\alpha\eqdef 4\ell \eps^2/ 3 d^2$.
\end{lemma}
\begin{proof}
We first consider the case where the measurement $\mathcal{M}$ may have an arbitrary number of outcomes. Our goal is to prove that the random variable $F^{\chi^2}_{\eps,d}(\mathcal{M},\bm{U})-c\eps^2/d$ is subexponential, where $\bm{U}$ is Haar-random. To accomplish this, we follow the approach in the proof of Lemma~7.6 in Ref.~\cite{BCL20-mixedness-testing}. Instead of bounding the tail of the random variable directly using Lemma~\ref{lem:chi_squared_concentration}, we consider its square root. We are then able to show a comparatively stronger bound on the Lipschitz constant of this function. Translating the resulting subgaussian tail on $\sqrt{F^{\chi^2}_{\eps,d}}$ into a subexponential tail on $F^{\chi^2}_{\eps,d}$ controls its deviations in the regime we care about. In particular, it suffices to show that the function $f$ which acts on $U\in \mathsf{U}(d)$ as
\begin{align*}
f : U \mapsto \sqrt{F^{\chi^2}_{\eps,d}(\mathcal{M},U)}-\expect{\bm{V}\sim \text{Haar}}\sqrt{F^{\chi^2}_{\eps,d}(\mathcal{M},\bm{V})}
\end{align*}
has a tail like $\exp(-\Omega(d^2t^2/\eps^2))$, for $U$ selected randomly from the Haar distribution. 

In more detail, note that
\begin{align*}
\expect{\bm{V}\sim \text{Haar}}\sqrt{F^{\chi^2}_{\eps,d}(\mathcal{M},\bm{V})}\leq \sqrt{\expect{\bm{V}\sim \text{Haar}}F^{\chi^2}_{\eps,d}(\mathcal{M},\bm{V})} \leq \frac{\eps}{\sqrt{d}} 
\end{align*}
by the Jensen inequality and Theorem~\ref{thm:main_nonadaptive_theorem}. Furthermore, the inequality
\[
c\eps^2/d + t \geq \left(\sqrt{c\eps^2/d}+\sqrt{t}\right)^2/2
\]
for any $t \ge 0$ entails that 
\begin{align*}
\Pr_{\bm{U}\sim \textnormal{Haar}} \left[F^{\chi^2}_{\eps,d}(\mathcal{M},\bm{U})>\frac{c\eps^2}{d} + t\right]&=\Pr_{\bm{U}\sim \textnormal{Haar}} \left[\sqrt{F^{\chi^2}_{\eps,d}(\mathcal{M},\bm{U})}>\sqrt{\frac{c\eps^2}{d} + t}\right]\\
&\leq \Pr_{\bm{U}\sim \textnormal{Haar}} \left[\sqrt{F^{\chi^2}_{\eps,d}(\mathcal{M},\bm{U})}>\eps\sqrt{\frac{c}{2d}} + \sqrt{\frac{t}{2}}\right].
\end{align*}
By choosing $c \coloneqq 2$ we find that if $f$ has a tail of $\exp(-\Omega(d^2t^2/\eps^2))$ we get
\begin{align*}
\Pr_{\bm{U}\sim \textnormal{Haar}} \left[F^{\chi^2}_{\eps,d}(\mathcal{M},\bm{U})>\frac{c\eps^2}{d} + t\right]\leq \exp\left(-\frac{Cd^2t}{\eps^2}\right)
\end{align*}
for some universal constant $C$, as required.

To arrive at the desired concentration of measure for $f$ we invoke Theorem~\ref{lem:unitary_levy_lemma}, according to which it suffices to show that $f$ is $O(\eps/\sqrt{d} \,)$-Lipschitz. Let $\bm{z}$ be the outcome obtained by measuring $\rho_{\eps,\bm{U}}$ with $\mathcal{M}$. It has conditional distribution $p_{\bm{z}|U}$ given $\bm{U}=U$. Also recall the distribution~$w$ over outcomes given by~$w\eqdef\expct_{\bm{U}\sim\text{Haar}}p_{\bm{z}|\bm{U}}$. For arbitrary $U,V\in\mathsf{U}(d)$, by Definitions~\ref{def:chi_squared_function} and~\ref{def:chi_squared} and the Triangle Inequality, we have
\begin{align*}
|f(U)-f(V)| &= \left\lvert\sqrt{F^{\chi^2}_{\eps,d}(\mathcal{M},U)}-\sqrt{F^{\chi^2}_{\eps,d}(\mathcal{M},V)}\right\rvert  \\
&=\left\lvert\sqrt{\expect{\bm{z}^\prime\sim w}\left(\frac{p_{\bm{z}|U}(\bm{z}^\prime)}{w(\bm{z}^\prime)}-1\right)^2}-\sqrt{\expect{\bm{z}^\prime\sim w}\left(\frac{p_{\bm{z}|V}(\bm{z}^\prime)}{w(\bm{z}^\prime)}-1\right)^2}\right\rvert \\
&\leq \sqrt{\expect{\bm{z}^\prime\sim w}\left(\frac{p_{\bm{z}|U}(\bm{z}^\prime)}{w(\bm{z}^\prime)}-\frac{p_{\bm{z}|V}(\bm{z}^\prime)}{w(\bm{z}^\prime)}\right)^2} .
\end{align*}
Let $\{M_z:z\in\mathcal{Z}\}$ be the measurement operators corresponding to $\mathcal{M}$. We have that $p_{\bm{z}|U}(z) = \Tr\left(M_{z}\rho_{\eps,U}\right)$ and $w(z)=\Tr(M_z)/d$ from Lemma~\ref{lem:haar_expecs_arb_povms}. Recalling the definition of $\rho_{\eps,U}$ from Eq.~\eqref{eqn:parametrization} we may simplify the right-hand side of the above inequality to arrive at
\begin{align*}
|f(U)-f(V)|\leq \frac{2\eps}{d}\sqrt{\sum_{z\in\mathcal{Z}} \frac{1}{w(z)}\left[\Tr\left(M_{z}(UQU^\dag - VQV^\dag)\right)\right]^2}.
\end{align*}
It suffices to show that the sum in the square root is at most $O(d)\norm{U-V}_{\mathrm{F}}^2$. Write $WDW^\dag$ for the spectral decomposition of the Hermitian matrix $UQU^\dag-VQV^\dag$, where $W$ is unitary and $D$ is diagonal. As explained below, we have
\begin{align*} 
\sum_{z\in\mathcal{Z}} \frac{1}{w(z)}\left[\Tr\left(M_{z}WDW^\dag\right)\right]^2 &= d^2\sum_{z\in\mathcal{Z}}w(z)\left[\Tr\left(\left(\frac{W^\dag M_{z} W}{w(z)d}\right)D\right)\right]^2\\
&\leq d^2 \sum_{z\in\mathcal{Z}}w(z) \Tr\left(\left(\frac{W^\dag M_{z}W}{w(z)d}\right)D^2\right)\\
&=d\sum_{z\in\mathcal{Z}}\Tr\left(W^\dag M_{z} W D^2\right)\\
&=d\norm{UQU^\dag - VQV^\dag}_{\mathrm{F}}^2\\
&\leq 4d\norm{U-V}_{\mathrm{F}}^2 \enspace,
\end{align*}
where in the second line we used the property that $WM_zW^\dag/(w(z)d)$ is positive semidefinite with unit trace and applied Jensen's inequality to deduce that $(\Tr(AD))^2 = \left(\sum_{i}A_{ii}D_{ii}\right)^2\leq \sum_i A_{ii}D_{ii}^2 = \Tr(AD^2)$ for any positive semidefinite matrix $A$ with unit trace. Also, in the fourth line we used the property that the measurement operators for the different outcomes $z$ sum to identity. In the final line, we use the matrix inequality $\norm{AB}_{\mathrm{F}}\leq \norm{A}\norm{B}_{\mathrm{F}}$ to deduce that
\begin{align}
\norm{UQU^\dag - VQV^\dag}_{\mathrm{F}} &= \frac{1}{2}\norm{(U+V)Q(U-V)^\dag + (U-V)Q(U+V)^\dag}_{\mathrm{F}}\nonumber\\
&\leq \norm{(U+V)Q(U-V)^\dag}_{\mathrm{F}}\nonumber\\
&\leq \left(\norm{UQ}+\norm{VQ}\right)\norm{U-V}_{\mathrm{F}}\nonumber\\
&\leq 2 \norm{U-V}_{\mathrm{F}} \enspace.
\end{align}
So $|f(U)-f(V)|\leq (4\eps/\sqrt{d} \,)\; \norm{U-V}_{\mathrm{F}}$, i.e., $f$ is $(4\eps/\sqrt{d} \,)$-Lipschitz and Eq.~\eqref{eq:chi_sqaured_concentration} follows. The proof in the~$\ell$-outcome case is identical, except that the expectation is then 
\[
\expct_{\bm{U}\sim\text{Haar}}F^{\chi^2}_{\eps,d}(\mathcal{M},\bm{U}) \le \frac{4 \ell \eps^2}{ 3 d^2}
\]
for~$d \ge 2$, in accordance with the bound in Theorem~\ref{thm:main_nonadaptive_theorem}.
\end{proof}

\subsection{Sample complexity for adaptive measurements}
\label{sec-lb-adaptive}

Using the concentration of measure results derived in Section~\ref{sec-distinguishability}, we can show lower bounds for single-copy tomography robust to adaptively chosen measurements, so long as the number of different measurements that may be performed is suitably bounded. Our intermediate goal is to construct an $\eps$-packing of states which are especially difficult to discriminate using the choice of measurements available to the learner. We invoke the tail bound from Lemma~\ref{lem:chi_squared_concentration} to claim that for a non-negligible fraction of states of the form in Eq.~\eqref{eqn:parametrization}, the measurement statistics from these measurements are uninformative. This is the content of the following lemma.
\begin{lemma}\label{lem:chi_squared_small_all_meas}
Fix an $\eps\in(0,1)$, positive integer $d \ge 4$, and a set of $m$ measurements $\{\mathcal{M}_1,\mathcal{M}_2,\dots,\mathcal{M}_m\}\subset\Xi(d)$. Let $c,C$ be the universal constants defined in Lemma~\ref{lem:chi_squared_concentration}, and let $\alpha\eqdef c\eps^2/d$. For Haar-random $\bm{U}\in\mathsf{U}(d)$, the probability that $F^{\chi^2}_{\eps,d}(\mathcal{M}_i,\bm{U})\leq \alpha + \eps^2\ln(3m)/Cd^2$ for every $i\in [m]$ is at least $2/3$. Furthermore, if $\max_{i\in[m]} \textnormal{rank}(\mathcal{M}_i) = \ell$, the claim holds with $\alpha\eqdef  4 \ell \eps^2/3 d^2$.
\end{lemma}
\begin{proof}
Applying the union bound over the $m$ possible measurements we find that the probability that there is some measurement $i\in [m]$ such that $F^{\chi^2}_{\eps,d}(\mathcal{M}_i,\bm{U})> \alpha +  \eps^2\ln(3m)/Cd^2$ is at most
\begin{align*}
\sum_{k=1}^m \Pr_{\bm{U}\sim \textnormal{Haar}}\left[F^{\chi^2}_{\eps,d}(\mathcal{M}_k,\bm{U}) > \alpha +\eps^2\ln(3m)/Cd^2\right] &\leq m\exp\left(-\frac{Cd^2}{\eps^2}\cdot  \frac{\eps^2\ln(3m)}{Cd^2}\right)=\frac{1}{3} \enspace,
\end{align*}
where the inequality follows from the tail bound in Lemma~\ref{lem:chi_squared_concentration}.
\end{proof}
We now use a probabilistic existence argument to show that there is a packing which has the desired properties.
\begin{corollary}\label{cor:packing_cor_efficient_meas}
Fix an $\eps\in(0,1)$ and positive integer $d \ge 4$. Let $\{\mathcal{M}_1,\dots,\mathcal{M}_m\}\subset\Xi(d)$ be a fixed set of measurements and define $\alpha, C$ as in Lemma~\ref{lem:chi_squared_small_all_meas}. There exists a set of~$N$ quantum states, $\mathcal{S} \coloneqq \{\rho_1,\dots,\rho_N\}\subset\mathsf{D}(d)$ with
\begin{align*}
\rho_i \coloneqq \frac{2\eps}{d} U_iQ_{d/2}  U_i^\dag + (1-\eps)\frac{\mathds{1}}{d}
\end{align*}
for some unitary operators $U_1,\dots,U_N\in \mathsf{U}(d)$ such that
\begin{enumerate}
\item $N \in\exp(\Omega(d^2))$, 
\item $\mathcal{S}$ is an $(\eps/2)$-packing, and
\item 
\label{cond3}
$F^{\chi^2}_{\eps,d}(\mathcal{M}_i,U_j)\leq \alpha +  \eps^2\ln(3m)/Cd^2$ for every $i\in [m]$ and $j\in[N]$.
\end{enumerate}
\end{corollary}
\begin{proof}
The proof is similar to that of Corollary~\ref{cor:packing_cor}, except that it has an extra step. Suppose we have constructed a set of $k$ quantum states $\mathcal{S}_k \coloneqq \{\rho_1,\dots,\rho_k\}$  with~$k \leq \lfloor\mathrm{e}^{d^2/32-\ln(2)}\rfloor$, where the states are as in the statement of the corollary with corresponding unitary operators $\{U_1,\dots,U_k\}$. Further suppose that~$\mathcal{S}_k$ is an~$\eps/2$-packing and that~$U_j$ satisfies the bound in part~\eqref{cond3} of the statement for all~$j \in [k]$ and $i\in [m]$. By setting the parameter $\xi \coloneqq 1/2$ in Lemma~\ref{lem:packing_construction} and making use of Lemma~\ref{lem:chi_squared_small_all_meas} and the union bound, we see that the probability of selecting a Haar-random unitary operator $\bm{U}\in\mathsf{U}(d)$ such that $\mathcal{S}_{k+1}\eqdef\mathcal{S}_k\cup \rho_{\eps,\bm{U}}$ no longer satisfies either condition is at most $1/2+1/3$. To be precise, the probability that $\norm{\rho_{\eps,\bm{U}}-\rho_j}_1\leq\eps/2$ for some $j\in [k]$ \textit{or} that $F^{\chi^2}_{\eps,d}(\mathcal{M}_i,\bm{U})> \alpha +  \eps^2\ln(3m)/Cd^2$ for some $i\in[m]$ is strictly less than one. Therefore, at least one state satisfying the desired properties exists, and the result follows by induction on $k$.
\end{proof}

We now have all the ingredients to derive the sample complexity for adaptive measurements.
\begin{theorem}\label{thm:tomography_efficient_meas}
Let $\eps\in(0,1)$. Any procedure for quantum tomography of $d$-dimensional quantum states that is $(\eps/2)$-accurate in trace distance and uses single-copy (possibly adaptive) measurements chosen from a fixed set of $m$ measurements requires 
\begin{align*}
n\in\Omega\left(d^3\left(1+\log(m)/d\right)^{-1}/\eps^2\right)
\end{align*}
samples of the unknown state.
\end{theorem}
\begin{proof}
Let $\mathcal{S} \coloneqq \{\rho_1,\dots,\rho_N\}$ be a set of $N\in\exp(\Omega(d^2))$ states which satisfies the  conditions in Corollary~\ref{cor:packing_cor_efficient_meas} for the choice of $m$ measurements $\{\mathcal{M}_1,\dots,\mathcal{M}_m\}\subset\Xi(d)$, with corresponding unitary operators $\{U_1,\dots,U_N\}\subset \mathsf{U}(d)$. Let $\bm{x}\sim\textnormal{Unif}([N])$ and $\bm{y} \coloneqq (\bm{y}_1,\dots,\bm{y}_n)$ be the measurement outcomes from applying $n$ possibly adaptive measurements, each of which is an element of $\{\mathcal{M}_1,\dots,\mathcal{M}_m\}$, on identical copies of $\rho_{\bm{x}}$. (Recall that $\rho_{\bm{x}} = \rho_{\eps,U_{\bm{x}}}$.) By Fano's inequality as well as the assumption that the output of the tomography algorithm is accurate to within trace distance $\eps/2$, we have $I(\bm{x}:\bm{y})\in \Omega(d^2)$. 

On the other hand, we can upper bound the mutual information from above by using the properties of the states which comprise $\mathcal{S}$. Firstly, by the chain rule for mutual information we have
\begin{align}\label{eq:mut_inf_sum_158}
I(\bm{x}:\bm{y}) &= \sum_{i=1}^n I(\bm{x}:\bm{y}_i|\bm{y}_{<i})
\end{align}
where we use the shorthand $\bm{y}_{<i}$ to refer to the sequence of random variables $\bm{y}_{i-1},\dots,\bm{y}_1$. For each $i\in [n]$, let $p_{\bm{y}_i|y_{<i},x}$ be the conditional distribution for the outcome of the $i^\text{th}$ measurement performed on the $i^\text{th}$ copy of the state $\rho_x$, given previous outcomes $y_{<i}$. The probabilities of this distribution are given by
\begin{align*}
p_{\bm{y}_i|y_{<i},x}(y) \coloneqq \Tr(M_{y}^{y_{<i}}\rho_x)
\end{align*}
for each possible outcome $y$, where $\{M_{y}^{y_{<i}}\}_{y}$ is the POVM corresponding to the $i^\text{th}$ measurement $\mathcal{M}^{y_{<i}}$ when the previous~$i-1$ outcomes are~$y_{<i}$. Also, let $w^{y_{<i}}(y)\eqdef\mathbb{E}_{\bm{U}\sim \text{Haar}}\Tr(M_{y}^{y_{<i}}\rho_{\eps,\bm{U}})$ be a fixed distribution, for each possible sequence of prior outcomes $y_{<i}$. Consider the $i^\text{th}$ term in the sum in the right-hand side of Eq.~\eqref{eq:mut_inf_sum_158}. We apply the upper bound on mutual information from Lemma~\ref{lem:mut_inf_ub} as well as Definition~\ref{def:chi_squared_function} for the function $F^{\chi^2}_{\eps,d}(\cdot,\cdot)$ to deduce that
\begin{align}
I(\bm{x}:\bm{y}_i|\bm{y}_{<i}) &= \expect{\bm{y}^\prime_{<i} \sim p_{\bm{y}_{<i}}}I(\bm{x}:\bm{y}_i|\bm{y}_{<i}=\bm{y}^\prime_{<i})\nonumber \\
&\leq \frac{1}{\ln(2)} \expect{\bm{y}^\prime_{<i}\sim p_{\bm{y}_{<i}}} \ \expect{\bm{x}^\prime\sim p_{\bm{x}|\bm{y}_{<i}}} \mathrm{D}_{\chi^2}\infdivx{p_{\bm{y}_i|\bm{y}^\prime_{<i},\bm{x}^\prime}}{w^{\bm{y}^\prime_{<i}}}\nonumber\\
&=  \frac{1}{\ln(2)} \expect{\bm{y}_{<i}} \ \expect{\bm{x}|\bm{y}_{<i}} \ F^{\chi^2}_{\eps,d}(\mathcal{M}^{\bm{y}_{<i}},U_{\bm{x}})\label{eq:expected_chi_squared_function_adaptive_lb_proof}\\
&\leq  \frac{1}{\ln(2)}\left(c\eps^2/d + \eps^2\ln(3m)/Cd^2\right)\nonumber\\
&\in O\left(\frac{\eps^2(1+\log(m)/d)}{d}\right).\label{eq:eqn_509}
\end{align}
where $c,C$ are the universal constants defined in Lemma~\ref{lem:chi_squared_small_all_meas}. The fourth line follows by the assumption that for every $y_{<i}$ we have $\mathcal{M}^{y_{<i}}=\mathcal{M}_j$ for some $j\in[m]$. Applying this argument to each of the $n$ mutual information terms in Eq.~\eqref{eq:mut_inf_sum_158} and combining with the relation $I(\bm{x}:\bm{y})\in\Omega(d^2)$ gives the desired lower bound.
\end{proof}

For a fixed finite gate set, the number of distinct $\text{polylog}(d)$-size quantum circuits is at most $\text{polylog(d)}^{\text{polylog(d)}}\in\exp(o(d))$. Hence, a lower bound of $\Omega(d^3/\eps^2)$ samples holds in the setting where the learner is restricted to such circuits. This improves the bound we obtain from the Holevo theorem by a factor of $d$. Furthermore, this lower bound is tight by the algorithm we present in Appendix~\ref{sec:single_copy_tomog_upper_bound}, along with the fact that random Clifford circuits, which are efficiently implementable~\cite{Aaronson_2004,vandenberg2021simple}, comprise a unitary 3-design~\cite{kueng2015stabilizer3design, webb2017clifford, zhu2017clifford}. The algorithm we present is nearly identical to an algorithm using Haar-random measurements given in Ref.~\cite[Section~5.1]{wright2016}, except we make use of the fact that measurements based on unitary 2-designs suffice\footnote{Ref.~\cite{guta2018fast-arxiv} makes a similar observation; namely, that a measurement based on state 2-designs yields a sample complexity on the order of $d^3\log(d)/\eps^2$.}.

Another particularly simple setting in which the above lower bound works well is that of $d$-outcome Pauli basis measurements, as considered by Yu~\cite{yu2020sample}. Yu shows a $\widetilde{O}(d^{3.32}/\eps^2)$ upper bound on the sample complexity of tomography with non-adaptive measurements, while Theorem~\ref{thm:tomography_efficient_meas} once again yields a lower bound of $\Omega(d^3/\eps^2)$, even with adaptive measurements.

We also have the following extension of the above theorem, which generalizes the lower bound for the case of two-outcome Pauli measurements due to Ref.~\cite{Flammia_2012} (although we do not consider possible dependence on the rank of the state here).
\begin{theorem}\label{thm:tomography_efficient_meas_const_outcome}
Let $\eps\in(0,1)$. Any procedure for quantum tomography of $d$-dimensional quantum states that is $(\eps/2)$-accurate in trace distance and uses single-copy (possibly adaptive), $\ell$-outcome measurements chosen from a fixed set of $m$ possible measurements requires
\begin{align*}
n\in\Omega\left(d^4 / (\ell + \log m) \eps^2\right)
\end{align*}
samples of the unknown state.
\end{theorem}
\begin{proof}
The proof is identical to that for Theorem~\ref{thm:tomography_efficient_meas} except that the right-hand side of Eq.~\eqref{eq:expected_chi_squared_function_adaptive_lb_proof} is of the order of $4 \ell \eps^2/3 d^2 + \eps^2\ln(3m)/Cd^2$, by Theorem~\ref{thm:main_nonadaptive_theorem}.
\end{proof}

\section{Sample complexity of classical shadows}
\label{sec:classical_shadows}

In this section, we consider \emph{classical shadows\/} and \emph{shadow tomography\/}, variants of state tomography which have received much attention recently. Building on the ideas developed in the previous sections, we obtain new bounds on the sample complexity of these problems.

\subsection{Classical shadows}

%\paragraph{Description of task.}

Full quantum state tomography is often unnecessary for determining important properties of a quantum system. For example, to verify the output of a quantum computer, one might only be concerned with comparing the state that is produced to some target pure state, perhaps by estimating their fidelity. Alternatively, in variational quantum algorithms an essential subroutine is to determine the expectation values of some observables encoding the cost function of interest. For both these tasks and more, succinctly represented information about the state known as a \textit{classical shadow}~\cite{Huang2020} can provide an exponential reduction in the number of copies of the state required to learn properties of interest. Informally, a classical shadow of a quantum state refers to a classical string, also called a \emph{sketch\/}, using which we can estimate the expectation values of any given sequence of~$M$ observables to within accuracy~$\eps$. The sketch is produced by the measurement of individual copies of the otherwise unknown state.

For any state~$\rho\in \mathsf{D}(d)$, consider the function~$f_\rho$ mapping~$\mathsf{Psd}(d)^M$  to~$\mathbb{R}^M$, defined as
\[
f_\rho( E_1, E_2, \dotsc, E_M) \coloneqq ( \Tr(E_i \rho) : i \in [M] ) \enspace.
\]
More formally, the associated task is defined below. In this definition, \textit{single-copy access} refers to restricting measurements to individual copies of an unknown state, as described in Section~\ref{sec:measurement_models}.

\begin{definition}[Classical shadows problem]
Given parameters $\eps\in(0,1)$, $B>0$, and $M \ge 1$, and single-copy access to $n$ copies of an unknown quantum state $\rho\in \mathsf{D}(d)$ the \emph{classical shadows problem\/} consists of computing a description of a function $f: \mathsf{Psd}(d)^M \to \mathbb{R}^M$, called a \emph{classical shadow\/}, such that for any fixed collection of $M$ observables~$(O_i : 0 \preceq O_i \preceq \mathds{1}, ~ i \in [M] )$ satisfying $\max_{i\in[M]}\Tr(O_i^2) = B$, it holds that $\norm{f(O_1,\dots,O_M)-f_\rho(O_1,\dots,O_M)}_{\infty}\leq \eps$ with probability at least $2/3$.
\end{definition}

%\paragraph{A sample-efficient algorithm.}

Huang, Kueng, and Preskill~\cite{Huang2020} give a procedure for computing classical shadows which uses only $n\in O(B\log(M)/\eps^2)$ efficient, nonadaptive measurements on single copies of the state $\rho$. Here, the measurements are implemented by using random $q$-qubit Clifford operators, which form a unitary 3-design. Then, the procedure performs a median-of-means estimation of the expectation values. Overall this is an unbounded improvement over full state tomography in the case where $\Tr(O_i^2)$ is at most a constant for the observables of interest $O_i$, since there is no explicit dependence on the dimension. They then show a matching lower bound in the nonadaptive measurement setting. However, this bound does not take into account the possibility of adaptive measurements. We turn to this in Section~\ref{sec:alternative_lb_classical_shadows}, focusing on the case where an upper bound on $B$ is not known.

%\subsection{Lower bound for classical shadows with finite measurement settings}
\subsection{Lower bound with a limited choice of measurements}
\label{sec:alternative_lb_classical_shadows}

In this section, we show how the arguments developed in the previous sections for quantum tomography can be adjusted to give a lower bound for classical shadows with adaptive measurements, when the measurements are chosen from a ``small enough'' set. We obtain this result by proving the same lower bound for a variant of \emph{shadow tomography\/}~\cite{Aaronson2020shadow} with single-copy measurements described below. 
\begin{definition}[Single-copy shadow tomography for bounded operators]
\label{def-shadow-tomography}
Given parameters $\eps\in(0,1)$ and $B>0$, single-copy access to $n$ copies of $\rho\in \mathsf{D}(2^q)$, as well as the description of $M$ observables~$(O_i : 0 \preceq O_i \preceq \mathds{1}, ~ i \in [M] )$ satisfying $\max_{i\in[M]}\Tr(O_i^2) = B$, the task is to output a vector $b\in \mathbb{R}^M$ such that with probability at least $2/3$ we have $|b_i-\Tr(O_i\rho)|\leq \eps$ for every $i\in [M]$.
\end{definition}
Note that the output of the classical shadows problem can be used to produce a solution to the shadow tomography problem, when the Frobenius norm of the input operators is suitably bounded. Hence, any lower bound on the sample complexity for the latter task applies to the classical shadows problem as well.

\begin{table}
\small
%\begin{adjustwidth}{-0.3cm}{}
\begin{center}
\begin{tabular}[h]{c | c | c}
& Upper bound \Tstrut & Lower bound \\
\hline
& & \\
Entangled & $\widetilde{O}(\log(d)\log^2(M)/\eps^4)$~\cite{BO21-data-analysis} & $\Omega(\log(M)/\eps^2)$~\cite{Aaronson2020shadow} \\
& & \\
Single-copy & $O(d\log(M)/\eps^2)$~\cite{Huang2020} & $ \Omega(\min\{M/\log(M),d\}/\eps^2)$~\cite{CCHL21-single-copy-measurements} \\
& & \\
Single-copy \& Efficient & $O(d\log(M)/\eps^2)$~\cite{Huang2020} & $\Omega(d\log(M) /\eps^2)$~(this work)
\end{tabular}
\suppress{
\begin{tabular}{c | c c c}
& Entangled & Single-copy & Single-copy \& Efficient\\
\hline
Upper bound\Tstrut & $\widetilde{O}(\log(d)\log^2(M)/\eps^4)$~\cite{BO21-data-analysis}& $O(d\log(M)/\eps^2)$~\cite{Huang2020} &$O(d\log(M)/\eps^2)$~\cite{Huang2020}\\
Lower bound & $\Omega(\log(M)/\eps^2)$~\cite{Aaronson2020shadow} & $\Omega(\min\{M/\log(M),d\}/\eps^2)$~\cite{CCHL21-single-copy-measurements} & $ \Omega(d\log(M) /\eps^2)$~(this work)
\end{tabular}
}
\end{center}
%\end{adjustwidth}
\caption{\label{tab:shadow_tomography_results}The best known upper and lower bounds on the sample complexity of shadow tomography for $M$ observables $O_1,\dots,O_M$, for $M \in \exp(O(d))$ for entangled measurements and~$M \in\exp(O(d^2))$ for single-copy measurements. Note that for~$M$ larger than the corresponding thresholds, we may use state tomography with joint measurements or single-copy measurements to achieve sample complexity of order~$d^2 / \eps^2$ and~$d^3 / \eps^2$, respectively. The lower bounds for~$M \in \exp( \Omega(d^2))$ are~$\Omega$ of~$d^2/ \eps^2, d/ \eps^2$, and~$d^3/ \eps^2$, respectively for the three cases above. The $\widetilde{O}$ notation hides loglog factors in $d$ and log factors in $1/\eps$.}
% and $\widetilde{\Omega}$ hides $\text{polylog}$ factors in $d$ and $1/\eps$.}
\end{table} 

Table~\ref{tab:shadow_tomography_results} summarizes known results on the sample complexity of shadow tomography under various assumptions about the measurements. In Theorem~\ref{thm:unentangled_shadow_tomography} below, we prove a lower bound on the sample complexity of single-copy shadow tomography for bounded operators, when the possible measurements available to the learning algorithm are limited in number. The bound implies that in the setting of single-copy measurements with efficient circuits (i.e., uniformly generated $\text{polylog(d)}$-size quantum circuits over a finite universal gate set), the non-adaptive classical shadows algorithm due to Huang \textit{et al.\/}~\cite{Huang2020} is optimal for single-copy shadow tomography. In contrast with the lower bound due to Ref.~\cite{CCHL21-single-copy-measurements} (in the second column of Table~\ref{tab:shadow_tomography_results}), a number of samples exponential in the number of qubits is inevitable using efficiently implementable measurements. This is a consequence of the fact that for a finite set of allowed measurements, one can always construct an instance of the classical shadows problem such that the measurement $\{O_i,\mathds{1}-O_i\}$ is not in the set, for some $i\in [M]$.

\begin{theorem}
\label{thm:unentangled_shadow_tomography}
Any algorithm for the classical shadows or the single-copy shadow tomography problem that only uses single-copy measurements chosen from a fixed set of~$m$ measurements requires
\[
\Omega \! \left( \frac{d \min\{ d^2, \log M \}}{ \eps^2 (1 + \log(m)/d)} \right)
\]
samples when $B = d/2$.
\end{theorem}
\begin{proof} 
As mentioned earlier, it suffices to prove the claimed lower bound for the shadow tomography problem. Consider any algorithm that only uses single-copy measurements chosen from a fixed set of~$m$ measurements~$\{\mathcal{M}_1,\dots,\mathcal{M}_m\}\subset \Xi(d)$.
We construct a set of hard input instances for this algorithm and show that the algorithm requires a large number of samples for these instances.

We observe, as in Ref.~\cite[Theorem~19]{Aaronson2020shadow}, that well-separated states of the form we have been studying (cf.\ Eq.~\eqref{eqn:parametrization}) can be distinguished well by the measurements operators given by their deviation from the completely mixed state. We build on this to show that there exists a special collection of $M$ states $\rho_1,\dots,\rho_M$, and observables $O_1,\dots,O_M$ whose expectation values enable us to uniquely identify a state from the~$M$ alternatives $\rho_1,\dots,\rho_M$. The states satisfy the additional property that the statistics obtained from measuring any of the~$\rho_i$ with any of~$m$ measurements~$ \mathcal{M}_j$ are not very informative. The lower bound then follows from Fano's inequality and the upper bound on the chi-squared divergence quantity we have been considering in the context of tomography. Since we may only take $M$ to be at most $\exp(\kappa d^2)$ for a universal constant~$\kappa$, the lower bound plateaus at this threshold.

More formally, we first construct the difficult instance of the shadow tomography problem. Let $\bm{U}\in\mathsf{U}(d)$ be a Haar-random unitary operator and, as before, let $Q\in\mathsf{Psd}(d)$ be a rank-$d/2$ orthogonal projection operator. By setting the parameter $t=1/3$ in Lemma~\ref{lem:concentration_projector_overlaps}, we get that for any fixed rank-$d/2$ orthogonal projection operator $P\in\mathsf{Psd}(d)$,
\begin{align}\label{eq:eqn_454}
\Pr[\Tr(P\bm{U}Q\bm{U}^\dag)\geq d/3] \leq \exp(-c' d^2)
\end{align}
for a universal constant $c'$. Since the algorithm only uses single-copy measurements from a fixed set of~$m$ measurements, Lemma~\ref{lem:chi_squared_small_all_meas} applies and we have the following result.
\begin{lemma}\label{lem:shadow_packing}
Fix an $\eps\in (0,1)$ and a positive integer $d\geq 4$. Define $\alpha, C$ as in Lemma~\ref{lem:chi_squared_small_all_meas}. There is a universal constant~$\kappa$, such that for any~$M \in [1, \exp(\kappa d^2 )]$, there exists a set of $M$ unitary operators $U_1,\dots,U_M \in\mathsf{U}(d)$ such that
\begin{enumerate}
    \item 
    \label{cond-dist}
    $\Tr(U_iQU_i^\dag \, U_jQU_j^\dag)\leq d/3$ for every $i,j\in [M],\ i\neq j$, and
    \item $F^{\chi^2}_{\eps,d}(\mathcal{M}_i,U_j)\leq \alpha + \eps^2\ln(3m)/Cd^2$ for every $i\in [m]$ and $j\in [M]$.
\end{enumerate}
\end{lemma}
\begin{proof}
The proof is similar to that for Corollary~\ref{cor:packing_cor_efficient_meas} except that we use Eq.~\eqref{eq:eqn_454} to ensure the first condition (instead of using Lemma~\ref{lem:chi_squared_small_all_meas}).
\end{proof}

Now, let $\mathcal{S} \coloneqq \{\rho_1,\dots,\rho_M \}\subset\mathsf{D}(d)$ be a collection of states of the form in Eq.~\eqref{eqn:parametrization}, given by
\begin{align*}
\rho_i \coloneqq \frac{2\eps}{d} U_i Q  U_i^\dag + (1-\eps)\frac{\mathds{1}}{d} \enspace.
\end{align*}
For any $i\in [M]$
\begin{align*}
\Tr(U_iQU_i^\dag\rho_i) = \frac{1}{2} + \frac{\eps}{2} \enspace,
\end{align*}
while by condition~\eqref{cond-dist} in Lemma~\ref{lem:shadow_packing} we have for any $j\neq i$
\begin{align*}
\Tr(U_jQU_j^\dag\rho_i) \leq \frac{1}{2} + \frac{\eps}{6} \enspace.
\end{align*}
This means that by estimating $\Tr(U_iQU_i^\dag\rho_x)$ with $\eps/12$ accuracy for every $i\in [M]$ we can identify the value of $x\in [M]$. Thus, we may use the algorithm for shadow tomography with input observables~$U_1QU_1^\dag,\dotsc, U_M Q U_M^\dag$ to discriminate between the~$M$ states in~$\mathcal{S}$ with probability at least $2/3$. We argue next that the algorithm requires a ``large'' number of single-copy measurements in order to accomplish this.

Let $\bm{x}$ be uniformly random over $[M]$ and $\bm{y} \coloneqq (\bm{y}_1,\dots,\bm{y}_n)$ be the measurement outcomes obtained from $n$ single-copy (possibly adaptive) measurements performed on distinct copies of the state $\rho_{\bm{x}}$. By chain rule for mutual information, we have
\begin{align}
\nonumber
I(\bm{x}:\bm{y}) &= \sum_{i=1}^n I(\bm{x}:\bm{y}_i|\bm{y}_{<i})\\
\nonumber
&\leq \sum_{i=1}^n \expect{\bm{y}_{<i}} \ \expect{\bm{x}|\bm{y}_{<i}} F^{\chi^2}_{\eps,d}(\mathcal{M}^{\bm{y}_{<i}},U_{\bm{x}})\\
&\in O\left(\frac{n\eps^2(1+\log(m)/d)}{d}\right) \enspace,
\label{eq:rhs_shadow_tomog_mut_inf}
\end{align}
where we have omitted some steps since they are identical to those leading to Eq.~\eqref{eq:eqn_509}. 

On the other hand, since the algorithm identifies the state~$\rho_{\bm{x}}$ with probability~$\ge 2/3$ from the measurement outcomes~$\bm{y}$, by Fano's Inequality, $I(\bm{x}:\bm{y}) \in \Omega(\log(M))$. This concludes the proof of Theorem~\ref{thm:unentangled_shadow_tomography}.
\end{proof}

Suppose an algorithm for classical shadows or shadow tomography uses only efficient single-copy measurements over a fixed, finite universal gate set. The number~$m$ of different measurements it may use is then~$O(\exp(\textrm{polylog}(d)))$. (See Appendix~\ref{sec:solovay_kitaev} for further justification and a slight generalization.) By the above lower bound, the algorithm requires $\Omega(d\min\{d^2,\log(M)\}/\eps^2)$ samples. In fact, this bound is optimal.

\begin{lemma}
\label{thm:unentangled_shadow_tomography-ub}
There is an algorithm that uses only efficient, single-copy measurements and
\[
O(d\min\{d^2,\log(M)\}/\eps^2)
\]
samples and solves the classical shadows and single-copy shadow tomography problems for arbitrary~$B$. 
\end{lemma}
\begin{proof}
The upper bound is achieved by the following procedure: use the random Clifford operator-based classical shadows algorithm from Ref.~\cite[Theorem~1]{Huang2020} if $M\leq \mathrm{e}^{d^2}$, and the random Clifford operator-based state tomography algorithm of Ref.~\cite{KUENG2017} otherwise. The former has sample complexity of the order of~$d \log(M)/ \eps^2$, and the latter~$d^3 / \eps^2$.
\end{proof}

\subsection{Sample means suffice for single-copy shadow tomography}\label{sec:simple_shadow_tomog_ub}

In this section we turn our attention to the case of nonadaptive --- but otherwise arbitrary --- single-copy measurements on an unknown state $\rho\in\mathsf{D}(d)$.
As can be seen from Table~\ref{tab:shadow_tomography_results}, the median-of-means algorithm due to Ref.~\cite{Huang2020} is optimal up to $\log$ factors for shadow tomography in this setting. Their proposal employs random Clifford operations to perform random basis measurements. However, we may take an even simpler approach using the same measurement scheme, which also turns out to be optimal. Specifically, we show that taking the sample means using the classical shadow reproduces the same upper bound on the overall sample complexity, which is $n\in O(\min\{d,M\}\log(M)/\eps^2)$, assuming $M\leq \mathrm{e}^{d^2}$.

We first handle the case when $M > d$. Suppose we apply a random Clifford operator $\bm{U}\in\mathsf{U}(d)$ and then measure an unknown state $\rho$ in the standard basis $\{\ket{j}\}_{j=1}^d\subset\mathbb{C}^d$. For a fixed Clifford operator $U\in\mathsf{U}(d)$ we may write the operators for this measurement as $\{U\outerprod{j}{j}U^\dag\}_{j=1}^d$. It is well-known that this random projective measurement is closely related to state $t$-designs, as explained in Appendix~\ref{sec:simple_shadow_tomog_ub}. Define the random variable $\hat{\rho}(\bm{U},\bm{j})=(d+1)\bm{U}|\bm{j}\rangle \langle \bm{j}|\bm{U}^\dag -\mathds{1}$ where $\bm{j}\in[d]$ is the random measurement outcome in the standard basis. We show the following theorem.
\begin{theorem}
    Let $\bm{U}$ and $\bm{j}$ be as defined above, and let $\bm{U}_1,\dots,\bm{U}_n$ and $\bm{j}_1,\dots,\bm{j}_n$ be i.i.d.\ copies of $\bm{U}$ and $\bm{j}$, respectively. There exists a positive integer $n\in O\left(\frac{d\log(M)}{\eps^2}\right)$ such that outputting the sample means $\frac{1}{n}\sum_{k=1}^n\Tr(O_i\hat{\rho}(\bm{U}_k,\bm{j}_k))$ for each $i\in [M]$ solves the classical shadows problem.
\end{theorem}

First, note that by Proposition~\ref{prop:unbiased_estimator_2_design} in Appendix~\ref{sec:upper_bounds_tomography} we have $\expct \hat{\rho}(\bm{U},\bm{j})=\rho$. We also make use of the following property.

\begin{proposition}[Prop.~S1, Sec.~5 in the supplementary materials for Ref.~\cite{Huang2020}]\label{prop:shadows_variance}
Let~$X$ be a Hermitian operator with  $-\mathds{1}\preceq X \preceq\mathds{1}$ acting on $\mathbb{C}^d$, and let $\hat{\rho}(\bm{U}, \bm{j})$ be as defined above. It holds that
\begin{align*}
\Var\left[\Tr(X \, \hat{\rho}(\bm{U},\bm{j}))\right]\leq 3\Tr(X^2).
\end{align*}
\end{proposition}
Finally, we require a concentration of measure property of bounded random variables known as Bernstein's inequality. This is stronger than Hoeffding's inequality when the variances of the random variables are sufficiently small. This version of Bernstein's inequality can be found in Ref.~\cite{vershynin_2018}, for example.
\begin{theorem}[Theorem 2.8.4 in Ref.~\cite{vershynin_2018}]
Let $\bm{x}_1,\dots,\bm{x}_n$ be independent, mean zero random variables such that $|\bm{x}_i|\leq K$ with probability $1$ for all $i\in [n]$. Then, for every $\eps\geq 0$, we have
\begin{align*}
\Pr\left[\left\lvert\sum_{i=1}^n\bm{x}_i\right\rvert\geq\eps \right]\leq 2\exp\left(\frac{-\eps^2/2}{\sigma^2+K\eps/3}\right)
\end{align*}
where $\sigma^2\eqdef\sum_{i=1}^n \expct \bm{x}_i^2$.
\end{theorem}
Now suppose that the observables given as input to the shadow tomography algorithm are~$O_i$, with $0 \preceq O_1,\dots,O_M\preceq \mathds{1}$. Define the random variables $f_i(\bm{U},\bm{j})\eqdef\Tr(O_i \, \hat{\rho}(\bm{U},\bm{j}))$ for each $i\in [M]$. It holds that 
\begin{align}
    \expct f_i(\bm{U},\bm{j}) = \Tr(O_i\expct \hat{\rho}(\bm{U},\bm{j}))=\Tr(O_i\rho) \enspace,
\end{align} 
so that~$f_i(\bm{U},\bm{j})$ is an unbiased estimator for~$\Tr(O_i\rho)$. If we perform the random measurement described above on $n$ separate copies of $\rho$, we obtain i.i.d.\ random variables $(\bm{U}_1,\bm{j}_1),\dots,(\bm{U}_n,\bm{j}_n)$. These define the classical shadow of the state as
\[
\frac{1}{n} \sum_{k = 1}^n \hat{\rho}( \bm{U}_k,\bm{j}_k) \enspace.
\]
The expectation value for~$O_i$ predicted by the classical shadow is the sample mean of the $i^\text{th}$ estimator $f_i$. For any $\eps > 0$, by Bernstein's inequality we have that
\begin{align*}
\Pr\left[\left\lvert\frac{1}{n}\sum_{k=1}^nf_i(\bm{U}_k,\bm{j}_k)-\Tr(O_i\rho)\right\rvert>\eps\right]\leq 2\exp\left(\frac{-\eps^2/2}{\sigma^2 + \eps K/(3n)}\right)
\end{align*}
where $\sigma^2\eqdef\frac{1}{n^2}\sum_{k=1}^n\Var[f_i(\bm{U}_k,\bm{j}_k)]$ and $K$ is such that $|f_i(\bm{U}_k,\bm{j}_k)-\Tr(O_i\rho)|\leq K$ with probability~1 for all $k\in [n]$. By definition $\norm{f_i}_{\infty}\leq d+1$ so $K$ can be taken to be $O(d)$, and by Proposition~\ref{prop:shadows_variance} we have $\sigma^2\leq 3d/n$. Taking $n\in O(d\log(M)/\eps^2)$, the probability above is at most $1/3M$. By the union bound, we may estimate $\Tr(O_i \rho)$ for all $i\in [M]$ to additive error $\eps$ using these $n$ samples, with failure probability at most $1/3$. We remark that in the setting where the measurements used by the algorithm are efficient, this describes the optimal procedure.

Consider $M\leq d$. In this case, we perform the two-outcome measurement with operators $\{O_i,\mathds{1}-O_i\}$ a total of $O(\log(M)/\eps^2)$ times for each $i\in [M]$. The sample means are then within~$\eps$ of the corresponding expectation values. The procedure uses $O(M\log(M)/\eps^2)$ samples of the state, and matches the information-theoretic lower bound proved in Ref.~\cite{CCHL21-single-copy-measurements} up to a factor of $\log^2(M)$ (see the second column of Table~\ref{tab:shadow_tomography_results}).

\section{Open problems}

We conclude with some directions for future work arising from the lower bounds in Sections~\ref{sec:adaptive_lower_bounds} and~\ref{sec:classical_shadows}. In Theorem~\ref{thm:tomography_efficient_meas_const_outcome} we incur a $\textrm{polylog}(d)$ factor in the denominator of the lower bound for tomography with efficient, constant-outcome, single-copy measurements. Can this be improved? Note that such a factor also appears in the denominator of the lower bound for binary Pauli measurements in Ref.~\cite{Flammia_2012}. Is there a way to incorporate rank-dependence into the lower bounds appearing in Section~\ref{sec:adaptive_lower_bounds} for adaptive tomography with limited measurement settings? The approach we took to incorporate the dependence on the norm parameter~$B$ into the lower bounds for classical shadows does not carry over well to the setting of rank-dependent quantum tomography, since the packing we constructed has states with rank up to $d$. Finally, are there simpler information-theoretic arguments that yield the unconditional bounds obtained in Refs.~\cite{chen2022tightcertification,chen2023doesadaptivityhelpquantum}?

\section*{Acknowledgements} 
This research was supported in part by a Discovery Grant from NSERC Canada and a grant from Fujitsu Labs America.

\printbibliography

@InProceedings{yu2021sampleefficientidentity,
  author =	{Yu, Nengkun},
  title =	{{Sample Efficient Identity Testing and Independence Testing of Quantum States}},
  booktitle =	{12th Innovations in Theoretical Computer Science Conference (ITCS 2021)},
  pages =	{11:1--11:20},
  series =	{Leibniz International Proceedings in Informatics (LIPIcs)},
  ISBN =	{978-3-95977-177-1},
  ISSN =	{1868-8969},
  year =	{2021},
  volume =	{185},
  editor =	{Lee, James R.},
  publisher =	{Schloss Dagstuhl -- Leibniz-Zentrum f{\"u}r Informatik},
  address =	{Dagstuhl, Germany},
  URL =		{https://drops.dagstuhl.de/entities/document/10.4230/LIPIcs.ITCS.2021.11},
  URN =		{urn:nbn:de:0030-drops-135504},
  doi =		{10.4230/LIPIcs.ITCS.2021.11}
}

@article{Aharonov2022,
  title = {Quantum algorithmic measurement},
  volume = {13},
  ISSN = {2041-1723},
  url = {http://dx.doi.org/10.1038/s41467-021-27922-0},
  DOI = {10.1038/s41467-021-27922-0},
  number = {1},
  journal = {Nature Communications},
  publisher = {Springer Science and Business Media LLC},
  author = {Aharonov,  Dorit and Cotler,  Jordan and Qi,  Xiao-Liang},
  year = {2022},
  month = feb 
}

@inproceedings{Anshu_2022,
   series={STOC ’22},
   title={Distributed Quantum inner product estimation},
   url={http://dx.doi.org/10.1145/3519935.3519974},
   DOI={10.1145/3519935.3519974},
isbn = {9781450392648},
   booktitle={Proceedings of the 54th Annual ACM SIGACT Symposium on Theory of Computing},
   publisher={ACM},
   address={New York, NY, USA},
   author={Anshu, Anurag and Landau, Zeph and Liu, Yunchao},
   year={2022},
pages = {44–-51},
numpages = {8},
   month=jun, collection={STOC ’22},
location = {Rome, Italy},
}

@article{Flammia_2024,
   title={Quantum chi-squared tomography and mutual information testing},
   volume={8},
   ISSN={2521-327X},
   url={http://dx.doi.org/10.22331/q-2024-06-20-1381},
   DOI={10.22331/q-2024-06-20-1381},
   journal={Quantum},
   publisher={Verein zur Forderung des Open Access Publizierens in den Quantenwissenschaften},
   author={Flammia, Steven T. and O'Donnell, Ryan},
   year={2024},
   month=jun, pages={1381}
}

@article{Roberts_2017,
   title={Chaos and complexity by design},
   volume={2017},
   ISSN={1029-8479},
   url={http://dx.doi.org/10.1007/JHEP04(2017)121},
   DOI={10.1007/jhep04(2017)121},
   number={4},
   journal={Journal of High Energy Physics},
   publisher={Springer Science and Business Media LLC},
   author={Roberts, Daniel A. and Yoshida, Beni},
   year={2017},
   month=apr }

@techreport{CD10-QMA2,
  title = {Short Multi-Prover Quantum Proofs for {SAT} without Entangled Measurements},
  author = {Chen, Jing and Drucker, Andrew},
  institution = {arXiv},
  number = {arXiv:1011.0716 [quant-ph]},
  year = {2010},
  address = {http://www.arxiv.org/},
  doi = {10.48550/ARXIV.1011.0716},
  url = {https://arxiv.org/abs/1011.0716},
  copyright = {arXiv.org perpetual, non-exclusive license},
}

@article{Flammia_2012,
   title={Quantum Tomography via Compressed Sensing: Error Bounds, Sample Complexity and Efficient Estimators},
   volume={14},
   ISSN={1367-2630},
   url={http://dx.doi.org/10.1088/1367-2630/14/9/095022},
   DOI={10.1088/1367-2630/14/9/095022},
   number={9},
   journal={New Journal of Physics},
   publisher={IOP Publishing},
   author={Flammia, Steven T. and Gross, David and Liu, Yi-Kai and Eisert, Jens},
   year={2012},
   month=sep,
   pages={095022}
}

@article{Haah_2017,
  title={Sample-Optimal Tomography of Quantum States},
  author={Haah, Jeongwan and Harrow, Aram W. and Ji, Zhengfeng and Wu, Xiaodi and Yu, Nengkun},
  journal={IEEE Transactions on Information Theory},
  volume = {63},
  issue = {9}, 
  month = sep,
  year={2017},
  pages={5628--5641},
  print-issn = {0018-9448},
  e-ISSN={1557-9654},
  url={http://dx.doi.org/10.1109/TIT.2017.2719044},
  DOI={10.1109/tit.2017.2719044},
  publisher={Institute of Electrical and Electronics Engineers (IEEE)},
}

@article{Hayden2006,
  doi = {10.1007/s00220-006-1535-6},
  url = {https://doi.org/10.1007/s00220-006-1535-6},
  year = {2006},
  month = mar,
  publisher = {Springer Science and Business Media {LLC}},
  volume = {265},
  number = {1},
  pages = {95--117},
  author = {Patrick Hayden and Debbie W. Leung and Andreas Winter},
  title = {Aspects of Generic Entanglement},
  journal = {Communications in Mathematical Physics}
}

@misc{guta2018fast-arxiv,
      title={Fast state tomography with optimal error bounds}, 
      author={Madalin Guta and Jonas Kahn and Richard Kueng and Joel A. Tropp},
      year={2018},
      eprint={1809.11162},
      archivePrefix={arXiv},
      primaryClass={quant-ph}
}

@article{guta2018fast,
	title = {Fast State Tomography with Optimal Error Bounds},
	author = {Madalin Gu{\c{t}}{\u{a}} and Jonas Kahn and Richard Kueng and Joel A. Tropp},
	journal = {Journal of Physics A: Mathematical and Theoretical},
	volume = {53},
	number = {20},
	pages = {204001},
	month = apr,
	year = 2020,
	publisher = {{IOP} Publishing},
	doi = {10.1088/1751-8121/ab8111},
	url = {https://doi.org/10.1088/1751-8121/ab8111},
	arxiv = {arXiv:1809.11162 [quant-ph]},
}

@inproceedings{OW16-tomography,
  author = {O'Donnell, Ryan and Wright, John},
  title = {Efficient Quantum Tomography},
  booktitle = {Proceedings of the Forty-Eighth Annual ACM Symposium on Theory of Computing},
  pages = {899--912},
  numpages = {14},
  year = {2016},
  location = {Cambridge, MA, USA},
  series = {STOC '16},
  publisher = {Association for Computing Machinery},
  address = {New York, NY, USA},
  isbn = {9781450341325},
  url = {https://doi.org/10.1145/2897518.2897544},
  doi = {10.1145/2897518.2897544},
  arxiv = {arXiv:1508.01907 [quant-ph]},
}

@article{HKP21-it-bounds,
  title = {Information-Theoretic Bounds on Quantum Advantage in Machine Learning},
  author = {Huang, Hsin-Yuan and Kueng, Richard and Preskill, John},
  journal = {Physical Review Letters},
  volume = {126},
  issue = {19},
  pages = {190505},
  numpages = {7},
  year = {2021},
  month = may,
  publisher = {American Physical Society},
  doi = {10.1103/PhysRevLett.126.190505},
  url = {https://link.aps.org/doi/10.1103/PhysRevLett.126.190505},
  arxiv = {arXiv:2101.02464 [quant-ph]},
}

@inproceedings{BCL20-mixedness-testing,
  author = {Sebastien Bubeck and Sitan Chen and Jerry Li},
  title = {Entanglement is Necessary for Optimal Quantum Property Testing},
  booktitle = {2020 IEEE 61st Annual Symposium on Foundations of Computer Science (FOCS)},
  year = {2020},
  month = nov,
  pages = {692--703},
  publisher = {IEEE Computer Society},
  address = {Los Alamitos, CA, USA},
  issn = {},
  doi = {10.1109/FOCS46700.2020.00070},
  url = {https://doi.ieeecomputersociety.org/10.1109/FOCS46700.2020.00070},
  arxiv = {arXiv:2004.07869 [quant-ph]},
}

@article{mahler2013adaptive,
  title = {Adaptive Quantum State Tomography Improves Accuracy Quadratically},
  author = {Mahler, D. H. and Rozema, Lee A. and Darabi, Ardavan and Ferrie, Christopher and Blume-Kohout, Robin and Steinberg, Aephraim M.},
  journal = {Physical Review Letters},
  volume = {111},
  issue = {18},
  pages = {183601},
  numpages = {5},
  year = {2013},
  month = oct,
  publisher = {American Physical Society},
  doi = {10.1103/PhysRevLett.111.183601},
  url = {https://link.aps.org/doi/10.1103/PhysRevLett.111.183601}
}

@inproceedings{BO21-data-analysis, 
  author = {B\u{a}descu, Costin and O'Donnell, Ryan}, 
  title = {Improved Quantum Data Analysis},
  booktitle = {Proceedings of the 53rd Annual ACM SIGACT Symposium on Theory of Computing}, 
  pages = {1398--1411}, 
  numpages = {14}, 
  location = {Virtual, Italy}, 
  series = {STOC 2021}, 
  year = {2021}, 
  publisher = {Association for Computing Machinery},
  isbn = {9781450380539}, 
  address = {New York, NY, USA}, 
  url = {https://doi.org/10.1145/3406325.3451109}, 
  doi = {10.1145/3406325.3451109},
  arxiv = {arXiv:2011.10908 [quant-ph]},
}

@article{Aaronson_2019,
   title={Online Learning of Quantum States},
   volume={2019},
   ISSN={1742-5468},
   url={http://dx.doi.org/10.1088/1742-5468/ab3988},
   DOI={10.1088/1742-5468/ab3988},
   number={12},
   journal={Journal of Statistical Mechanics: Theory and Experiment},
   publisher={IOP Publishing},
   author={Aaronson, Scott and Chen, Xinyi and Hazan, Elad and Kale, Satyen and Nayak, Ashwin},
   year={2019},
   month=dec,
   pages={124019}
}

@book{watrous2018,
      title={The Theory of Quantum Information},
      author={John Watrous},
      year={2018},
      publisher={Cambridge University Press},
    address = {Cambridge, UK},
}

@phdthesis{wright2016,
  author       = {John Wright}, 
  title        = {How to Learn a Quantum State},
  school       = {Carnegie Mellon University},
  year         = 2016,
}

@book{fano1966,
	author = {Robert M. Fano},
	title = {Transmission of Information: a Statistical Theory of Communications},
	year = {1966},
	publisher = {MIT Press},
    address = {Cambridge, MA, USA},
}

@book{M19-random-matrix-theory,
author = "Elizabeth S. Meckes",
title = "The Random Matrix Theory of the Classical Compact Groups",
series = "Cambridge Tracts in Mathematics",
volume = "218",
publisher = "Cambridge University Press",
address = "Cambridge, UK",
doi = "10.1017/9781108303453",
online-isbn = "9781108303453",
month = jul,
year = "2019",
}

@ARTICLE{sason_2016,
  author={Sason, Igal and Verd{\'u}, Sergio},
  journal={IEEE Transactions on Information Theory}, 
  title={$f$-Divergence Inequalities}, 
  year={2016},
  volume={62},
  number={11},
  pages={5973--6006},
  doi={10.1109/TIT.2016.2603151}
}

@book{nielsen_chuang_2010, 
    place={Cambridge}, 
    title={Quantum Computation and Quantum Information},
    edition = {10th Anniversary Edition},
    DOI={10.1017/CBO9780511976667}, 
    publisher={Cambridge University Press}, 
    address = {Cambridge, UK},
    author={Nielsen, Michael A. and Chuang, Isaac L.}, 
    year={2010}
}

@techreport{yu2020sample,
      title={Sample Efficient Tomography via {Pauli} Measurements},
      author={Nengkun Yu},
  institution = {arXiv},
  number = {arXiv:2009.04610 [quant-ph]},
      year={2020},
  address = {http://www.arxiv.org/},
}

@article{struchalin_2021,
  title = {Experimental Estimation of Quantum State Properties from Classical Shadows},
  author = {Struchalin, G.I. and Zagorovskii, Ya. A. and Kovlakov, E.V. and Straupe, S.S. and Kulik, S.P.},
  journal = {PRX Quantum},
  volume = {2},
  issue = {1},
  pages = {010307},
  numpages = {11},
  year = {2021},
  month = jan,
  publisher = {American Physical Society},
  doi = {10.1103/PRXQuantum.2.010307},
  url = {https://link.aps.org/doi/10.1103/PRXQuantum.2.010307}
}

@article{Huang2020,
  doi = {10.1038/s41567-020-0932-7},
  url = {https://doi.org/10.1038/s41567-020-0932-7},
  year = {2020},
  month = jun,
  publisher = {Springer Science and Business Media {LLC}},
  volume = {16},
  number = {10},
  pages = {1050--1057},
  author = {Hsin-Yuan Huang and Richard Kueng and John Preskill},
  title = {Predicting Many Properties of a Quantum System from Very Few Measurements},
  journal = {Nature Physics}
}

@article{Aaronson2020shadow,
  author = {Aaronson, Scott},
  title = {Shadow Tomography of Quantum States},
  journal = {SIAM Journal on Computing},
  volume = {49},
  number = {5},
  pages = {STOC18-368--STOC18-394},
  year = {2020},
  doi = {10.1137/18M120275X},
  URL = {https://doi.org/10.1137/18M120275X},
}

@article{KUENG2017,
title = {Low Rank Matrix Recovery from Rank One Measurements},
journal = {Applied and Computational Harmonic Analysis},
volume = {42},
number = {1},
pages = {88--116},
year = {2017},
issn = {1063-5203},
doi = {https://doi.org/10.1016/j.acha.2015.07.007},
url = {https://www.sciencedirect.com/science/article/pii/S1063520315001037},
author = {Richard Kueng and Holger Rauhut and Ulrich Terstiege}
}

@book{vershynin_2018, 
  place={Cambridge}, 
  series={Cambridge Series in Statistical and Probabilistic Mathematics}, title={High-Dimensional Probability: An Introduction with Applications in Data Science}, 
  DOI={10.1017/9781108231596}, 
  publisher={Cambridge University Press},
  address = {Cambridge, UK},
  author={Vershynin, Roman}, 
  year={2018}, 
  collection={Cambridge Series in Statistical and Probabilistic Mathematics}
}

@article{adaptivebayesian2012,
  title = {Adaptive Bayesian Quantum Tomography},
  author = {Husz\'ar, Ferenc and Houlsby, Neil M. T.},
  journal = {Physical Review A},
  volume = {85},
  issue = {5},
  pages = {052120},
  numpages = {5},
  year = {2012},
  month = may,
  publisher = {American Physical Society},
  doi = {10.1103/PhysRevA.85.052120},
  url = {https://link.aps.org/doi/10.1103/PhysRevA.85.052120}
}

@book{Cover2005,
  doi = {10.1002/047174882x},
  url = {https://doi.org/10.1002/047174882x},
  year = {2005},
  month = apr,
  publisher = {John Wiley \& Sons, Inc.},
  address = {Hoboken, New Jersey, USA},
  author = {Thomas M. Cover and Joy A. Thomas},
  title = {Elements of Information Theory}
}

@book{Tsybakov2009,
author="Tsybakov, Alexandre B.",
title="Introduction to Nonparametric Estimation",
year="2009",
publisher="Springer New York",
address="New York, NY",
pages="77--135",
isbn="978-0-387-79052-7",
doi="10.1007/978-0-387-79052-7_2",
url="https://doi.org/10.1007/978-0-387-79052-7_2"
}

@article{Aaronson_2004,
   title={Improved Simulation of Stabilizer Circuits},
   volume={70},
   ISSN={1094-1622},
   url={http://dx.doi.org/10.1103/PhysRevA.70.052328},
   DOI={10.1103/physreva.70.052328},
   number={5},
  issue = {5},
  pages = {052328},
  numpages = {14},
   journal={Physical Review A},
   publisher={American Physical Society (APS)},
   author={Aaronson, Scott and Gottesman, Daniel},
   year={2004},
   month=nov,
}

@inproceedings{odonnell2017efficient2,
author = {O'Donnell, Ryan and Wright, John},
title = {Efficient Quantum Tomography {II}},
year = {2017},
isbn = {9781450345286},
publisher = {Association for Computing Machinery},
address = {New York, NY, USA},
url = {https://doi.org/10.1145/3055399.3055454},
doi = {10.1145/3055399.3055454},
booktitle = {Proceedings of the 49th Annual ACM SIGACT Symposium on Theory of Computing},
pages = {962--974},
numpages = {13},
location = {Montreal, Canada},
series = {STOC 2017},
}

@inproceedings{CCHL21-single-copy-measurements,
  author = {Sitan Chen and Jordan Cotler and Hsin-Yuan Huang and Jerry Li},
  title = {Exponential Separations Between Learning With and Without Quantum Memory},
  booktitle = {2021 IEEE 62nd Annual Symposium on Foundations of Computer Science (FOCS)},
  pages = {574--585},
  year = {2022},
  month = feb,
  publisher = {IEEE Computer Society},
  address = {Los Alamitos, CA, USA},
  issn = {},
  doi = {10.1109/FOCS52979.2021.00063},
  url = {https://doi.ieeecomputersociety.org/10.1109/FOCS52979.2021.00063},
  arxiv = {arXiv:2111.05881 [quant-ph]},
}

@inproceedings{ALL22-distr-ip, 
  author = {Anshu, Anurag and Landau, Zeph and Liu, Yunchao}, 
  title = {Distributed Quantum Inner Product Estimation},
  booktitle = {Proceedings of the 54th Annual ACM SIGACT Symposium on Theory of Computing}, 
  pages = {44--51}, 
  numpages = {8},
  year = {2022},
  location = {Rome, Italy}, 
  series = {STOC 2022},
  isbn = {9781450392648}, 
  publisher = {Association for Computing Machinery}, 
  address = {New York, NY, USA}, 
  url = {https://doi.org/10.1145/3519935.3519974}, 
  doi = {10.1145/3519935.3519974},
  arxiv = {arXiv:2111.03273 [quant-ph]},
}

@techreport{chen2021hierarchy,
      author={Sitan Chen and Jordan Cotler and Hsin-Yuan Huang and Jerry Li},
      title={A Hierarchy for Replica Quantum Advantage},
      institution={arXiv},
      number={arXiv:2111.05874 [quant-ph]},
      year={2021},
  address = {http://www.arxiv.org/},
  doi = {https://doi.org/10.48550/arXiv.2111.05874},
  url = {https://arxiv.org/abs/2111.05874},
}

@article{buscemi2010,
  author={Buscemi, Francesco and Datta, Nilanjana},

  journal={IEEE Transactions on Information Theory}, 

  title={The Quantum Capacity of Channels With Arbitrarily Correlated Noise}, 

  year={2010},

  volume={56},

  number={3},

  pages={1447-1460},

  doi={10.1109/TIT.2009.2039166}}

@INPROCEEDINGS {AE07-state-designs,
author = {Andris Ambainis and Joseph Emerson},
booktitle = {2007 22nd Annual IEEE Conference on Computational Complexity},
title = {Quantum $t$-designs: $t$-wise Independence in the Quantum World},
year = {2007},
volume = {},
issn = {1093-0159},
pages = {129--140},
doi = {10.1109/CCC.2007.26},
url = {https://doi.ieeecomputersociety.org/10.1109/CCC.2007.26},
publisher = {IEEE Computer Society},
address = {Los Alamitos, CA, USA},
month = jun,
}

@techreport{kueng2015stabilizer3design,
  author = {Kueng, Richard and Gross, David},
  title = {Qubit Stabilizer States are Complex Projective 3-Designs},
  institution = {arXiv},
  number = {arXiv:1510.02767 [quant-ph]},
  year = {2015},
  address = {http://www.arxiv.org/},
  doi = {10.48550/ARXIV.1510.02767},
  url = {https://arxiv.org/abs/1510.02767},
}

@article{webb2017clifford,
author = {Webb, Zak},
title = {The Clifford Group Forms a Unitary 3-Design},
journal = {Quantum Information \& Computation},
volume = {16},
number = {15-16},
month = nov,
year = {2016},
pages = {1379--1400},
numpages = {22},
issue_date = {November 2016},
publisher = {Rinton Press, Incorporated},
address = {Paramus, NJ},
issn = {1533-7146}
}

@article{zhu2017clifford,
  title = {Multiqubit Clifford groups are unitary 3-designs},
  author = {Zhu, Huangjun},
  journal = {Physical Review A},
  volume = {96},
  issue = {6},
  pages = {062336},
  numpages = {7},
  year = {2017},
  month = dec,
  publisher = {American Physical Society},
  doi = {10.1103/PhysRevA.96.062336},
  url = {https://link.aps.org/doi/10.1103/PhysRevA.96.062336}
}

@INPROCEEDINGS{vandenberg2021simple,
  author={Van Den Berg, Ewout},
  booktitle={2021 IEEE International Conference on Quantum Computing and Engineering (QCE)}, 
  title={A simple method for sampling random Clifford operators}, 
  year={2021},
  volume={},
  number={},
  pages={54--59},
  publisher = {IEEE Computer Society},
  address = {Los Alamitos, CA, USA},
  doi={10.1109/QCE52317.2021.00021}}

@techreport{chen2022tightcertification,
  author = {Chen, Sitan and Huang, Brice and Li, Jerry and Liu, Allen}, 
  title = {Tight Bounds for Quantum State Certification with Incoherent Measurements},
  institution = {arXiv},
  number={arXiv:2204.07155 [quant-ph]},
  year = {2022},
  doi = {10.48550/ARXIV.2204.07155},
  url = {https://arxiv.org/abs/2204.07155},
}

@INPROCEEDINGS{chen2023doesadaptivityhelpquantum,
author={Sitan Chen and Brice Huang and Jerry Li and Allen Liu and Mark Sellke},
booktitle = {2023 IEEE 64th Annual Symposium on Foundations of Computer Science (FOCS)},
title = {When Does Adaptivity Help for Quantum State Learning?},
year = {2023},
volume = {},
issn = {},
pages = {391--404},
doi = {10.1109/FOCS57990.2023.00029},
url = {https://doi.ieeecomputersociety.org/10.1109/FOCS57990.2023.00029},
publisher = {IEEE Computer Society},
address = {Los Alamitos, CA, USA},
month = {nov}
}

@mastersthesis{Lowe21-tomography,
  author = {Lowe, Angus},
  title = {Learning Quantum States Without Entangled Measurements},
  type = {{M.Math.}\ Thesis},
  school = {University of Waterloo},
  address = {Waterloo, Ontario, Canada},
  month = oct,
  year = {2021},
  url={http://hdl.handle.net/10012/17663},
  note = {},
}

@misc{LN22-tomography,
  author = {Angus Lowe and Ashwin Nayak},
  title = {Improved Lower Bounds for Learning Quantum States with Unentangled Measurements}, 
  howpublished = {Presented at the 25th Annual Conference on Quantum Information Processing},
  month = mar # " 7--12,",
  year = {2022},
  location = {Pasadena, CA, USA},
}

@techreport{HRS05-coset-states,
  author = {Hallgren, Sean and Roetteler, Martin and Sen, Pranab},
  title = {Limitations of Quantum Coset States for Graph Isomorphism},
  institution = {arXiv},
  number = {arXiv:quant-ph/0511148},
  year = {2005},
  address = {http://www.arxiv.org/},
  doi = {10.48550/ARXIV.QUANT-PH/0511148},
  url = {https://arxiv.org/abs/quant-ph/0511148},
}

@article{HMRRS10-coset-states, 
  author = {Hallgren, Sean and Moore, Cristopher and R\"{o}tteler, Martin and Russell, Alexander and Sen, Pranab}, 
  title = {Limitations of Quantum Coset States for Graph Isomorphism},
  journal = {Journal of the ACM}, 
  volume = {57}, 
  number = {6}, 
  month = nov,
  year = {2010},
  articleno = {34}, 
  numpages = {33},
  pages = {1--33},
  issue_date = {October 2010}, 
  publisher = {Association for Computing Machinery}, 
  address = {New York, NY, USA}, 
  issn = {0004-5411}, 
  url = {https://doi.org/10.1145/1857914.1857918}, 
  doi = {10.1145/1857914.1857918} 
}

@article{AW02,
  author = {Ahlswede, Rudolph and Winter, Andreas},
  title={Strong Converse for Identification via Quantum Channels}, 
  journal = {IEEE Transactions on Information Theory}, 
  year = {2002},
  volume = {48},
  number = {3},
  pages = {569--579},
  doi={10.1109/18.985947},
  arxiv = {arXiv:quant-ph/0012127},
}

\appendix

\section{Haar integrals}\label{sec:haar_integrals}

The Haar measure $\mu$ is the unique unitarily invariant probability measure on the space of unitary operators, $\mathsf{U}(d)$. Using this measure, one may define channels $\Phi_k : (\mathbb{C}^{d\times d})^{\otimes k} \to (\mathbb{C}^{d\times d})^{\otimes k}$ of the form
\begin{align}
\Phi_{k}(X) = \int_{\mathsf{U}(d)} U^{\otimes k} X (U^{\dag})^{\otimes k} \mathrm{d}\mu(U),
\end{align}
which are referred to as ``$k$-fold twirl" operations. In the rest of this section, we evaluate this channel explicitly in the case where the operator $X$ is a tensor product of orthogonal projectors onto subspaces of $\mathbb{C}^d$. Though it is well-known that one can in principle compute such expressions using Weingarten functions (see, e.g., {\cite{Roberts_2017}}) we elect to perform these calculations in a more self-contained way since it is not too much extra work to do so.
Following the presentation in Ref.~\cite{watrous2018}, we make use of an important result on the structure of permutation-invariant operators. Recall from Section~\ref{sec:preliminaries} that $S_k$ is the symmetric group on $\{1,\dots,k\}$ and $W_\pi$ is the operator on $(\mathbb{C}^{d\times d})^{\otimes k}$ that permutes the $k$ tensor factors according to the permutation $\pi\in S_k$. 
\begin{theorem}[Theorem 7.15 in Ref.~\cite{watrous2018}]\label{thm:perm_inv_op}
Let $k > 0$ be a positive integer and $X \in (\mathbb{C}^{d\times d})^{\otimes k}$ be an operator. The following are equivalent:
 \begin{enumerate}
 \item $[X,U^{\otimes k}]=0 \ \forall U\in \mathsf{U}(d)$.
 \item $X = \sum_{\pi\in S_k}v(\pi)W_\pi$ for some choice of $v \in \mathbb{C}^{|S_k|}$.
 \end{enumerate}
\end{theorem}
Since $\Phi_k(X)$ satisfies the first condition, we can apply the theorem to write the output of the channel as a linear combination of permutation operators. This helps us evaluate the Haar integrals which arise in this work.
\begin{proposition}\label{prop:haar_expecs}
Let $d>1$ be a positive integer, $Q\in\mathsf{Psd}(d)$ a rank-$r$ orthogonal projection operator, and $\bm{U}\in\mathsf{U}(d)$ a Haar-random unitary operator. It holds that
\begin{align*}
\expct \bm{U}Q\bm{U}^\dag = \frac{r\mathds{1}}{d}.
\end{align*}
\end{proposition}
\begin{proof}
We can write the expectation as
\begin{align*}
\int_{\mathsf{U}(d)} UQU^\dag\mathrm{d}\mu(U) = \Phi_1(Q).
\end{align*}
By Theorem~\ref{thm:perm_inv_op} we have
\begin{align*}
\mathbb{E}\  \bm{U} Q \bm{U}^\dag = \kappa \mathds{1}
\end{align*}
where $\kappa\in \mathbb{C}$ is some coefficient depending on $Q$. Recalling that $Q$ is a rank-$r$ orthogonal projection operator, taking the trace of both sides and solving for $\kappa$ yields $\kappa=r/d$.
\end{proof}

\begin{proposition}\label{prop:haar_expecs_2}
Let $d>1$ be a positive integer. Let $\Pi_1,\Pi_2\in\mathsf{Psd}(d)$ be orthogonal projection operators of rank $r_1,r_2$, respectively, such that the image of~$\Pi_1$ is contained in that of~$\Pi_2$.
% which are diagonal in the same basis and where $r_1\leq r_2$. 
For $\bm{U}\in\mathsf{U}(d)$ a Haar-random unitary operator it holds that
\begin{align*}
\expct\ \bm{U}^{\otimes 2}(\Pi_1\otimes \Pi_2)(\bm{U}^\dag)^{\otimes 2} = \frac{r_1}{d(d^2-1)}\left[(r_2d-1)\mathds{1}+(d-r_2)W\right]
\end{align*}
where $W$ is the swap operator acting on $(\mathbb{C}^{d})^{\otimes 2}$.
\end{proposition}
\begin{proof} 
We can write the expectation as
\begin{align*}
\int_{\mathsf{U}(d)} U^{\otimes 2}(\Pi_1\otimes \Pi_2)(U^\dag)^{\otimes 2}\mathrm{d}\mu(U) = \Phi_2(\Pi_1\otimes \Pi_2).
\end{align*}
By Theorem~\ref{thm:perm_inv_op} we have
\begin{align*}
\expct\ \bm{U}^{\otimes 2}(\Pi_1\otimes \Pi_2)(\bm{U}^\dag)^{\otimes 2} = \alpha \mathds{1}\otimes \mathds{1} + \beta W
\end{align*}
where $W$ is the swap operator and $\alpha,\beta\in \mathbb{C}$ are some coefficients depending on $Q$. Left-multiplying by $\mathds{1}\otimes \mathds{1}$ or $W$ and taking the trace of both sides yields
\begin{align*}
\Tr(\Pi_1\otimes \Pi_2) = r_1r_2 = \alpha d^2 + \beta d, \quad \Tr(W(\Pi_1\otimes \Pi_2)) = r_1 = \alpha d + \beta d^2 \enspace,
\end{align*}
as~$\Pi_1 \Pi_2 = \Pi_1$. This allows us to solve for $\alpha,\beta$:
\begin{align}\label{eqn:alpha_beta_coeffs}
\alpha = \frac{r_1(r_2d-1)}{d(d^2-1)}, \quad \beta = \frac{r_1(d-r_2)}{d(d^2-1)}.
\end{align}
This concludes the proof of the proposition.
\end{proof}

We also make use of the expectations of operators of the following form.
\begin{proposition}
\label{prop-haar-expct-3}
Let~$d \ge 1$ and~$\bU$ be a Haar-random unitary operator over~$\complex^d$. For any~$i,j \in [d]$, we have
\begin{align*}
\expect{\bU} ~ \bU \ket{i} \tensor \bU \ket{j} & = 0 \enspace, \\
\expect{\bU} ~ \bra{i} \bU^\adjoint \tensor \bra{j} \bU^\adjoint & = 0 \enspace, \\
\expect{\bU} ~ \bU \ket{i} \tensor \bra{j} \bU^\adjoint & = \frac{ \updelta_{ij} }{d} \sum_{k = 1}^d \ket{k} \tensor \bra{k} \enspace, \qquad \text{and} \\
\expect{\bU} ~ \bra{j} \bU^\adjoint \tensor \bU \ket{i} & = \frac{ \updelta_{ij}}{d} \sum_{k = 1}^d \bra{k} \tensor \ket{k} \enspace,
\end{align*}
where~$\updelta_{ij} = 1$ if~$i = j$ and~$0$ otherwise.
\end{proposition}
\begin{proof}
The second identity follows from the first by taking the adjoint, which commutes with taking the expectation over~$\bU$. Similarly, the fourth identity follows from the third by conjugating with the swap operator on~$\complex^d \tensor \complex^d$.

The first identity follows from the invariance of the Haar measure under multiplication by~$\complexi \, \id$; we have
\[
\expect{\bU} ~ \bU \ket{i} \tensor \bU \ket{j} =  \complexi^2 \, \expect{\bU} ~ \bU \ket{i} \tensor \bU \ket{j} = 0 \enspace.
\]
Similarly, if~$i \neq j$, then by the invariance of the Haar-measure under multiplication on the right by the unitary operator~$\id - 2 \density{j}$, the third identity holds.

Let~$A$ be the left hand side of the third identity when~$i = j$. Then~$\bra{k} A \ket{l} = 0$ if~$k \neq l$, by the invariance of the Haar-measure under multiplication on the right by the operator~$\id - 2 \density{l}$. Furthermore, $ \bra{k} A \ket{k} = \expct_\bU \abs{ \bra{k} \bU \ket{i}}^2 = 1/d$, by the invariance of the Haar measure under permutations of the standard basis elements.
\end{proof}

\section{Algorithms for quantum tomography}
\label{sec:upper_bounds_tomography}

\subsection{Tomography with entangled measurements}
In the entangled measurement model, it has been shown by O'Donnell and Wright~\cite{OW16-tomography} and Haah et al.~\cite{Haah_2017} that $O(d^2/\eps^2)$ copies of the state suffice to estimate it to $\eps$-accuracy in trace distance with high probability\footnote{Originally, the upper bound presented in Haah et al.~\cite{Haah_2017} had an additional factor of $\log(d/\eps)$, which was subsequently removed in the thesis of Wright~\cite{wright2016}.}. At the same time, a matching lower bound was also shown in~\cite{Haah_2017}. So the sample complexity of tomography in the entangled measurement setting is known up to a constant factor, for constant probability of success. A full description of these algorithms is outside the scope of this work, requiring ideas from representation theory and in particular the relationship between certain representations on $(\mathbb{C}^{d})^{\otimes n}$. We refer the interested reader to Chapters~2 and~5 of Wright's PhD thesis~\cite{wright2016}.

\subsection{Tomography with random basis measurements}\label{sec:single_copy_tomog_upper_bound}

For completeness we describe an algorithm which achieves a sample complexity of $O(d^3/\eps^2)$ for $\eps$-accurate tomography (in trace distance) using efficiently implementable, nonadaptive measurements. The analysis we present is due to Wright~\cite[Section~5.1]{wright2016}, with minor differences. We also point out that measurement based on a state~$2$-design suffices. These may be derived from a \emph{spherical~$4$-design\/} or a unitary~2-design. 

An algorithm for the bounded-rank case follows from Ref.~\cite[Theorem~2]{KUENG2017}. Haah \textit{et al.\/} sketch the details of this algorithm in Ref.~\cite[Section~II.A]{Haah_2017}. They invoke an ``operator Chernoff bound'' due to Ahlswede and Winter~\cite{AW02} to conclude that the sample average of the outer product of $m$ i.i.d.\ standard normal vectors~$\ket{\bpsi_i} \in \complex^d$ with~$\ket{\bpsi_i} \sim \normal(0, \id)$ approximates the identity operator~$\id$. Formally, we have
\begin{align}
\label{eq-approx-id}
\norm{ \frac{1}{m} \sum_{i = 1}^m \density{\bpsi_i} - \id } \le \alpha \enspace,
\end{align}
for a constant~$\alpha > 0$, with probability at least~$3/4$, provided~$m \in \Omega( d (\ln d)/ \alpha^2)$. This leads to a sample complexity of~$O(r^2 d / \eps^2 )$ for~$r \ge \ln d$, and~$O(rd (\ln d)/\eps^2 )$ for~$r \le \ln d$. A stronger tail inequality~\cite[Theorem~4.6.1]{vershynin_2018} guarantees that Eq.~\eqref{eq-approx-id} holds for a suitable constant~$\alpha$, with probability at least~$1 - 2 \exp(-m)$ as long as~$m \ge d$. This gives us the optimal bound of~$O(r^2 d / \eps^2)$ on the sample complexity of the algorithm. Gu{\c{t}}{\u{a}}, Kahn, Kueng, and Tropp~\cite[Theorem~2]{guta2018fast} give a different algorithm that also achieves the optimal sample complexity.

Let $\rho\in\mathsf{D}(d)$ be the state to be learned, and $\{\ket{j}\}_{j=1}^d$ be the standard basis. Consider sampling a random unitary operator $\bm{U}$ comprising a unitary 2-design and then performing the basis measurement corresponding to the measurement operators $\{\bm{U}\ket{j}\bra{j}\bm{U}^\dag\}_{j=1}^d$, obtaining outcome $\bm{j}$. Suppose we do this on $n$ separate copies of the state, resulting in iid random variables $(\bm{U}_1,\bm{j}_1),\dots,(\bm{U}_n,\bm{j}_n)$ where $\bm{U}_i$ is the $i^\text{th}$ random unitary operator and $\bm{j}_i$ is the outcome from the $i^\text{th}$ measurement. Define $\hat{\rho}(U,j)\eqdef (d+1)U\ket{j}\bra{j}U^\dag - \mathds{1}$ for~$U\in\mathsf{U}(d)$ and $j\in [d]$.
\begin{proposition}\label{prop:unbiased_estimator_2_design}
It holds that
\begin{align*}
\expct \hat{\rho}(\bm{U},\bm{j}) = \rho.
\end{align*}
\end{proposition}
\begin{proof}
Let $p_{\bm{U}}$ denote the distribution of $\bm{U}$ and $p_{\bm{j}|U}(j)$ the probability of obtaining outcome $j$ given that $U$ is drawn. We have
\begin{align}\label{eq:eqn_1040}
\expct \bm{U}\ket{\bm{j}}\bra{\bm{j}}\bm{U}^\dag &= \sum_{j=1}^d\expect{\bm{U}\sim p_{\bm{U}}}p_{\bm{j}|\bm{U}}(j)\ \bm{U}\ket{j}\bra{j}\bm{U}^\dag\nonumber\\
&=\sum_{j=1}^d\expect{\bm{U}\sim p_{\bm{U}}} \bra{j}\bm{U}^\dag \rho \, \bm{U} \ket{j} \bm{U}\ket{j}\bra{j}\bm{U}^\dag.
\end{align}
Consider the~$j$th term in the sum above. We may write that term equivalently as
\begin{align}\label{eq:eqn_partial_trace}
\expect{\bm{U}\sim p_{\bm{U}}}\Tr_2\left((\bm{U}\ket{j}\bra{j}\bm{U}^\dag)^{\otimes 2}(\mathds{1}\otimes \rho)\right) = \Tr_2\left(\expect{\bm{V}\sim \text{Haar}} (\bm{V}\ket{j}\bra{j}\bm{V}^\dag)^{\otimes 2}(\mathds{1}\otimes \rho)\right)
\end{align}
where the equality follows from linearity of trace and the choice of~$\bm{U}$ as a 2-design. Note that it suffices that the measurement operators be derived from a state~$2$-design (see, e.g., Ref.~\cite{AE07-state-designs}). Proposition~\ref{prop:haar_expecs_2}  gives an explicit solution to the Haar integral inside the partial trace for the general case of a rank-$r$ projector rather than $\ket{j}\bra{j}$. Taking $r=1$, we find that
\begin{align*}
\expect{\bm{V}\sim \text{Haar}} (\bm{V}\ket{j}\bra{j}\bm{V}^\dag)^{\otimes 2} = \frac{1}{d(d+1)}\left[\mathds{1}\otimes\mathds{1}+W\right].
\end{align*}
Substituting into the right-hand side of Eq.~\eqref{eq:eqn_partial_trace} and making use of the identities $\Tr_2(W(\mathds{1}\otimes \rho))=\rho$ and $\Tr(\rho)=1$ we find that it is equal to $\frac{1}{d(d+1)}\left(\mathds{1}+\rho\right)$. Using the property that this holds for any $j\in [d]$ and substituting into Eq.~\eqref{eq:eqn_1040} we obtain the relation $\expct \bm{U}\ket{\bm{j}}\bra{\bm{j}}\bm{U}^\dag = \frac{1}{d+1}\left(\mathds{1}+\rho\right)$. The proposition then follows from the definition of $\hat{\rho}(\bm{U},\bm{j})$.
\end{proof}
In other words, $\hat{\rho}(\bm{U},\bm{j})$ is an unbiased estimator of $\rho$. We take the empirical average of the $n$ independent samples of this estimator $\frac{1}{n}\sum_{i=1}^n\hat{\rho}(\bm{U}_i,\bm{j}_i)$ which we obtained by measuring $n$ separate copies of the state. Then the squared distance between the estimator and the true state in terms of the metric induced by the Frobenius norm is
\begin{align*}
\expct\norm{\frac{1}{n}\sum_{i=1}^n\hat{\rho}(\bm{U}_i,\bm{j}_i) - \rho}_{\mathrm{F}}^2 &= \frac{1}{n^2}\expct\norm{\sum_{i=1}^n\left(\hat{\rho}(\bm{U}_i,\bm{j}_i) - \rho\right)}_{\mathrm{F}}^2\\
&= \frac{1}{n^2} \Tr\left(\expct \left[\sum_{i=1}^n (\hat{\rho}(\bm{U}_i,\bm{j}_i) - \rho)\right]^2\right).
\end{align*}
It is straightforward to show that for a sum of $n$ mean-zero, independent random matrices $\bm{A}_i$ it holds that $\expct\left[\sum_{i=1}^n\bm{A}_i\right]^2=\sum_{i=1}^n \expct\bm{A}_i^2$, which entails that the right-hand side of the above is
\begin{align*}
\frac{1}{n^2} \sum_{i=1}^n\Tr\left(\expct (\hat{\rho}(\bm{U}_i,\bm{j}_i) - \rho)^2\right) &= \frac{1}{n^2}\sum_{i=1}^n \left(\expct \Tr(\hat{\rho}(\bm{U}_i,\bm{j}_i)^2) - \Tr(\rho^2)\right)\\
&\leq \frac{1}{n^2} \sum_{i=1}^n\expct\Tr(\hat{\rho}(\bm{U}_i,\bm{j}_i)^2)\\
&= \frac{d^2+d-1}{n}
\end{align*}
where the inequality used $\Tr(\rho^2)\geq 0$ and the final line comes from the following calculation. For a Hermitian matrix $A$, we have $\Tr(A^2)=\sum_{i=1}^d\lambda_i(A)^2$. In our case, all eigenvalues of the operator $(d+1)U\ket{j}\bra{j}U^\dag - \mathds{1}$ except one are~$-1$, and one eigenvalue is~$d$. Using the matrix inequality $\norm{\cdot}_1\leq \sqrt{d}\norm{\cdot}_{\mathrm{F}}$, we obtain the inequality
\begin{align*}
\expct\norm{\frac{1}{n}\sum_{i=1}^n\hat{\rho}(\bm{U}_i,\bm{j}_i) - \rho}_1^2\leq \frac{d(d^2+d-1)}{n}.
\end{align*}
Substituting $n\in O(d^3/\eps^2)$ gives us the desired upper bound on error in expectation. We can achieve error at most~$\eps$ with high (constant) probability using Markov's inequality, with a constant factor increase in the number of samples.

\subsection{Tomography with binary Pauli measurements}

In the setting of binary Pauli measurements there exists perhaps the most straightforward tomography algorithm, to the point where its $O(d^4/\eps^2)$ sample complexity is folklore. However, since we show that this is the information-theoretically optimal algorithm for a class of nonadaptive measurement scenarios, it may be worth reviewing. The general $q$-qubit Pauli matrices are the various Hermitian and unitary $q$-fold tensor products of the set of single-qubit Pauli matrices $\{\mathds{1},\sigma_x,\sigma_y,\sigma_z\}\subset \mathbb{C}^{2\times 2}$. This means that there are $4^q = d^2$ different $q$-qubit Pauli matrices $\mathcal{P}_d = \{P_1,\dots,P_{d^2}\}$, where we let $d=2^q$. These operators form an orthogonal basis for the set of $d$-dimensional Hermitian matrices $\mathsf{H}(d)$ so that an arbitrary $\rho\in\mathsf{D}(d)$ can be written
\begin{align*}
\rho = \frac{1}{d}\sum_{i=1}^{d^2}\Tr(P_i\rho)P_i.
\end{align*}
The straightforward algorithm here is then to estimate each of the coefficients $\Tr(P_i\rho)$ with sufficient accuracy, which will serve as a complete description of the estimate of $\rho$. Consider the $d^2$ POVMs $\mathcal{M}_i$ with corresponding measurement operators $\{\frac{1}{2}(\mathds{1}\pm P_i)\}$ for each $i\in [d^2]$, with possible outcomes $\bm{z}_i\in \{\pm 1\}$ defined in the obvious way. Then $\bm{z}_i$ is an unbiased estimator for the $i^\text{th}$ Pauli coefficient, and performing this measurement $s\in\mathbb{Z}_+$ times results in iid random variables $\{\bm{z}_{i,j}\}_{j=1}^s$. Let us then take the empirical average of the $s$ samples corresponding to the $i^\text{th}$ Pauli measurement $\bm{\mu}_i \eqdef \frac{1}{s}\sum_{j=1}^s\bm{z}_{i,j}$, for each $i\in [d^2]$, which requires a total of $sd^2$ measurements on separate copies of $\rho$. We then consider our estimate of the state to be $\hat{\bm{\rho}}\eqdef \frac{1}{d}\sum_{i=1}^{d^2} \bm{\mu}_i P_i$, which clearly satisfies $\expct \hat{\bm{\rho}}=\rho$. We may then compute
\begin{align*}
\expct \norm{\hat{\bm{\rho}}-\rho}_{\mathrm{F}}^2 &= \frac{1}{d}\sum_{i=1}^{d^2}\expct {\lvert\bm{\mu}_i - \Tr(P_i\rho)\rvert}^2\\
&= \frac{1}{d}\sum_{i=1}^{d^2} \Var[\bm{\mu}_i]\\
&= \frac{1}{ds^2}\sum_{i=1}^{d^2}\sum_{j=1}^s \Var[\bm{z}_{i,j}]\\
&\leq \frac{d}{s}
\end{align*}
where in the third line we used the property $\Var[a\bm{x}]=a^2\Var[\bm{x}]$ for a random variable $\bm{x}$, as well as the fact that the variance is additive for independent random variables. The final line follows since $|\bm{z}_{i,j}|=1$. Using the inequality $\norm{\cdot}_1\leq \sqrt{d}\norm{\cdot}_{\mathrm{F}}$, we find for $s=d^2/\eps^2$, it holds that $\expct \norm{\hat{\bm{\rho}}-\rho}_1\leq \eps$. We can once again convert this statement about convergence in expectation to convergence with high probability using Markov's inequality, which leads to the conclusion that $\eps$-accurate tomography in trace distance is achievable using at most $sd^2 = d^4/\eps^2$ binary Pauli measurements on separate copies of $\rho$.

\section{Measurements with polynomial-size circuits}
\label{sec:solovay_kitaev}

Theorems~\ref{thm:tomography_efficient_meas}, \ref{thm:tomography_efficient_meas_const_outcome}, and~\ref{thm:unentangled_shadow_tomography} give lower bounds for quantum learning tasks in the  setting of adaptive measurements, when they are drawn from a finite set of possible measurements. These results thus also limit the power of adaptivity using measurements that can be implemented with polynomial-size circuits. This includes efficiently implementable measurements, i.e., measurements whose circuits are also uniformly generated. In this section, we explain what it means for a family of measurements to have polynomial-size circuits. Fix a (possibly infinite) universal gate set $\mathcal{G}$ consisting of constant-arity gates (e.g., one- and two-qubit gates).
\begin{definition}[Measurements with polynomial-size circuits, constant number of outcomes]
\label{def:efficient_meas}
For any~$q \ge 1$, suppose $\cA_q$ is a collection of measurements acting on $q$-qubit quantum states. We say the family of measurements~$(\cA_q : q \ge 1)$ has \emph{polynomial-size\/} if there exist polynomials $p_1,p_2$ such that for each $q$ and measurement $\mathcal{M} \in \cA_q$ there is a quantum circuit on $q + p_1(q)$ qubits with at most $p_2(q)$ gates from $\mathcal{G}$ that implements $\mathcal{M}$. I.e., the measurement $\mathcal{M}$ has the action
\begin{align*}
\mathcal{M} :\rho \mapsto \sum_{y\in\{0,1\}^S} \bra{y}\Tr_{[\ell]\backslash S}\left(U (\rho\otimes \ket{\overline{0}} \! \bra{\overline{0}})U^\dag\right)\ket{y} \; \ket{y} \! \bra{y}
\end{align*}
for any state $\rho\in \mathsf{D}(2^q)$, where $S\subseteq [\ell]$, $\ell \eqdef q+p_1(q)$, $\ket{\overline{0}}\in (\mathbb{C}^2)^{\otimes p_1(q)}$ is the all-zero state for $p_1(q)$ ancilla qubits and $U \in\mathsf{U}(2^{q + p_1(q)} )$ is the unitary operator given by the composition of the gates in the circuit.

We say that the measurements in the family have a constant number of outcomes if there is a positive integer~$r$ such that for all measurements~$\mathcal{M} \in \cA_q$ and for all~$q \ge 1$, we have~$\rank(\mathcal{M}) \le r$.
\end{definition}
In the case where $\mathcal{G}$ is finite, by a counting argument we may verify that the number of distinct measurements~$m$ in~$\cA_q$ for any family with polynomial size is at most~$\text{poly}(q)^{\text{poly}(q)}$ which is in~$\exp(o(d))$, where~$d \eqdef 2^q$ is the dimension of the system. It follows immediately from Theorem~\ref{thm:tomography_efficient_meas} and this bound that $\Omega(d^3/\eps^2)$ single-copy, possibly adaptive, efficient measurements are necessary to perform tomography. (Note that \emph{efficient\/} measurements are also required to be \emph{uniformly generated\/}, in addition to having polynomial-size circuits.) Similarly, we may infer a bound for shadow tomography using only efficient single-copy measurements from Theorem~\ref{thm:unentangled_shadow_tomography} (see the remark after the theorem). 

We may extend this reasoning to the case where $\mathcal{G}$ has infinite cardinality, but consists of gates of constant arity --- for example, when all single-qubit gates are included in the set. This comes at the cost of the loss of a multiplicative factor of at most $\text{polylog}(1/\eps)$ in the lower bounds. This is accomplished by an application of the Solovay-Kitaev theorem and adjusting the general argument we have been using to prove the bounds. We replace each measurement with a suitably accurate approximation with a circuit over a finite gate set, and show that the approximation results in at most a small constant deviation from the original distribution over measurement outcomes. In the sequel, we refer to the case where a learner performs measurements with circuits over the gate set $\mathcal{G}$ as the \textit{original\/} strategy. Fix any \textit{finite\/} universal gate set $\mathcal{G}^\prime$ that contains the inverses of all the gates in it. 
\begin{proposition}\label{prop:efficient_gates}
Let~$d \eqdef 2^q$ for some~$q \ge 1$.
Suppose a learner performs $n\in O(d^3/\eps^2)$ adaptive measurements on single copies of quantum states in $\mathsf{D}(d)$, where each measurement can be implemented with a circuit of size at most a polynomial~$t$ in~$q$ using an infinite set $\mathcal{G}$ of gates with constant arity. There is an adaptive measurement strategy consisting of single-copy measurements with circuits of size of order~$qt ( \log q + \log(1/\eps))$ over~$\mathcal{G}'$ such that for any state~$ \rho \in \mathsf{D}(d)$, the distribution over the~$n$ measurement outcomes obtained from measuring~$\rho$ is $0.01$-close in total variation distance to the corresponding distribution obtained with the original strategy.
\end{proposition}
\begin{proof}
Suppose learner performs measurements with circuits of size at most~$t$ over the gate set $\mathcal{G}$ in the original strategy. Suppose that this learner obtains the outcomes $\bm{y}_1,\dots,\bm{y}_n$ using the original strategy, and consider the $i^\text{th}$ measurement in the sequence $\mathcal{M}_i^{y_{<i}}:\mathsf{D}(2^q)\to\mathsf{D}(2^{u(q)})$ for some fixed sequence of previous outcomes $y_{<i}$, and polynomial~$u(q)$. By the Solovay-Kitaev Theorem, for any $\delta\in (0,1)$ there is a measurement $\Phi_i^{y_{<i}}:\mathsf{D}(2^q)\to\mathsf{D}(2^{u(q)})$ which can be implemented using circuits of size $t' \eqdef t\cdot \text{polylog}(t/\delta)$ gates from $\mathcal{G}^\prime$ and which satisfies
\begin{align}\label{eq:application_of_solovay_kitaev}
\norm{\mathcal{M}_i^{y_{<i}}-\Phi_i^{y_{<i}}}_{\diamond}\leq 2\delta \enspace,
\end{align}
where $\norm{\cdot}_{\diamond}$ is the completely bounded trace norm (some times also called the diamond norm). In other words, for any fixed state $\rho\in\mathsf{D}(2^q)$ the total variation distance between the distributions over outcomes obtained by measuring~$\rho$ according to the two measurements is at most $\delta$. 

Suppose the learner adopts the \textit{modified strategy\/} given by the measurements $\Phi_i^{y_{<i}}$ for every $i\in [n]$, given the previously observed outcomes $y_{<i}$. We show by induction that the deviation of the resulting distribution from that of the original strategy grows linearly with $n$. For each $i\in [n]$, let $\bm{y}^\prime_i$ denote the measurement outcome from the $i^\text{th}$ measurement using the modified strategy. Note that these random variables have the same set of possible outcomes, which are bit-strings of length at most $\text{poly}(q)$. Define the corresponding conditional distributions~$p$ and~$\phi$ over outcomes as
\begin{align*}
p(y_k|y_{<k}) \eqdef \Pr\left[\bm{y}_k=y_k|\bm{y}_{<k}=y_{<k}\right],\qquad \phi(y_k|y_{<k}) = \Pr\left[\bm{y}_k^\prime=y_k|\bm{y}_{<k}^\prime=y_{<k}\right]
\end{align*}
as well as marginal probabilities $p(y_{<k})$, $\phi(y_{<k})$, for each~$k \in [1,n]$. Let us also define the notation $y_{\leq k} \eqdef y_{<k+1}$. For the first measurement outcome, the total variation distance between the two distributions is $\frac{1}{2}\sum_{y_1}|p(y_1)-\phi(y_1)|\leq \delta$, using Eq.~\eqref{eq:application_of_solovay_kitaev}. Now suppose that
\begin{align*}
\sum_{y_{<k}}|p(y_{<k})-\phi(y_{<k})| \leq 2(k-1)\delta.
\end{align*}
for some $k>1$. Then we have
\begin{align*}
\sum_{y_{\leq k}}|p(y_{\leq k})-\phi(y_{\leq k})| &= \sum_{y_k}\sum_{y_{<k}}\left|p(y_k|y_{< k})p(y_{<k})-\phi(y_k|y_{<k})\phi(y_{<k})\right|\\
&\leq \sum_{y_k}\sum_{y_{<k}}|p(y_k|y_{<k})p(y_{<k}) - \phi(y_k|y_{<k})p(y_{<k})|\\ &\quad\quad\quad + \sum_{y_k}\sum_{y_{<k}}|\phi(y_k|y_{<k})p(y_{<k}) - \phi(y_k|y_{<k})\phi(y_{<k})|\\
&= \sum_{y_{<k}}p(y_{<k})\sum_{y_k}|p(y_k|y_{<k})-\phi(y_k|y_{<k})| + \sum_{y_{<k}}|p(y_{<k})-\phi(y_{<k})|\\
&\leq 2\delta + 2(k-1)\delta\\
&= 2k\delta.
\end{align*}
Hence, the total variation distance between the two distributions corresponding to all $n$ outcomes is at most $n\delta$. Taking $\delta = 1/(100n)$ ensures that the total error from the modified strategy is at most $0.01$. Moreover, all measurements in the modified strategy are implemented with circuits of size at most $t\cdot \text{polylog}(100nt)$ over the finite gate set $\mathcal{G}^\prime$. Since $t(q)$ is a polynomial in~$q$ and $n\in O(d^3/\eps^2)$, the total number of gates~$t'$ is of order $qt (\log q + \log(1/\eps))$.
\end{proof}

Note that the total variation distance being equal to $0.01$ is not significant --- the point is that this modified strategy only affects the success probability of the learning procedure by a small constant. For example, consider the task of quantum state tomography with adaptive single-copy measurements. Let $\bm{x}$ be distributed over $\mathsf{D}(d)$ and $\bm{y}=\bm{y}_1,\dots,\bm{y}_n$ be the measurement outcomes obtained using the original strategy on $n\in O(d^3/\eps^2)$ copies of $\bm{x}$. Let $\bm{y}^\prime=\bm{y}^\prime_1,\dots,\bm{y}^\prime_n$ be the outcomes obtained using the modified strategy on $n$ copies of $\bm{x}$. Prop.~\ref{prop:efficient_gates} says that for any value $x$ of the random variable $\bm{x}$,
\begin{align*}
\norm{p_{\bm{y}|x}-p_{\bm{y}^\prime|x}}_1\leq 1/100 \enspace.
\end{align*}
Therefore, if the original strategy succeeds in identifying the state to within accuracy $\eps$ with ``high'' probability (say, $\ge 2/3$) then the modified strategy succeeds with high probability as well. However, the modified strategy uses measurements drawn from a finite set of measurements, which is the setting for which our lower bounds apply. By counting the number of distinct circuits of size~$t'$ we see that the total number of distinct measurements~$m$ used in the modified strategy is at most~$\text{poly}(t')^{t'}$. So~$\log m$ is of the order of
\begin{align*}
\text{poly}(q) (\log q + \log(1/\eps)) \enspace.
\end{align*}
By Theorem~\ref{thm:tomography_efficient_meas} and Prop.~\ref{prop:efficient_gates} we have that
\begin{align}
    \Omega\left(\frac{d^3}{\eps^2 (1+\textnormal{polylog}(d) \, u(d,1/\eps) /d)}\right)
\end{align}
samples are required, where
\[
u(d,1/\eps) \eqdef (\log \log d + \log(1/\eps)) \log(\log \log d + \log(1/\eps)) \enspace.
\]
Similarly, for the single-copy shadow tomography problem (cf. Def.~\ref{def-shadow-tomography}), by Theorem~\ref{thm:unentangled_shadow_tomography}
\begin{align}
    \Omega\left(\frac{d\min\{\log(M),\  d^2\}}{\eps^2 (1+\textnormal{polylog}(d) \, u(d,1/\eps) /d)}\right)
\end{align}
samples are required, even when the measurements are implemented efficiently using a constant arity gate set of possibly infinite cardinality. Note that these bounds are asymptotically smaller than those for finite gate sets only when the approximation parameter~$\eps$ is exponentially small in the dimension~$d$. 

\end{document}